\documentclass[11pt,,a4paper]{article}
\usepackage{amsmath}\usepackage{epsf,amsfonts,amsthm}\usepackage{amscd,amssymb}
\usepackage{xcolor,epic,eepic}\usepackage{epsfig}
\usepackage{fontenc,indentfirst, delarray,amsfonts,amsmath,amssymb}
\usepackage{rotating}
\usepackage{mathdots}
\usepackage[T1]{fontenc}
\usepackage[matrix,arrow,curve]{xy}
\usepackage{amsmath}
\usepackage{amssymb}
\usepackage{amsthm}
\usepackage{amscd}
\usepackage{amsfonts}
\usepackage{graphicx}
\usepackage{fancyhdr}
\usepackage{dsfont,texdraw}
\usepackage{amsmath}\usepackage{epsf,amsfonts,amsthm}\usepackage{amscd,amssymb}
\usepackage{xcolor,epic,eepic}\usepackage{epsfig}
\usepackage{fontenc,indentfirst, delarray,amsfonts,amsmath,amssymb}
\usepackage{rotating}
\usepackage{amscd}
\usepackage{amsmath}
\usepackage{amsthm}
\usepackage{mathrsfs}
\usepackage{latexsym}
\usepackage{amssymb}
\usepackage{amsfonts}
\theoremstyle{plain}

\usepackage{tikz}
\usetikzlibrary{arrows,chains,matrix,positioning,scopes,snakes}
\makeatletter
\tikzset{join/.code=\tikzset{after node path={%
\ifx\tikzchainprevious\pgfutil@empty\else(\tikzchainprevious)%
edge[every join]#1(\tikzchaincurrent)\fi}}}
\makeatother

\tikzset{>=stealth',every on chain/.append style={join},
         every join/.style={->}}
\tikzset{
    >=stealth',
    punkt/.style={
           rectangle,
           rounded corners,
           draw=black, very thick,
           text width=6.5em,
           minimum height=2em,
           text centered},
    pil/.style={
           ->,
           thick,
           shorten <=2pt,
           shorten >=2pt,}
}

\usepackage[T1]{fontenc}
\newcommand{\bee}{\begin{enumerate}}
\newcommand{\eee}{\end{enumerate}}
\newcommand{\benn}{\begin{equation*}}
\newcommand{\eenn}{\end{equation*}}
\newcommand{\be}{\begin{equation}}
\newcommand{\ee}{\end{equation}}
\newcommand{\bean}{\begin{eqnarray}}
\newcommand{\eean}{\end{eqnarray}}
\newcommand{\bea}{\begin{eqnarray*}}
\newcommand{\eea}{\end{eqnarray*}}
\newcommand{\w}{\wedge}
\newcommand{\p}{\partial}

\newcommand{\Ci}{C^{\infty}}
\newcommand{\N}{\mathbb{N}}

\newcommand{\R}{\mathbb{R}}

\newcommand{\Q}{\mathbb{Q}}

\newcommand{\lp}{\left(}
\newcommand{\rp}{\right)}
\newcommand{\op}[1]{\!\!\mathop{\rm ~#1}\nolimits}

\newcommand{\mbi}{\mathbb{I}}
\newcommand{\dd}{\op{d}}
\newcommand{\cF}{{\cal F}}

\newcommand{\Diff}{\operatorname{Diff}}

\mathchardef\za="710B  
\mathchardef\zb="710C  
\mathchardef\zg="710D  
\mathchardef\zd="710E  
\mathchardef\zve="710F 
\mathchardef\zz="7110  
\mathchardef\zh="7111  
\mathchardef\zy="7112 

\mathchardef\zi="7113  
\mathchardef\zk="7114  
\mathchardef\zl="7115  
\mathchardef\zm="7116  
\mathchardef\zn="7117  
\mathchardef\zx="7118  
\mathchardef\zp="7119  
\mathchardef\zr="711A  
\mathchardef\zs="711B  
\mathchardef\zt="711C  
\mathchardef\zu="711D  
\mathchardef\zf="711E 
\mathchardef\zq="711F  
\mathchardef\zc="7120  
\mathchardef\zw="7121  
\mathchardef\ze="7122  
\mathchardef\zvy="7123  
\mathchardef\zvw="7124  
\mathchardef\zvr="7125 
\mathchardef\zvs="7126 
\mathchardef\zvf="7127  
\mathchardef\zG="7000  
\mathchardef\zD="7001  
\mathchardef\zY="7002  
\mathchardef\zL="7003  
\mathchardef\zX="7004  
\mathchardef\zP="7005  
\mathchardef\zS="7006  
\mathchardef\zU="7007  
\mathchardef\zF="7008  
\mathchardef\zW="700A  

\newcommand{\cyclic}{\mathop{\kern0.9ex{{+}
\kern-2.15ex\raise-.25ex\hbox{\Large\hbox{$\circlearrowright$}}}}\limits}

 \newcommand{\cS}{{\cal S}}

 \newcommand{\cP}{{\cal P}}
 \newcommand{\cC}{{\cal C}}
 \newcommand{\cA}{{\cal A}}
 
 \newcommand{\cD}{{\cal D}}
 \newcommand{\cO}{{\cal O}}
 \newcommand{\cB}{{\cal B}}
 
 \newcommand{\cQ}{{\cal Q}}
 \newcommand{\cI}{{\cal I}}

 \newcommand{\cR}{{\cal R}}

\newtheorem{rem}{Remark}
\newtheorem{theo}[rem]{Theorem}
\newtheorem{prop}[rem]{Proposition}
\newtheorem{lem}[rem]{Lemma}
\newtheorem{cor}[rem]{Corollary}
\newtheorem{ex}[rem]{Example}

\newtheorem{defi}[rem]{Definition}

\newcommand{\ul}{\underline}
\newcommand{\h}{\op{Hom}}
\newcommand{\0}{\otimes}
\newcommand{\id}{\op{id}}

\DeclareMathAlphabet{\mathpzc}{OT1}{pzc}{m}{it}

\newcommand{\Hom}{\mathrm{Hom}}

 \newcommand{\cJ}{\mathcal{J}}

\begin{document}
\title{\bf On Koszul-Tate resolutions and Sullivan models}
\date{}
\author{Damjan Pi\v{s}talo and Norbert Poncin}
\maketitle

\begin{abstract} We report on Koszul-Tate resolutions in Algebra, in Mathematical Physics, in Cohomological Analysis of {\small PDE}-s, and in Homotopy Theory. Further, we define an abstract Koszul-Tate resolution in the frame of $\cD$-Geometry, i.e., geometry over differential operators. We prove Comparison Theorems for these resolutions, thus providing a dictionary between the different fields. Eventually, we show that all these resolutions are of the new $\cD$-geometric type.

\end{abstract}

\vspace{2mm} \noindent {\bf MSC 2010}: 18G55, 16E45, 35A27, 32C38, 16S32\medskip

\noindent{\bf Keywords}: Koszul-Tate resolution, relative Sullivan algebra, gauge theory, partial differential equation, jet bundle, compatibility complex, model category, homotopy theory, $\cD$-module
\thispagestyle{empty}

\tableofcontents

\section{Introduction}

The present paper arose from our interest in a coordinate-free approach to the moduli space of solutions of a system of partial differential equations ({\small PDE}-s) modulo symmetries and in particular to the Batalin-Vilkovisky formalism for gauge theories. Vinogradov's Cohomological Analysis of {\small PDE}-s \cite{Sascha} is a landmark in this field. It interprets the solution space as a smooth manifold inside an infinite jet space. Beilinson and Drinfeld \cite{BD} view this solution space as a $\cD$-scheme, where $\cD$ denotes the ring of linear differential operators of a smooth scheme $X$.\medskip 

We are convinced that the best framework for the Batalin-Vilkovisky formalism is \linebreak\ul{homotopical algebraic $\cD$-geometry} \cite{BPP3}, a combination of $\cD$-Geometry and Homotopical Algebraic Geometry \cite{TV08}. This idea leads to derived $\cD$-stacks, i.e., sheaves from the category $\tt DG\cD A$ of differential graded commutative algebras over $\cD$ to the category of simplicial sets. The definition of the sheaf condition uses an appropriate model structure on $\tt DG\cD A$ \cite{BPP4}. The corresponding cofibration-trivial-fibration factorization and cofibrant replacement functor provide a \underline{minimal relative Sullivan $\cD$-algebra}, which turned out to be a good candidate for the \underline{Koszul-Tate resolution} -- the first step of the Batalin-Vilkovisky construction.\medskip

In this paper, \underline{we report on a series of Koszul-Tate resolutions}: on the Koszul resolution \cite{CE} {\it of a regular sequence}, the Koszul-Tate resolution \cite{Tate} {\it of a quotient ring}, the Koszul-Tate resolution {\it in Gauge Field Theory} \cite{HT} -- in particular on the Koszul-Tate resolution {\it in a regular first-order on-shell reducible gauge theory} \cite{BBH} and the Koszul-Tate resolution {\it in Cohomological Analysis of {\small PDE}-s} \cite{Ver}. Each one of these resolutions builds on the chro\-no\-lo\-gi\-cal\-ly preceding ones.\medskip

Next \ul{we comment on relative Sullivan algebras}. In \cite{Qui}, Quillen introduces model categories as a suitable framework for resolutions. A standard model structure on the category $\tt DG\Q A$ of differential graded commutative algebras over the field $\Q$ of rational numbers, as well as the small object argument, lead, for any morphism, to a relative Sullivan $\Q$-algebra that models this morphism \cite{Halperin} -- the relative Sullivan minimal model of the morphism. Although this relative Sullivan algebra is quasi-isomorphic to the target of the considered morphism and thus resolves this target differential graded algebra if it is concentrated in degree zero, and although it has the same structure as the Koszul-Tate resolution of a quotient ring, \ul{relative Sullivan minimal models and Koszul-Tate resolutions appeared independently} in their respective fields. We observed (see above) this similarity in structure after having defined a model structure on $\tt DG\cD A$, in order to deal with the Batalin-Vilkovisky formalism. Since this projective model structure exists only if the underlying smooth scheme $X$ (see above) is a smooth affine variety, our candidate-Koszul-Tate-resolution (see above), which we finally called the {\it cofibrant replacement} Koszul-Tate resolution \cite{BPP4}, also exists only in this case. We noticed that its main structure can nevertheless be extended to the general situation of an arbitrary smooth scheme, what leads then to a general abstract Koszul-Tate resolution, which we describe in this paper and which does {\it always} exist. Since we observed later on that an equivalent structure is used in \cite{BD} under the name of semi-free differential graded $\cD$-algebra, we refer to the latter resolution as the {\it $\cD$-geometric} Koszul-Tate resolution.\medskip

Beyond the surveys that we described in the two preceding paragraphs, we show in the present paper that \ul{all the Koszul-Tate resolutions that we reviewed are $\cD$-geometric} Koszul-Tate resolutions. This result provides additional evidence for our afore-mentioned conviction that homotopical $\cD$-geometry is the appropriate setting for the Batalin-Vilkovisky formalism.\medskip

The comparisons of the various Koszul-Tate resolutions are a major difficulty in view of the distinct languages used. Since we believe that such passages -- between the Algebra and Physics worlds \cite{Tate}, \cite{HT}, \cite{Bar}, the world of partial differential equations \cite{Sascha}, \cite{Ver}, the world of Homotopy Theory \cite{Qui}, \cite{Sul}, \cite{Halperin}, and the world of $\cD$-Geometry \cite{BD} (\cite{BPP3}, \cite{BPP4}) -- are lacking in the literature, we give precise comparison results for the Koszul-Tate resolutions that we considered, thus providing a kind of \ul{dictionary between different fields}.\medskip
 
We assume that most readers are familiar with homotopy and model categories (if not, a concise introduction can be found in the appendices of \cite{BPP1} and \cite{BPP2}), but we give a short introduction to regular first-order on-shell reducible field theories and provide in the appendix a smallest possible introduction to the jet bundle formalism in Field Theory and Cohomological Analysis.

\section{Koszul-Tate resolutions in Algebra and in Physics}\label{CKTR}

\subsection{Koszul resolution of a regular surface}\label{KRS}

Let $\zS$ be an embedded $p$-dimensional submanifold of $\R^n$. This means that, for each $x\in\zS$, there is an open neighborhood $\zW\subset\R^n$ such that $\zS\cap\zW$ is described by a {\bf regular} cartesian equation $E\in\Ci(\zW,\R^{n-p})$. By `regular' we mean that the equations $E_{\frak a}\in\Ci(\zW,\R)$ are independent, i.e., that the rank $\zr(\p_xE)$ is equal to $n-p$, for all $x\in \zS\cap \zW$. Assume for simplicity that the first $n-p$ columns of the Jacobian matrix are independent and use the decomposition $x=(x',x'')\in\R^{n-p}\times\R^p$. Then, locally, in the neighborhood of $\zS$, we have $E=E(x',x'')\Leftrightarrow x'=x'(E,x'')$. It follows that, locally, in the new coordinates $(E,x'')$, the equation of $\zS$ is $E=0$.\medskip

Adopt now the standpoint of Mathematical Physics and consider a submanifold $\zS\subset\R^n$ that is {\it globally} described by the equations $E_{\frak a}=0$, for all ${\frak a}$.\medskip

One of the fundamental consequences of regularity is the structure of the ideal $I(\zS)$ made of those smooth functions $\Ci(\R^n)$ that vanish on $\zS$. It is clear that any linear combination $F=\sum_{\frak a} F^{\frak a} E_{\frak a}$, $F^{\frak a}\in\Ci(\R^n)$, of the equations belongs to $I(\zS)$. Conversely, if $F\in I(\zS)$, we get, working in the coordinates $(E,x'')$, $$F(E,x'')=\int_0^1d_t\left(F(tE,x'')\right)\dd t=\sum_{\frak a} E_{\frak a}\int_0^1\left(\p_{E_{\frak a}}F\right)(tE,x'')\dd t\;=:\sum_{\frak a} F^{\frak a} E_{\frak a}\;.$$

We are now prepared to recall the construction of the Koszul resolution of the function algebra $\Ci(\zS)$ of
\be\label{E1}\zS:E_{\frak a}=0,\;\forall {\frak a}\in\{1,\ldots,n-p\}\;,\ee where the $E_{\frak a}$ are the first coordinates of an appropriate coordinate system $(E,x'')$ of $\R^n$.

\begin{defi} The {\bf Koszul complex of the regular surface} (\ref{E1}) is the chain complex made of the free Grassmann algebra $$\op{K}=\Ci(\R^n)\otimes \cS[\zf^{{\frak a}*}]$$ on $n-p$ odd generators $\zf^{{\frak a}*}$ -- associated to the equations (\ref{E1}) -- and of the Koszul differential \be\label{KDiff}\zd_{\op{K}}=\sum_{\frak a}E_{\frak a}\,\p_{\zf^{{\frak a}*}}\;.\ee\end{defi}

\begin{rem}{\em Notice that the base ring for the tensor products $\0$ and $\cS$ has not been specified and that these products are merely a sensible notation for graded commutative polynomials in the generators with coefficients in $\Ci(\R^n)$.}\end{rem}

\begin{prop}\label{MathPhysKR}The Koszul complex of $\zS$ is a resolution of $\Ci(\zS)$, i.e., the homology of $(\op{K},\zd_{\op{K}})$ is given by \be\label{Homology}H_0(\op{K})=\Ci(\zS)\quad\text{and}\quad H_k(\op{K})=0,\;\forall k>0\;.\ee We refer to this resolution as the {\bf Koszul resolution} of $\Ci(\zS)$.\end{prop}

Indeed, in degree 0, the cycles are the functions in $\Ci(\R^n)$ and the boundaries are the elements of $$\zd_{\op{K}}\{\sum_{\frak b}F^{\frak b}\,\zf^{{\frak b}*}\}=\{\sum_{\frak a}F^{\frak a}\,E_{\frak a}\}=I(\zS)\;,$$ so that $H_0(\op{K})=\Ci(\zS)\,.$ The proof that the higher homology spaces vanish is technical and not really instructive. It is based on the fact that the operator $h=\sum_{\frak{a}}\zf^{{\frak{a}}*}\p_{E_{\frak{a}}}$ is a homotopy between the Euler vector field or number operator $E_{\frak a}\p_{E_{\frak a}}+\zf^{{\frak a}*}\p_{\zf^{{\frak a}*}}$ and the zero chain map, so that any chain $c$ reads $$c(E,x'',\zf^*)=c(0,x'',0)+\int_0^1\frac{\op{d}\zl}{\zl}((\zd_{\op{K}}h+h\zd_{\op{K}})c)(\zl E,x'',\zl\zf^*)\;.$$

\subsection{Koszul resolution of a regular sequence}\label{KRRS}

Let $R$ be a commutative unital ring, $M$ a module over $R$ of finite rank $r$, and let $d\in M^*:=\h_R(M,R)$ be an $R$-linear map.

\begin{defi} The {\bf Koszul complex of the covector $d$} is the graded $R$-module $\bigwedge_R M$ endowed with the differential $\zd_{\op{K}}$ given by the extension of $d$ as a degree $-1$ derivation. We denote the Koszul chain complex of $d$ by $K[d\,]$.\end{defi}

More precisely, the Koszul differential is given by $$\zd_{\op{K}}(m_1\w\ldots\w m_k)=\sum_{\ell=1}^k(-1)^{\ell-1}(dm_\ell)\;m_1\w\ldots\widehat{\ell}\ldots\w m_k\;.$$

Assume now that $M=R^r$ is free with basis $(e_{\frak a})_{\frak a}$. In this case, the linear map $d$ is given by $d=(E_1,\ldots,E_r)$ ($E_{\frak a}\in R$). It is well-known that the Koszul complex $K[E_1,\ldots,E_r]$ coincides with the tensor product $K[E_1\,]\0\ldots\0 K[E_r\,]$ of the Koszul complexes $K[E_{\frak a}]$, where $E_{\frak a}$ is viewed as an $R$-linear map from $R e_{\frak a}$ to $R$. Indeed, in both cases the underlying $R$-module is $\bigoplus_{k=0}^rR^{\,\complement^k_r}$ ($\,\complement^k_r$ is the binomial coefficient) and the differential is defined by $$\zd_{\op{K}}(e_{{\frak a}_1}\w\ldots\w e_{{\frak a}_k})=\sum_{\ell=1}^k(-1)^{\ell-1} E_{{\frak a}_\ell}\, e_{{\frak a}_1}\w\ldots\widehat{\ell}\ldots\w e_{{\frak a}_k}\;.$$ The degree 0 Koszul homology module $H_0(K[E_1,\ldots, E_r])$ is the quotient of the kernel $R$ by the image $\{\zd_{\op{K}}(\zr)=\sum_{\frak a} \zr^{\frak a}E_{\frak a}, \zr\in R^r\}$, i.e., the quotient \be\label{OnShFun}R/(E_1,\ldots,E_r)\ee of the ring $R$ by its ideal generated by the $E_{\frak a}\,$.\medskip

The considered Koszul complex can of course be written $$K[E_1,\ldots,E_r]=R\0\cS[e_1,\ldots,e_r]\quad\quad\text{and}\quad\quad \zd_{\op{K}}=\sum_{\frak a}E_{\frak a}\,\p_{e_{\frak a}}\;,$$ provided we view the $e_{\frak a}$ as degree 1 generators: the definitions of the present subsection and the just mentioned homology result coincide with those of the preceding subsection.\medskip

As easily checked, the degree 1 homology module is (at least if $R$ is a $\mathbb Q$-algebra) the quotient of the cycles $\{\zr\in R^r:\sum_{\frak a}\zr^{\frak a}E_{\frak a}=0\}$ by the trivial cycles $$\{\zr=\Theta\,(E_1,\ldots, E_r)^{\tilde{}}:\Theta\in \op{Sk}(r,R)\}\;,$$ where `tilde' is the transpose and where $\op{Sk}(r,R)$ denotes the skew-symmetric $r\times r$ matrices with entries in $R$.\medskip

In the language of the preceding subsection this means that $H_1(K[E_1,\ldots, E_r])$ is given by the linear relations between the equations modulo the trivial relations.\medskip

If all the `relations' are trivial, as well as all the `higher relations' pertaining to the higher homology modules, the Koszul complex is a resolution of the quotient (\ref{OnShFun}) of `on-shell functions'.

\begin{defi} A sequence $(E_1,\ldots,E_r)$ of elements of $R$ is called {\bf regular}, if $E_{\frak a}$ is not a zero divisor of $R/(E_1,\ldots,E_{{\frak a}-1})$, for all ${\frak a}\in\{1,\ldots,r\}$, and if $R/(E_1,\ldots,E_r)\neq 0$.\end{defi}

To illustrate the definition, we consider the case $r=2$. The existence of a relation $$\zr^1E_1+\zr^2E_2=0$$ is equivalent to $E_2[\zr^2]=0$, with $[\zr^2]\in R/(E_1)$. It follows from the regularity assumption that $\zr^2=\zr^3E_1$, so that $E_1(\zr^1+\zr^3E_2)=0$. Applying again the regularity, we obtain $\zr^1=-\zr^3E_2$ and, finally, $$\left(
                                                                                     \begin{array}{c}
                                                                                       \zr^1 \\
                                                                                       \zr^2 \\
                                                                                     \end{array}
                                                                                   \right)=
\left(
  \begin{array}{cc}
    0 & -\zr^3 \\
    \zr^3 & 0 \\
  \end{array}
\right)
\left(
  \begin{array}{c}
    E_1 \\
    E_2 \\
  \end{array}
\right)\;,
$$
so that the linear combination $\zr^1E_1+\zr^2E_2=-\zr^3E_2E_1+\zr^3E_1E_2$ vanishes trivially. This fact that regularity implies that all relations are trivial, can be extended, first to higher $r$, and second to higher relations. Actually we have the following (well-known) mathematical variant of Proposition \ref{MathPhysKR}:

\begin{prop} If the sequence $(E_1,\ldots, E_r)$ of elements of $R$ is regular, the Koszul complex $K[E_1,\ldots,E_r]$ resolves the quotient ($\ref{OnShFun}$). More generally, if the ring $R$ is local and the $R$-module $M$ is a finitely generated, a sequence $(E_1,\ldots,E_r)$ of elements of $R$ is $M$-regular if and only if the Koszul complex $K[E_1,\ldots,E_r]$ resolves the quotient $M/(E_1,\ldots,E_r)M\,.$ \end{prop}

\begin{rem}{\em\label{Triv} It follows that, if the sequence $(E_1,\ldots, E_r)$ of elements of $R$ is regular, then any relation $\sum_{\frak a}\zr^{\frak a}E_{\frak a}=0$ is {\bf trivial}, in the sense that $\zr=\Theta\,(E_1,\ldots, E_r)^{\tilde{}}$, with $\Theta\in \op{Sk}(r,R)$.}\end{rem}

\subsection{Koszul-Tate resolution of a quotient ring}

In \cite{Tate}, J. Tate starts from a Noetherian commutative unital ring $R$, defines the category $\tt DGRA$ of differential graded commutative unital $R$-algebras $A$ as usual except that the $R$-module $A_0$ in degree 0 is assumed to be just $R\cdot 1_A$, he calls such a differential graded $R$-algebra $A$ free, if there exist homogeneous generators $(e_1,e_2,\ldots)$ such that $A=R\0\cS[e_1,e_2,\ldots]$ and each $R$-module $A_n$ of degree $n>0$ contains only a finite number of $e_i$, and, finally, says that $A\in\tt DGRA$ is acyclic if $H_n(A)=0$, for all $n>0$ (for most other authors acyclic means that one has in addition $H_0(A)=0$).\medskip

The paper contains two main theorems.

\begin{theo}\label{TateTheo1} For any ideal $I\subset R$ of any Noetherian commutative unital ring $R$, there exists a free resolution of $R/I$ in $\tt DGRA$.\end{theo}

\noindent{\it Sketch of Proof.} Note first that, for a commutative ring, Noetherian, i.e, the property that any ascending chain of ideals stabilizes, is equivalent to the property that any ideal is finitely generated. Let now $(E_1,\ldots,E_r)$ be the generators of the ideal $I$, set \be\label{KT0}X^0=R\0\cS[e_1,\ldots,e_r]\;,\ee with all generators $e_{\frak a}$ in degree 1, and define the differential $d^0$ on $X^0$ by \be\label{KT1}d^0=\sum_{\frak a}E_{\frak a}\,\p_{e_{\frak a}}\;.\ee The homology module $H_0(X^0)$ is the module $R/I$. The complex $(X^0,d^0)$ is clearly the {\it Koszul complex} $(K[E_1,\ldots,E_r],\zd_{\op{K}})$ of the sequence $(E_1,\ldots,$ $E_r)$ -- see Subsection \ref{KRRS}. However, since this sequence is here not assumed to be regular, the higher homology modules do not necessarily vanish. If the module $H_1(X^0)$ does not vanish, we choose $F_1,\ldots,F_s\in\ker_1 d^0$ such that the homology classes $[F_{\frak b}]^0$ generate $H_1(X^0)$, and we set \be\label{KT2}X^1=R\0\cS[e_1,\ldots,e_r,f_1,\ldots,f_s]\;,\ee with all generators $f_{\frak b}$ in degree 2 and $d^1$ defined by \be\label{KT3}d^1=\sum_{\frak a}E_{\frak a}\,\p_{e_{\frak a}}+\sum_{\frak b}F_{\frak b}\,\p_{f_{\frak b}}\;.\ee Of course, $H_0(X^1)=H_0(X^0)=R/I$. As for $H_1(X^1)$, note that $\ker_1d^1=\ker_1d^0$ and that $\op{im}_1d^1$ (resp., $\op{im}_1d^0$) is made of the linear combinations of the type $$\sum_{\frak b}r^{\frak b}F_{\frak b}+\sum_{{\frak a}'{\frak a}''}r^{{\frak a}'{\frak a}''}(E_{{\frak a}'}e_{{\frak a}''}-E_{{\frak a}''}e_{{\frak a}'})\quad\quad\left(\text{resp., }\sum_{{\frak a}'{\frak a}''}r^{{\frak a}'{\frak a}''}(E_{{\frak a}'}e_{{\frak a}''}-E_{{\frak a}''}e_{{\frak a}'})\right)\;.$$ Let now $[c]^1\in H_1(X^1)$. Since $[c]^0\in H_1(X^0)$, we have $[c]^0=[\sum_{\frak b} r^{\frak b} F_{\frak b}]^0$, so that $$c=\sum_{\frak b} r^{\frak b} F_{\frak b}+\sum_{{\frak a}'{\frak a}''}r^{{\frak a}'{\frak a}''}(E_{{\frak a}'}e_{{\frak a}''}-E_{{\frak a}''}e_{{\frak a}'})\;$$ and $[c]^1=0$. It suffices to iterate the procedure and to construct $X^k$, such that $H_0(X^k)=R/I$ and $H_p(X^k)=0$, for all $1\le p\le k$. Then, the inductive limit $X$ of the direct system of free differential graded $R$-algebras $X^k$ is the resolving free differential graded $R$-algebra.\hspace{1cm}    $\square$\medskip

\begin{rem}{\em We refer to Tate's extension $X$ of the Koszul complex $X^0$ associated to $I=(E_1,\ldots,E_r)$ as the Koszul-Tate resolution of $R/(E_1,\ldots,E_r)$. Tate's method allows to find a resolution, even if the sequence $(E_1,\ldots,E_r)$ is not regular, i.e., if not all `relations' and `higher relations' are trivial. The procedure starts from the Koszul complex (\ref{KT0})-(\ref{KT1}), whose chain module is constructed from generators $e_{\frak a}$ associated to the equations $E_{\frak a}$ (see (\ref{KT0})) and whose differential is the corresponding characteristic differential (\ref{KT1}). Then, one associates additional generators $f_{\frak b}$ to the non-trivial 1-cycles $F_{\frak b}=\sum_{\frak a}F^{\frak a}_{\frak b}e_{\frak a}$ or non-trivial relations $d^1F_{\frak b}=\sum_{\frak a}F^{\frak a}_{\frak b}E_{\frak a}=0$ (see (\ref{KT2})) and extends the differential accordingly (see (\ref{KT3})), in order to kill the non-trivial 1-cycles or relations. The procedure is now iterated, i.e., still new generators are added and new similar extensions of the differential are considered to kill the higher non-trivial relations. The Noetherian hypothesis allows to obtain finitely generated terms $X_p$, $p\ge 0$.}\end{rem}

The second theorem of \cite{Tate} is valid without the Noetherian property:

\begin{theo}\label{TateTheo4} Let $I\subset R$ be an ideal of a commutative unital ring $R$. Assume that there exist a commutative unital ring $S$, as well as ideals $P\subset J\subset S$, which are generated by regular $S$-sequences $(P_1,\ldots,P_s)$ and $(J_1,\ldots,J_r)$, respectively, and which are such that $S/P=R$ and $J/P=I$. Denote the classes in these quotients by $\bar{\bullet}$, set $E_{\frak a}=\bar{J_{\frak a}}$, and set $P_{\frak b}=\sum_{\frak a}s_{\frak b}^{\frak a}J_{\frak a}$. Then the differential graded $R$-algebra \be\label{KT4}Y=R\0\cS[e_1,\ldots,e_r,f_1,\ldots,f_s]\;,\ee with all generators $e_{\frak a}$ (resp., $f_{\frak b}$) in degree 1 (resp., 2) and with differential $d$ defined by \be\label{KT5}d=\sum_{\frak a}E_{\frak a}\,\p_{e_{\frak a}}+\sum_{\frak b}(\sum_{\frak a}\bar s^{\frak a}_{\frak b}e_{\frak a})\,\p_{f_{\frak b}}\;,\ee is a free resolution of $R/I$.\end{theo}

The proof \cite{Tate} of this result is technical and will not be sketched. On the other hand, the following observation is worth being emphasized.\medskip

Up to the end of this subsection, we use the language of Mathematical Physics, we interpret the generators $E_{\frak a}$ of $I$ as the equations of a shell in some ambient space, the ring $R$ as the functions of this space, and the ideal $I=(E_1,\ldots,E_r)$ as the functions that vanish on-shell. Moreover, a new concept of triviality will appear. Until now, a relation between the equations $E_{\frak a}$ was considered as trivial, if the column of its coefficients could be obtained by applying a skew-symmetric matrix to the column made of the $E_{\frak a}$ (Remark \ref{Triv}). We will refer to a relation between the equations as {\bf weakly trivial}, if all its coefficients vanish on-shell. It is clear that trivial implies weakly trivial: Theorem 10.1 in \cite{HT} shows that in the context of Physics a weakly trivial relation is always trivial.

\begin{rem}{\em\label{TateHT} The assumptions of Theorem \ref{TateTheo4} imply that $Y$ is a resolution of $R/I$, what in turn entails that there exist relations between the equations, which are not trivial and thus not weakly trivial, i.e., at least one of their coefficients does not vanish on-shell. Further, these relations are independent, in the sense that, if a linear combination between them vanishes on-shell, then all the coefficients vanish on-shell.}\end{rem}

The interest of this remark resides in the fact that, in the Mathematical Physics' literature, this context -- existence of non-weakly-trivial relations between the equations and only weakly trivial relations between these relations -- does not result as here from a Koszul-Tate resolution, which was constructed under certain assumptions, but this setting is essentially the starting point in the Physicists' attempt to build a Koszul-Tate resolution of shell functions -- which is itself the first step in the construction of the Batalin-Vilkovisky ({\small BV}) resolution (interesting ideas and a survey on {\small BV} can be found in \cite{Jim Stasheff}) of the functions of the shell modulo gauge symmetries (see Appendix A, Section \ref{Vino}).\medskip

Let us explain Remark \ref{TateHT}.\medskip

Since, for any $\frak b$, the class $\bar\bullet$ of the generator $P_{\frak b}=\sum_{\frak a}s_{\frak b}^{\frak a}J_{\frak a}\in P$ vanishes, we have between the equations the relation \be\label{1Cyc}\sum_{\frak a}\bar s_{\frak b}^{\frak a}E_{\frak a}=0\;.\ee The kernel $\ker_1d$ is made of the 1-chains $c_1=\sum_{\frak a}r^{\frak a}e_{\frak a}$ that induce a relation $\sum_{\frak a}r^{\frak a}E_{\frak a}=0$, and, as easily checked, the image $\op{im}_1d$ is made of the boundaries $dc_2$ of the 2-chains $c_2$, which are (at least if $R$ is a $\Q$-algebra) of the type \be\label{2Bound}dc_2=\sum_{\frak a}\sum_{\frak c}\zr^{{\frak a}{\frak c}}E_{\frak c}\,e_{\frak a}+\sum_{\frak a}\sum_{\frak d}\zr^{\frak d}\bar s^{\frak a}_{\frak d}\,e_{\frak a}\;,\ee where $(\zr^{{\frak a}{\frak c}})$ is a shew-symmetric $r\times r$ matrix and where $(\zr^{\frak d})$ is an $s\times 1$ matrix, both with entries in $R$. Remember that the new generators $f_{\frak d}$ have been added to make homologically non-trivial 1-cycles trivial. The 1-cycle $\sum_{\frak a}\bar s^{\frak a}_{\frak b}e_{\frak a}$ (see (\ref{1Cyc})) is visibly homologically trivial due to the adjunction of the generators $f_{\frak d}$, which are responsible for the second term in the {\small RHS} of Equation (\ref{2Bound}). In other words, it was not homologically trivial before the addition of these new generators, i.e., the $\bar s^{\frak a}_{\frak b}$ are not of the form $\sum_{\frak c}\zr^{{\frak a}{\frak c}}E_{\frak c}$, with $(\zr^{{\frak a}{\frak c}})$ skew-symmetric, or, still, the relation $\sum_{\frak a}\bar s^{\frak a}_{\frak b}E_{\frak a}=0$ is not trivial and so not weakly trivial (a (longer) direct proof of this fact can be given).\medskip

\begin{rem}{\em Similarly, still other generators have to be added, if not all relations between the just considered relations are homologically trivial. Since {\it no} additional generators were added, all relations between the relations (\ref{1Cyc}) {\it are} homologically trivial, i.e., more precisely, for any relation $\sum_{\frak b} \zs^{\frak b}\bar{s}^{\frak a}_{\frak b}=0$, ${\frak a}\in\{1,\ldots,r\}$, the corresponding 2-cycle $\sum_{\frak b}\zs^{\frak b}f_{\frak b}$ is homologically trivial.}\end{rem}

Concerning the statement that there exist only weakly trivial `relations' between the relations (\ref{1Cyc}), we now assume that $\sum_{\frak b} \zs^{\frak b}\bar{s}^{\frak a}_{\frak b}\in I$, for all ${\frak a}$, and prove that $\zs^{\frak b}\in I$, for all ${\frak b}$. Since $$\sum_{\frak b} \zs^{\frak b}\bar{s}^{\frak a}_{\frak b}=\sum_{\frak c}z^{{\frak c}{\frak a}}E_{\frak c}\;,$$ we have the weakly trivial relation $$\sum_{\frak a}\sum_{\frak c}z^{{\frak c}{\frak a}}E_{\frak c}E_{\frak a}=\sum_{\frak b} \zs^{\frak b}\sum_{\frak a}\bar{s}^{\frak a}_{\frak b}E_{\frak a}=0\;,$$ which is, in view of \cite[Theorem 10.1]{HT}, trivial, what means that the matrix $(z^{{\frak c}{\frak a}})$ can be chosen shew-symmetric. Hence, the sum $c_2$ below is a 2-chain and
$$dc_2=d(\frac{1}{2}\sum_{\frak a}\sum_{\frak c}z^{{\frak a}{\frak c}}e_{\frak a}e_{\frak c}+\sum_{\frak b}\zs^{\frak b}f_{\frak b})=$$ $$\sum_{\frak a}\sum_{\frak c}z^{{\frak a}{\frak c}}E_{\frak c}\,e_{\frak a}+\sum_{\frak a}\sum_{\frak b}\zs^{\frak b}\bar s^{\frak a}_{\frak b}\,e_{\frak a}=-\sum_{\frak a}\sum_{\frak b} \zs^{\frak b}\bar{s}^{\frak a}_{\frak b}e_{\frak a}+\sum_{\frak a}\sum_{\frak b}\zs^{\frak b}\bar s^{\frak a}_{\frak b}\,e_{\frak a}=0\;.$$ As $Y$ is acyclic, the 2-cycle $c_2$ is the boundary of a 3-chain $c_3$ -- made of terms in $e_{\frak a}e_{\frak c}e_{\frak g}$ and of the term $\sum_{\frak a}\sum_{\frak b}r^{{\frak a}{\frak b}}e_{\frak a}f_{\frak b}$. The terms of the first type induce in the boundary only terms in $e_{\frak a}e_{\frak c}$ and the terms of the second type generate, in addition to terms in $e_{\frak a}e_{\frak c}$, the terms $$\sum_{\frak b}\sum_{\frak a}r^{{\frak a}{\frak b}}E_{\frak a}\,f_{\frak b}\;.$$ In view of freeness, we deduce that $\zs^{\frak b}=\sum_{\frak a}r^{{\frak a}{\frak b}}E_{\frak a}\in I$.

\begin{rem}{\em The differential in Theorem \ref{TateTheo4} is analogous to that of Theorem \ref{TateTheo1}: in Theorem \ref{TateTheo1} we dealt with the relations $\sum_{\frak a}F^{\frak a}_{\frak b}E_{\frak a}=0$, ${\frak b}\in\{1,\ldots,s\}$, and added the term $\sum_{\frak b}(\sum_{\frak a}F^{\frak a}_{\frak b}e_{\frak a})\,\p_{f_{\frak b}}$ to the differential, and in Theorem \ref{TateTheo4} the relations are $\sum_{\frak a}\bar s_{\frak b}^{\frak a}E_{\frak a}=0$, ${\frak b}\in\{1,\ldots,s\}$, and we added the term $\sum_{\frak b}(\sum_{\frak a}\bar s^{\frak a}_{\frak b}e_{\frak a})\,\p_{f_{\frak b}}$.}\end{rem}

\subsection{Koszul-Tate resolution of shell functions in a gauge theory}\label{RIGT}

\begin{rem}{\em In the following, we use standard concepts, results, and notation of the theory of {\small PDE}-s in the jet bundle formalism \cite{Sascha} (see Appendix A, Section \ref{Vino}).}\end{rem}

\subsubsection{Regular first-order on-shell reducible gauge theory}\label{RIGTDef}

In {\bf field theory}, fields are sections $\zf\in\zG(\zp)$ of a vector bundle $\zp:E\to X$. Since we will consider gauge theories from the standpoint of Physics, we work systematically in a trivialization of $E$ (fiber coordinates $u=(u^1,\ldots,u^r)$ -- we will sometimes write $u^a$ instead of $u$) over a coordinate patch of $X$ (coordinates $x=(x^1,\ldots,x^n)$ -- we may write $x^i$ instead of $x$), or, we just assume that $E=\R^n\times\R^r$. The dynamics of the considered field theory is given by a {\bf functional} $S$ acting on compactly supported sections $\zf\in\zG(\zp)$, $$ S[\phi]=\int_X {\cal L}(x^i,u^a_\za)|_{j^{k-1}\zf} \dd x\in\R\;,$$ where $j^{k-1}$ denotes the $(k-1)$-jet and where the Lagrangian ${\cal L}$ is a function ${\cal L}\in\cF(\zp_{k-1})$ of the $(k-1)$-jet bundle of $\zp$ (jet bundle coordinates $(x^i,u^a_\za)$) such that ${\cal L}(x^i,0)=0$ (it suffices to set $\widetilde F(x^i,u^a_\za):=F(x^i,u^a_\za)-F(x^i,0),$ for any function $F\in\cF=\cF(\zp_\infty)$ of the infinite jet space of $\zp$, to see that $\cF=\Ci(X)\oplus \widetilde\cF$, where the functions in $\widetilde \cF$ vanish on the zero section). Equivalently, we may use the corresponding {\bf Euler-Lagrange equations} \be\label{EL}\zd_{u^a}{\cal L}|_{j^{k}\zf}=(-D_x)^{\za}\p_{u^a_{\za}}{\cal L}|_{j^{k}\zf}=0\;,\ee where $\zd_{u^a}$ is the algebraicized Euler-Lagrange operator and $D_{x^i}$ the total derivative with respect to $x^i$ (see Appendix A, Section \ref{Vino}).\medskip

The extended algebraicized Euler-Lagrange equations \be\label{ELE}D_x^{\za}\zd_{u^a}{\cal L}=0\ee define the constraint surface or shell $\zS$ in the infinite jet space $J^\infty(\pi)$. The solutions $\phi$ of the original Euler-Lagrange equations (\ref{EL}) are those compactly supported sections $\zf\in\zG(\zp)$ that satisfy the condition $(j^\infty\zf)(X)\subset \zS\,$ (we mostly ignore local aspects). We denote by $I(\zS)\subset \cF$ the ideal of those functions in $\cF$ that vanish on-shell. If $f\in I(\zS)$, we write $f\approx 0\,$.\medskip

As for any system of linear equations, we may find linear relations between the considered equations (\ref{ELE}) (see `compatibility complex' in Appendix A), i.e., relations of the type
\be\label{NO2} N_\za^{a}D_x^{\za}\delta_{u^a}{\cal L}\equiv 0\;,\ee with $N_\za^{a}\in \cF$. It is easy to write such relations, if we use coefficients in $I(\zS)$. Indeed, for any functions $n^{[ab]}\in\cF$ (that are antisymmetric in $a,b$), we have the linear relation $n^{[ab]}\p_{u^b}{\cal L}\;\p_{u^a}{\cal L}\equiv 0$ between the equations $\p_{u^a}{\cal L}=0$. What we actually have in mind are non-trivial linear relations, i.e., relations of the type (\ref{NO2}), but with at least one coefficient $N_\za^{a}\notin I(\zS)$. We refer to such relations as {\bf non-trivial Noether identities}.\medskip

A deep result \cite{Noether, Yvette}, which is already present in elementary Mechanics, is the 1:1 correspondence between, roughly speaking, `symmetries of the action' (resp., `gauge symmetries') and conserved currents (resp., Noether identities). It motivates the definition of a {\bf gauge theory} as a field theory (see above) with non-trivial Noether identities. \medskip

The efficient investigation of gauge theories is subject to some regularity conditions that we now describe. More precisely, the {\bf regularity conditions for first-order reducible gauge theories} can be formulated as follows:\medskip

{\bf Assumption 1}. For any $\ell\in\N$, the {\small LHS}-s $D_x^{\za}\zd_{u^a}{\cal L}$ of the equations of $\zS$, up to order $k+\ell$ (i.e., since ${\cal L}\in\cF(\zp_{k-1})$, we consider derivatives $D_x^{\za}$ up to order $\ell$), can be separated into {\bf two packages} $E_{\frak a}$ and $E_{\zD}$ (of course, the ranges of $(\za,a)$ and of $({\frak a},\zD)$ are the same) (we could even only ask that the $D_x^{\za}\zd_{u^a}{\cal L}$ and the $(E_{\frak a},E_\zD)$ be related by an invertible matrix, i.e., that $$D_x^{\za}\zd_{u^a}{\cal L}=M^{\za\frak a}_{a}E_{\frak a}+M^{\za\zD}_{a}E_{\zD}\;,$$ where the matrix $M=(M^{\za\frak a}_{a},M_{a}^{\za\zD})$, with row index $(\za,a),$ is invertible; however, to simplify, we often ignore this matrix in the following, just as we ignore, as mentioned before, a number of local aspects).\medskip

{\bf Assumption 2}. The functions $E_{\frak a}\in\cF(\zp_{k+\ell})$ are independent. This is the actual {\bf regularity condition} (see Subsection \ref{KRS}). In other words, we assume that (locally -- but we ignore this restriction) the $E_{\frak a}=E_{\frak a}(x^i,u^a_{\za})$ can be chosen as the first fiber coordinates of a new coordinate system $(x^i,E_{\frak a},u''^a_{\za})$ in $J^{k+\ell}(\zp)$: $$(x^i,u'^a_{\za},u''^a_{\za})\leftrightarrow(x^i,E_{\frak a},u''^a_{\za})\;.$$

{\bf Assumption 3}. The functions $E_{\zD}$ are linear consequences of the functions $E_{\frak a}$: $E_{\zD}=F^{\frak a}_{\zD}E_{\frak a}$, with $F^{\frak a}_{\zD}\in\cF(\zp_{k+\ell})$. It follows that $E_{\zD}=0$, if $E_{\frak a}=0$: the $E_{\frak a}$ (resp., $E_{\zD}$) are the {\bf independent (resp., dependent) equations}.\medskip

{\bf Assumption 4}. The dependent equations $E_\zD$ are total derivatives of a finite number of dependent equations $E_\zd=F_\zd^{\frak b}E_{\frak b}$ ($\zd\in\{1,\ldots,K\}$), i.e., there is a finite number $K$ of {\bf generators $E_\zd$ by differentiation}: $E_{\zD}=D_x^{\zb}E_{\zd}$.\medskip

{\bf Assumption 5}. Note that the differences $E_\zD-F^{\frak a}_\zD E_{\frak a}\equiv 0$ are non-trivial Noether identities. We assume that, if $E_\zD=D_x^\zb E_\zd$, the derivative $D_x^{\zb}$ of the Noether identity $E_\zd - F^{\frak b}_\zd E_{\frak b}\equiv 0$ is the preceding Noether identity associated to $E_\zD\,.$ If we write this requirement out, we find an invertibility condition for some matrix, which is called the {\bf first-order reducibility assumption} ({\small IA}) of the considered gauge theory. \medskip

The assumptions 1-5 are satisfied in many physically relevant examples, in particular in the Klein-Gordon case and in electromagnetism.\medskip

Consider now a {\bf regular first-order reducible gauge theory}, i.e., a field theory, which admits non-trivial Noether identities (i.e., non-trivial gauge symmetries) and satisfies the assumptions 1-5.

\begin{prop}\label{IrrSetNOProp} In a regular first-order reducible gauge theory, there exists an irreducible set of non-trivial Noether operators.\end{prop}

Indeed, consider the Noether identities $E_\zd - F^{\frak b}_\zd E_{\frak b}\equiv 0$ and write them in the form \be\label{IrrSetNO}R^{a}_{\zd\za}\,D_x^\za\zd_{u^a}{\cal L}\equiv 0\quad(\zd\in\{1,\ldots,K\})\;.\ee Recall that the tuple of the $D_x^\za\zd_{u^a}{\cal L}$ is given by the action of an invertible matrix $M$ on the tuple made of the $E_{\frak a},E_\zD$. We often assume for simplicity that this matrix is identity. Even if we take this matrix into account, we see easily that the Noether identities (\ref{IrrSetNO}) are {\it non-trivial}.\medskip

A compatibility operator (roughly, non-trivial linear total differential relations between the equations -- see Appendix A) can itself admit a compatibility operator (relations between the relations). Similarly, Noether identities can be related by first-stage Noether identities, which satisfy second-stage Noether identities... It is naturel to refer to the existence of non-trivial higher-stage Noether identities as the reducibility of the considered gauge theory. Since we deal in this text with a {\sl first-order reducible} gauge theory, no non-trivial first-stage Noether identity should exist, i.e., any linear total differential operator $(S_\zb^1\ldots S_\zb^K) D_x^\zb$ such that $S_\zb^\zd D_x^\zb\circ R^a_{\zd\za}D_x^\za=0$ should be trivial, should vanish, or, still, all its coefficients should vanish. As mentioned, in the present approach to the Koszul-Tate resolution, `trivial' (resp., `non-trivial') means what has been called `weakly trivial' (resp., `not weakly trivial') in the preceding subsection, i.e., it means that all the coefficients vanish (resp., at least one coefficient does not vanish) on $\zS$. Hence, we actually deal with {\bf first-order on-shell reducibility}. This means that \be\label{OnShIrr}S_\zb^\zd D_x^\zb\circ R^a_{\zd\za}D_x^\za\approx 0\quad\text{must imply that}\quad S_\zb^\zd\approx 0\quad(\forall\;\zd\in\{1,\ldots,K\})\;.\ee It can be shown \cite{Bar} that this {\it first-order on-shell reducibility condition} really holds -- in view of the above first-order reducibility assumption ({\small IA}).\medskip

In view of (\ref{IrrSetNO}) and (\ref{OnShIrr}), the linear total / horizontal differential operators $R^a_\zd=R^a_{\zd\za}D_x^\za$ are the irreducible set of non-trivial Noether operators, which has been announced in Proposition \ref{IrrSetNOProp}.

\begin{rem}\label{IrRed}{\em Observe that {\it regularity} does no longer mean, as in Subsection \ref{KRS}, that all the equations $E_{\frak a}$ are independent, but that some equations $E_{\frak a}$ are independent. The other equations $E_\zD$ are dependent and they are generated via differentiation by a finite number of dependent equations $E_\zd$. These dependent generators induce Noether identities, i.e., non-trivial relations between the equations. These relations are themselves on-shell independent, i.e., there are no non-trivial first-stage Noether identities. The latter situation is referred to as `{\it irreducibility}' in \cite{Bar}, whereas it is called `{\it first-order reducibility}' in \cite{HT}}.\end{rem}

\subsubsection{Koszul-Tate resolution in a regular first-order reducible theory}\label{KTRASF}

In this subsection, we report on a {\bf Koszul-Tate resolution} of the algebra $\Ci(\zS)=\cF/I(\zS)$ of functions of the shell $\zS$, in the case of a regular (on-shell) first-order reducible gauge theory. We are thus in the situation (\ref{IrrSetNO}) -- (\ref{OnShIrr}), which has already been described in Remark \ref{TateHT}, and we build a resolution that is similar to the one of Theorem \ref{TateTheo4}. Since the irreducible non-trivial Noether operators $R^a_\zd$, or, still, the Noether identities $R^a_{\zd\za}\,D_x^\za\zd_{u^a}{\cal L}\equiv 0$ and their extensions \be\label{NI}D_x^\zb\; R^a_{\zd\za}\,D_x^\za\zd_{u^a}{\cal L}\equiv 0\;,\ee correspond to the `irreducible non-trivial' relations $\sum_{\frak a}\bar s^{\frak a}_{\frak b}E_{\frak a}=0$ of Remark \ref{TateHT}, we do, just as in Theorem \ref{TateTheo4}, not only associate degree 1 generators $\zf^{\za*}_a$ to the equations $D_x^\za\zd_{u^a}{\cal L}=0$ of $\zS$, but we assign further degree 2 generators $C^{\zb*}_\zd$ to the (irreducible) relations (\ref{NI}) (no degree 3 generators are needed). The candidate for a Koszul-Tate resolution of $\Ci(\zS)$ is then the chain complex, whose chains are the elements of the free Grassmann algebra \be\label{KT}\op{KT}=\cF\otimes \cS[\zf^{\za*}_a,C^{\zb*}_\zd]\;\ee (see Equation (\ref{KT4})) and whose differential is defined by $$\zd_{\op{KT}}=D^\za_x\zd_{u^a}{\cal L}\;\p_{\zf^{\za*}_a}+D_x^\zb(R^a_{\zd\za}\,\zf^{\za*}_a)\;\p_{C^{\zb*}_\zd}\;$$ (see Equation (\ref{KT5})).\medskip

Just as the {\it fiber coordinates $u_\za^a$ (in the following, we denote them by $\zf_\za^a$) of the jet space} of $E$ are algebraizations of the derivatives $\p_x^\za\zf^a$ of the components of a section $\zf$ (field) of the vector bundle $\zp:E\to X$, the generators $\zf^{\za*}_a$ and $C^{\zb*}_\zd$ symbolize the total derivatives $D_x^\za\zf^*_a$ and $D_x^\zb C^*_\zd$ of the components of sections $\zf^*$ and $C^*$ (fermionic antifield and bosonic antifield) of the pullback bundles $\zp_\infty^*F_1\to J^\infty E$ and $\zp_\infty^*F_2\to J^\infty E$ of some vector bundles $F_1\to X$ and $F_2\to X$. Hence, {\it the $\zf^{\za*}_a$ and $C^{\zb*}_\zd$ can be thought of as the fiber coordinates of the horizontal jet spaces} of $\zp_\infty^*F_1$ and $\zp_\infty^*F_2\,$, respectively.\medskip

In the sequel, we thus put the antifields $\zf^*$ and $C^*$ on an equal footing with the fields $\zf$. More precisely, we extend the definition of the total derivatives by setting \be\label{ExtTotDer}\bar D_{x^{i}}=\p_{x^i}+ \zf^a_{i\za}\p_{\zf^a_\za}+ \zf^{i\za*}_a\p_{\zf^{\za*}_a}+C^{i\zb*}_\zd\p_{C^{\zb*}_\zd}\;,\ee so that they act on functions of the {\bf extended jet space}, and we finally define the Koszul-Tate differential by \be \zd_{\op{KT}}=D^\za_x\zd_{u^a}{\cal L}\;\p_{\zf^{\za*}_a}+\bar D_x^\zb(R^a_{\zd\za}\,\bar D_x^\za\zf^{*}_a)\;\p_{C^{\zb*}_\zd}\;\;.\ee

The homology of $(\op{KT},\zd_{\op{KT}})$ is actually concentrated in degree 0, where it coincides with $\Ci(\zS)$. Indeed, the $0$-cycles are the functions $\cF$ and the $0$-boundaries are the $$\zd_{\op{KT}}\left(\sum F^a_\za\zf^{\za*}_a\right)=\sum F^a_\za D_x^\za\zd_{u^a}{\cal L}\approx 0\;.$$ In view of the regularity assumption 2, the equations $E_{\frak a}$ play the same role as in Subsection \ref{KRS}, so that the ideal $I(\zS)$ of those functions of $\cF$ that vanish on $\zS$ is made of the combinations $\sum F^{\frak a} E_{\frak a}$. Therefore, not only any $0$-boundary belongs to $I(\zS)$, but, conversely, any function of $I(\zS)$ reads $$\sum F^{\frak a}E_{\frak a}=\sum F^{\frak a}(M^{-1})^a_{{\frak a}\za}\,D^\za_x\zd_{u^a}{\cal L}=\zd_{\op{KT}}\left(\sum F^{\frak a}\,(M^{-1})^a_{{\frak a}\za}\,\zf^{\za*}_a\right)\;$$ and is therefore a $0$-boundary. It follows that $H_0(\op{KT})=\cF/I(\zS)=\Ci(\zS)$. To show that the homology vanishes in higher degrees, one needs the first-order reducibility assumption ({\small IA}) \cite{Bar}.\medskip

In fact, the above {\bf irreducible set of non-trivial Noether operators} $R^a_\zd$ {\bf is generating}, in the sense that any Noether operator $(N^1_\za\ldots N^r_\za)D_x^\za$, i.e., any total differential operator such that $N^a_\za\,D_x^\za\zd_{u^a}{\cal L}\equiv 0$, uniquely reads \be\label{univ}N^a_\za\,D_x^\za=S^\zd_{\zg}\,D_x^\zg\circ R^a_{\zd\zb}\,D_x^\zb+M^{[a,b}_{\za,\zb]}\,D_x^\zb\zd_{u^b}{\cal L}\,D_x^\za\;,\ee where the coefficients belong to $\cF$ and satisfy $S^\zd_\zg\not\approx 0$ and $M^{[a,b}_{\za,\zb]}=-M^{[b,a}_{\zb,\za]}$. Hence, in a regular first-order reducible gauge theory, any Noether operator $(N^1\ldots N^r)$ coincides on-shell with a composite $(S^\zd\circ R^1_\zd\ldots S^\zd\circ R^r_\zd)$ of the irreducible set of Noether operators with some total differential operators. This result is actually a quite straightforward corollary of the fact that $H_1(\op{KT})=0$.

\subsubsection{Koszul-Tate resolution in a regular higher-order reducible theory}\label{RegOnRed}

The existence of non-trivial first- or higher-stage Noether identities is referred to as `reducibility' in \cite{Bar} and as `higher order reducibility' in \cite{HT}. The precise description of higher order reducibility and of the corresponding physical background \cite{HT} would lead far beyond the scope of this text. Let us thus just mention that, from a mathematical standpoint, higher order reducibility is similar to Verbovetsky's framework, which we describe in the next section, except that Verbovetsky considers regular off-shell reducibility.

\section{Koszul-Tate resolution in Cohomological Analysis of PDE-s}\label{CCKTR}

Below we detail some ideas of \cite{Ver} adopting a slightly different standpoint.

\begin{rem}{\em As in the preceding subsection, we will use -- now without further reference -- standard concepts, results, and notation of the cohomological analysis of {\small PDE}-s \cite{Sascha}. For a summary of the needed knowledge, we refer the reader to Appendix A, Section \ref{Vino}. In Subsection \ref{KTRCC}, we also use some ideas of the $\cD$-geometric approach to {\small PDE}-s \cite{BD}. Some details can be found in Appendix B, Section \ref{Jet}, as well as in \cite{BPP1}, \cite{BPP2}, \cite{BPP3}, \cite{BPP4}.}\end{rem}

\subsection{Triviality, regularity and off-shell reducibility assumptions}

In Subsection \ref{RIGT}, we described -- within the smooth geometric setting and for a fixed choice of coordinates -- the classical Koszul-Tate resolution used in Mathematical Physics. The starting point was made of field theoretic Euler-Lagrange equations, with Noether identities relating them, and with precise regularity and first-order on-shell reducibility assumptions.\medskip

In the present case, the context will be as well {\it smooth geometry} and, just as in Mathematical Physics, we will work {\it in local coordinates}, although some aspects are developed in a coordinate-free manner. Our springboard will be any {\it not necessarily linear {\small PDE}}, for which we formulate {\it regularity} and {\it off-shell reducibility conditions}.\medskip

More precisely, let $\zp:E\to X$ and $\zr_1:F_1\to X$ be smooth vector bundles of ranks $r$ and $r_1$, respectively, over a smooth manifold of dimension $n$. Take a not necessarily linear formally integrable {\bf {\small PDE}} $\zS^0\subset J^k(\zp)$ of order $k$, which is implemented by a not necessarily linear differential operator $D\in\op{DO}_k(\zp,\zr_1)$: $\zS^0=\ker\psi_D$, where $\psi_D\in{\tt FB}(J^k(\zp),F_1)$ is the representative fiber bundle morphism of $D$. Recall (see Section \ref{Vino}) that $$\op{DO}_k(\zp,\zr_1)\simeq {\tt FB}(J^k(\zp),F_1)\simeq \cF_k(\zp,\zr_1):=\zG(\zp_k^*(\zr_1))\subset \zG(\zp_\infty^*(\zr_1))=:\zG(R_1)=:\cR_1\;$$ (in the sequel, we often denote a vector bundle over $X$ by a Greek lower-case character, its pullback over $J^\infty(\zp)$ by the corresponding Latin capital, and the module of sections of the latter by the same calligraphic letter). As usual, we denote by $\zS\subset J^\infty(\zp)$ the infinite prolongation of $\zS^0\subset J^k(\zp)$: $\zS=\ker\psi^\infty_D$, where $\psi^\infty_D\in{\tt FB}(J^\infty(\zp),J^\infty(\zr_1))$ is the infinite prolongation of $\psi_D$.\medskip

We now describe the locality and regularity hypotheses used in \cite{Ver}. In fact, the author assumes that $\zS$ is contained in a small open subset $U\subset J^\infty(\zp)$, in which there exist coordinates $(x^i,u^a_\za)$. Also in the bundle $\zr_1$ fiber coordinates are fixed (we will not need their denotation, only its index $\zl\in\{1,\ldots,r_1\}$ will be used). In addition to these {\bf triviality} conditions, he formulates a {\bf regularity} requirement for $\zS$. Just as for the classical Koszul-Tate resolution of Mathematical Physics, it is assumed that some equations of $\zS$ can be chosen as first or last coordinates of a new system (of course, the equations of $\zS$ read in the considered trivializations $D_x^\za\psi_D^\zl=0$, for all $\za\in\N^n$ and $\zl\in\{1,\ldots,r_1\}$). More precisely, the {\it neighborhood} $U$ of $\zS$ is assumed to be a trivial bundle over $\zS$, in the sense that there is an isomorphism $\Phi:U\to \zS\times V$, where $V$ is a star-shaped {\it neighborhood} of 0 in $\R^\infty$, such that the coordinates $v=(v^1,v^2,\ldots)$ in $V$ are precisely certain equations of $\zS$ (not necessarily all of them): for any ${\frak a}\in\N$, there is an $\za_{\frak a}\in \N^n$ and  a $\zl_{\frak a}\in\{1,\ldots,r_1\}$, such that $v^{\frak a}=D_x^{\za_{\frak a}}\psi_D^{\zl_{\frak a}}$. This means that the fiber coordinates $v(\zk)$ of a point $\zk\in\zS$, which are obtained by projecting $\zF(\zk)$ on the second factor $V$, vanish. In addition, the projection of $\zF(\zk)$, $\zk\in\zS$, on the first factor $\zS$, is simply $\zk$.\medskip

Although in the following we systematically consider the open subset $U\subset J^\infty(\zp)$ instead of the whole jet space, we do not always insist on this restriction (and even write for simplicity sometimes $J^\infty(\zp)$ instead of $U$).\medskip

The latter regularity condition has the same {\it fundamental consequence} as in Subsections \ref{KRS} and \ref{KTRASF}: if a function $F\in\cF$ vanishes on $\zS$, it is a finite sum of the type $$F=\sum F_{\za_{\frak a},\zl_{\frak a}} D_x^{\za_{\frak a}}\psi_D^{\zl_{\frak a}}\;,$$ with $F_{\za_{\frak a},\zl_{\frak a}}\in\cF$. In other words, a function $F\in\cF$ belongs to the ideal $I(\zS)$ if and only if it reads $F=\Psi(\psi_D)$, for some $\Psi\in\cC\!\op{Diff}(\cR_1,\cF)$.\medskip

{\sl In Subsection \ref{RIGT}}, we assumed first-order on-shell reducibility, i.e., we assumed that there are no on-shell first stage Noether identities. More precisely, there does exist a generating irreducible set of Noether operators $R^a_\zd=R^a_{\zd\za} D_x^\za$, or, still, a horizontal linear differential operator $$\zD_1=\left(\begin{array}{ccc}R^1_1&\ldots &R^r_1\\\vdots & & \vdots\\ R^1_K&\ldots & R^r_K\end{array}\right)\;\;,$$ i.e., an operator $\zD_1\in \cC\!\op{Diff}(\zp_\infty^*(\zr_1),\zp_\infty^*(\zr_2))$ (in the considered special case of Subsection \ref{RIGT}, the bundle $\zr_1$ coincides with the bundle $\zp$). In this new notation, the relations $R^a_{\zd\za} D_x^\za\zd_{u^a}{\cal L}\equiv 0$, for all $\zd\in\{1,\ldots,K\},$ read $\zD_1(\zd_{u^\bullet}{\cal L})\equiv 0$. Note that the {\small LHS} of the algebraicized Euler-Lagrange equations $\zd_{u^\bullet}{\cal L}=0$ is the representative morphism $\psi_D$ of a not necessarily linear differential operator $D\in\op{DO}(\zp,\zr_1)$. The universal linearization of the latter is a horizontal linear differential operator $\ell_D\in\cC\!\op{Diff}(\zp^*_\infty(\zp),\zp^*_\infty(\zr_1))$. When linearizing the identity $\zD_1(\psi_D)\equiv 0$, we get $\zD_1\circ\ell_D=0$. Since $\zD_1$ is generating, it does not vanish and, for any operator $\nabla\in\cC\!\op{Diff}(\zp^*_\infty(\zr_1),\zp^*_\infty(\zr_2'))$, such that $\nabla(\psi_D)\equiv 0$, there is an operator $\square\in\cC\!\op{Diff}(\zp^*_\infty(\zr_2),\zp^*_\infty(\zr'_2))$, such that $\nabla\approx\square\circ\zD_1$, see Equation (\ref{univ}). Hence, roughly speaking, the restriction $\zD_1|_\zS$ is an on-shell compatibility operator for $\ell_D|_\zS$, and the mentioned first-order on-shell reducibility means that there is no on-shell compatibility operator for $\zD_1|_\zS$, see Equation (\ref{OnShIrr}).\medskip

{\sl We now come back to the context of \cite{Ver}}. The restricted linearization $\ell_D|_\zS$ of the considered operator $D$ admits a compatibility operator $\zD_\zS\in\cC\!\op{Diff}(\cR_1|_\zS,\cR_2|_\zS)$. One of the first results in \cite{Ver} states that $\zD_\zS$ can be extended to an operator $\zD_1\in\cC\!\op{Diff}(\cR_1,\cR_2)$, such that $\zD_1(\psi_D)=0$. Just as any other horizontal linear differential operator, the extension $\zD_1$ admits a formally exact compatibility complex. However, the latter is a priori neither finite, nor are its $\cF$-modules $\cR_i$ modules of sections of vector bundles of finite rank. One of the main assumptions of \cite{Ver} is that there exists a finite formally exact compatibility complex \be\label{CompComp}0\longrightarrow\cR_1\stackrel{\zD_1}{\longrightarrow}\cR_2\stackrel{\zD_2}{\longrightarrow}\ldots\stackrel{\zD_{k-2}}{\longrightarrow}\cR_{k-1}\longrightarrow 0\;,\ee whose $\cF$-modules $\cR_i$ are all modules $\cR_i=\zG(R_i)=\zG(\zp^*_\infty(\zr_i))$, where the $\zr_i:F_i\to X$ are rank $r_i$ smooth vector bundles, and whose arrows are horizontal operators $\zD_i\in\cC\!\op{Diff}(\cR_i,\cR_{i+1})$. This hypothesis is of course an {\bf off-shell reducibility condition}.

\subsection{Koszul-Tate resolution induced by a compatibility complex \cite{Ver}}\label{KTRCC}

Formal exactness of (\ref{CompComp}) implies in particular that, when applying the horizontal infinite jet functor $\bar{\cal J}^\infty$ to the complex (\ref{CompComp}), we get an exact sequence of $\cF$-modules: \be\label{ExSeq1}0\longrightarrow\bar{\cal J}^\infty(\cR_1)\stackrel{\bar\psi^\infty_{\zD_1}}{\longrightarrow}\bar{\cal J}^\infty(\cR_2)\stackrel{\bar\psi^\infty_{\zD_2}}{\longrightarrow}\ldots\stackrel{\bar\psi^\infty_{\zD_{k-2}}}{\longrightarrow}\bar{\cal J}^\infty(\cR_{k-1})\longrightarrow 0\;.\ee 
Next we use the left exact contravariant Hom functor $\h_\cF(-,\cF)$, what leads to the exact sequence $$\h_\cF(\bar{\cal J}^\infty(\cR_1),\cF)\stackrel{-\circ\bar\psi^\infty_{\zD_1}}{\longleftarrow}\h_\cF(\bar{\cal J}^\infty(\cR_2),\cF)\stackrel{-\circ\bar\psi^\infty_{\zD_2}}{\longleftarrow}\ldots$$ \be\label{ExSeq2}\stackrel{-\circ\bar\psi^\infty_{\zD_{k-2}}}{\longleftarrow}\h_\cF(\bar{\cal J}^\infty(\cR_{k-1}),\cF)\longleftarrow 0\;\ee of $\cF$-modules. The identification of representative morphisms with the corresponding differential operators finally gives the exact sequence \be\label{ExSeq3}\cC\!\op{Diff}(\cR_1,\cF)\stackrel{-\circ\zD_1}{\longleftarrow}\cC\!\op{Diff}(\cR_2,\cF)\stackrel{-\circ\zD_2}{\longleftarrow}\ldots\stackrel{-\circ\zD_{k-2}}{\longleftarrow}\cC\!\op{Diff}(\cR_{k-1},\cF)\longleftarrow 0\;\ee \cite[Section 5.5.5]{Sascha}. The completion \be\label{Complex}0\longrightarrow\cC\!\op{Diff}(\cR_{k-1},\cF)\stackrel{-\circ\zD_{k-2}}{\longrightarrow}\cC\!\op{Diff}(\cR_{k-2},\cF)\stackrel{-\circ\zD_{k-3}}{\longrightarrow}\ldots
\stackrel{-\circ\zD_{1}}{\longrightarrow}\cC\!\op{Diff}(\cR_{1},\cF)\stackrel{-(\psi_D)}{\longrightarrow}\cF\longrightarrow 0\;\ee of the latter sequence by $-(\psi_D)$ is a complex of $\cF$-modules for the natural grading given by the subscripts of the $\cR_i$. This complex, which is exact in all spots, except, maybe, in degrees 0 and 1, is actually made of $\cF[\cD]$-modules (see Section \ref{Jet}). Indeed, in view of Equation (\ref{LDOCoeffFun}), we have $$\cF[\cD]:=\cF\0\cD\simeq\cC\cD(J^\infty(\zp)):=\cC\!\op{Diff}(\cF,\cF)\;,$$ so that the $\cF[\cD]$-action is given by left composition (except for $\cF$). Hence, the arrows of this complex are $\cF[\cD]$-linear maps and the complex itself is a differential graded $\cF[\cD]$-module $$(\cC\!\op{Diff}(\cR_\bullet,\cF),\zd_{\op{{\frak{KT}}}})\in\tt DG\cF[\cD]M\;,$$ where $\zd_{\op{{\frak{KT}}}}$ is the direct sum of the maps in (\ref{Complex}). The graded symmetric tensor algebra functor $\cS_\cF$ sends this module to the free differential graded $\cF[\cD]$-algebra \be\label{Complex2}(\op{{\frak{KT}}},\zd_{\op{{\frak{KT}}}}):=(\cS_\cF\;\cC\!\op{Diff}(\cR_\bullet,\cF),\zd_{\op{{\frak{KT}}}})\in\tt DG\cF[\cD]A\;,\ee whose differential is a degree $-1$ graded derivation of the graded symmetric tensor product. The latter complex is the {\bf Koszul-Tate complex}, in the sense of \cite{Ver}, associated to the considered partial differential equation.\medskip

The homology space $H_0(\op{{\frak{KT}}})$ coincides with $\Ci(\zS)$ (in view of the above-mentioned fundamental consequence of the regularity condition, the standard argument goes through) and the higher homology spaces vanish (as suggested by the above sequences). To prove this statement, it suffices to show that the Koszul-Tate complex (\ref{Complex2}) coincides -- as claimed -- with the Koszul-Tate complex defined in \cite{Ver} and to use the corresponding result therein. The algebra of Koszul-Tate chains is defined in \cite{Ver} as the graded polynomial function algebra ${\cal P}\!\op{ol}(\bar J^\infty(R_\bullet))$. As usual, the polynomial functions ${\cal P}\!\op{ol}(\bar J^\infty(R_\bullet))$ are the smooth functions $\cF(\bar J^\infty(R_\bullet))$ that are polynomial along the fibers of the considered bundle -- here $\bar J^\infty(R_\bullet)\to J^\infty(\zp)$. Just as the polynomial functions of a vector bundle $G\to X$ are defined by $${\cal P}\!\op{ol}(G):=\zG(\cS G^*)\simeq \cS_{\Ci(X)}\zG(G^*)=\cS_{\Ci(X)}\h_{\Ci(X)}(\zG(G),{\Ci(X)})\;,$$ the polynomial functions considered here are defined by $${\cal P}\!\op{ol}(\bar J^\infty(R_\bullet)):=\cS_\cF \h_\cF(\bar{\cal J}^\infty(\cR_\bullet),\cF)\simeq\cS_\cF\;\cC\!\op{Diff}(\cR_\bullet,\cF)\;.$$ Hence, the Koszul-Tate chains of \cite{Ver} and those defined above do coincide. Moreover, the Koszul-Tate differential is defined in \cite{Ver} as an odd evolutionary vector field $\zd$ of $\bar J^\infty(R_\bullet)$. Such a graded derivation, when restricted as here to ${\cal P}\!\op{ol}(\bar J^\infty(R_\bullet))$, is completely defined by its values on the polynomial functions that are linear along the fibers, i.e., on $\h_\cF(\bar{\cal J}^\infty(\cR_\bullet),\cF)\simeq\cC\!\op{Diff}(\cR_\bullet,\cF)$ -- and by its values on $\cF$. But on $\nabla_i\in\cC\!\op{Diff}(\cR_i,\cF)$ (resp., $F\in\cF$), this evolutionary field is given by $\zd(\nabla_i)=\nabla_i\circ\zD_{i-1}$, if $i\ge 2$, and by $\zd(\nabla_1)=\nabla_1(\psi_D)$ (resp., $\zd(F)=0$) \cite[Proposition 5]{Ver}. Hence, the odd derivations $\zd$ and $\zd_{\op{{\frak{KT}}}}$ coincide, the Koszul-Tate complexes $({\cal P}\!\op{ol}(\bar J^\infty(R_\bullet)),\zd)$ and $(\op{{\frak{KT}}},\zd_{\op{{\frak{KT}}}})$ coincide, and so do their homologies.

\section{Koszul-Tate resolution in Homotopy Theory}\label{KTRHT}

\begin{rem}\label{ModelDiffOp}{\em In this section, we use the model structure of the category $\tt DG\cD A$ of differential non-negatively graded commutative unital algebras over the ring $\cD$ of differential operators. We aim at providing an, as far as possible, self-contained exposition. For further details on definitions, results, on $\cD$-modules, sheaves, model categories $\ldots\,,$ the reader may consult Appendix B, Section \ref{Jet}, as well as \cite{BPP1}, \cite{BPP2}, and \cite{BPP4}, and in particular the appendices therein. Note that, whereas the frame for the preceding sections was algebra or smooth geometry, the context of the mentioned papers and this section is algebraic geometry. We will work over a {\em smooth scheme}, since for an arbitrary, maybe singular, scheme $X$, the notion of left $\cD_X$-module is meaningless \cite[Remark p. 56]{BD}. }\end{rem}

\subsection{Model structure on $\tt DG\cD A$}\label{KTRHT1}

Let $X$ be a smooth scheme and let $\cO_X$ (resp., $\cD_X$) be the sheaf of rings of functions (resp., differential operators) of $X$. Denote by ${\tt qcCAlg}(\cO_X)$ (resp., ${\tt qcCAlg}(\cD_X)$) the category of commutative unital $\cO_X$-algebras (resp., commutative unital $\cD_X$-algebras, i.e., commutative unital $\cO_X$-algebras, whose $\cO_X$-module structure can be extended to a $\cD_X$-module structure, such that vector fields $\zy\in\cD_X$ act as derivations on the product) that are quasi-coherent as $\cO_X$-modules. {\bf We will refer to the objects of this category as $\cO_X$-algebras (resp., $\cD_X$-algebras)} (this convention differs from the one adopted in \cite{BPP1}, \cite{BPP2}, and \cite{BPP4}). The forgetful functor has a left adjoint \cite{BD} $${\cal J}^\infty:{\tt qcCAlg}(\cO_X)\to {\tt qcCAlg}(\cD_X):\op{For}\;,$$ called the {\bf jet functor} (see Appendix B, Subsection \ref{ConsJetFunc}).

\begin{prop}\label{VBFunAlg} Let $\zp:E\to X$ be an algebraic vector bundle of finite rank over a smooth scheme $X$ and denote by $\cO_E$ the structure sheaf of the scheme $E$. If $\zp_*$ stands for the direct image by $\zp,$ we have $\cO^E_X:=\zp_*\cO_E\in {\tt qcCAlg}(\cO_X)$ and thus ${\cal J}^\infty(\cO^E_X)\in{\tt qcCAlg}(\cD_X)$.\end{prop}

The $\cD_X$-algebra ${\cal J}^\infty(\cO^E_X)$ (or its total section $\cD_X(X)$-algebra ${\cal J}^\infty(\cO^E_X)(X)$) is the $\cD$-geometric counterpart of the function algebra $\cO(J^\infty E)$ $=\cF(\zp_\infty)=\cF$ of the infinite jet space of a smooth vector bundle $\zp:E\to X$. Note that we prefer in this section the notation $J^\ell E$ to the notation $J^\ell(\zp)$ ($0\le\ell\le\infty$). Proposition \ref{VBFunAlg} is rather natural. A proof can be found in Appendix B, Subsection \ref{ProofVBFunAlg}.\medskip

\begin{rem}{\em In \cite{BPP1} and \cite{BPP2}, as well as in \cite{BPP4}, we proved in particular that the category $\tt DG\cD A$ of differential non-negatively graded commutative unital algebras over $\cD=\cD_X(X)$ admits a cofibrantly generated model structure, if $X$ is a {\em smooth affine variety}. This theorem results from the transfer of the model structure on the category $\tt DG\cD M$ of differential non-negatively graded $\cD$-modules to the category $\tt DG\cD A$. Actually, the categories under investigation are the category ${\tt DG_+qcMod}(\cD_X)$ of {\sl sheaves} of differential non-negatively graded $\cO_X$-quasi-coherent $\cD_X$-modules and the category ${\tt DG_+qcCAlg}(\cD_X)$ of {\sl sheaves} of differential non-negatively graded $\cO_X$-quasi-coherent commutative unital $\cD_X$-algebras, i.e., of commutative monoids in the symmetric monoidal category ${\tt DG_+qcMod}(\cD_X)$. The restriction to a smooth affine variety (both assumptions are necessary) allows to show that the total section functor yields an equivalence of categories \be\label{EquivCat}\zG(X,-):{\tt DG_+qcCAlg}(\cD_X)\rightleftarrows\tt DG\cD A\;,\ee and similarly for ${\tt DG_+qcMod}(\cD_X)$ and $\tt DG\cD M$. These equivalences allow in turn to avoid the problem of the non-existence of a projective model structure on ${\tt DG_+qcMod}(\cD_X)$ for an arbitrary smooth scheme \cite{Gil06} and so the problem of the non-existence of a transferred structure on ${\tt DG_+qcCAlg}(\cD_X)$.}\end{rem}

Before we describe the model structure of $\tt DG\cD A$, we recall the

\begin{defi}[\cite{BPP4}]\label{DefRSDA} A {\bf relative Sullivan $\cD$-algebra} $(\,${\small RS$\cD\!$A}$\,)$ is a $\tt DG\cD A$-morphism $(\,$standard definition$\,)$
$$(A,d_A)\to(A\0 \cS V,d)\;$$ $(\,$the tensor product functor $\0$ and the graded symmetric tensor algebra functor $\cS$ are taken over the ring $\cO=\cO_X(X)$ and the differential $d$ is usually not the standard differential on a tensor product$\,)$ that sends $a\in A$ to $a\0 1_\cO\in A\0 \cS V$. Here $V$ is a free non-negatively graded $\mathcal{D}$-module $$V=\bigoplus_{\za\in J}\,\cD\cdot v_\za\;,$$ which admits a homogeneous basis $(v_\za)_{\za\in J}$ that is indexed by a well-ordered set $J$, and is such that \be\label{Lowering}d v_\za=d (1_A\0 v_\za) \in A\0 \cS V_{<\za}\;,\ee for all $\za\in J$. In the last requirement, we set $V_{<\za}:=\bigoplus_{\zb<\za}\cD\cdot v_\zb\,$. We refer to Property (\ref{Lowering}) by saying that $d$ is {\bf lowering}.\medskip

A {\small RS$\cD\!$A} with the property \be\label{minimal}\za\le\zb \Rightarrow \deg v_\za\le\deg v_\zb\;\ee $(\,$resp., with the property $d=d_A\0 \id + \id\0\, d_\cS,$ where $d_\cS$ is a differential on $\cS V$ $(\,$in particular the differential $d_\cS=0$$\,)$; over $(A,d_A)=(\cO,0)$$\,)$ is called a {\bf minimal} {\small RS$\cD\!$A} $(\,$resp., a {\bf split} {\small RS$\cD\!$A}; a {\bf Sullivan $\cD$-algebra} $(\,${\small S$\cD\!$A}$\,)$$\,)$  and it is often simply denoted by $(A\0 \cS V,d)$ $(\,$resp., $(A\otimes \cS V,d)$; $(\cS V,d)$$\,)$.
\end{defi}

The concept of relative Sullivan $\cD$-algebra is similar to the notion of relative Sullivan $\Q$-algebra, which originates from Rational Homotopy Theory.

\begin{theo}\label{FinGenModDGDA} The category $\tt DG\mathcal{D}A$ of differential non-negatively graded commutative unital algebras over the ring $\cD=\cD_X(X)$ of total sections of the sheaf $\cD_X$ of differential operators of a smooth affine variety $X$, is a finitely $(\,$and thus a cofibrantly$\,)$ generated model category $(\,$in the sense of \cite{GS} and in the sense of \cite{Hov}$\;)$. The weak equivalences are the $\tt DG\cD A$-morphisms that induce an isomorphism in homology, the fibrations are the $\tt DG\cD A$-morphisms that are surjective in all positive degrees $p>0$, and the cofibrations are exactly the retracts of the relative Sullivan $\cD$-algebras.\end{theo}

Further, we describe in \cite{BPP1}, \cite{BPP2}, and \cite{BPP4} explicit functorial cofibration-fibration factorizations, as well as an explicit functorial cofibrant replacement functor. These descriptions are too long to be recalled here.\medskip

When remembering that the coproduct in $\tt DG\cD A$ is the tensor product, we get from \cite{Hir2}:

\begin{prop} For any differential graded $\cD$-algebra $A$, the coslice category $A\downarrow \tt DG\cD A$ carries a cofibrantly generated model structure given by the adjoint pair $L_{\0}:{\tt DG\cD A}\rightleftarrows A\downarrow{\tt DG\cD A}:\op{For}$, in the sense that its distinguished morphism classes are defined by $\op{For}$ and its generating cofibrations and generating trivial cofibrations are given by the functor $L_\0\,$, which sends $B$ in $\tt DG\cD A$ to $A\to A\0 B$ in $\tt A\downarrow\tt DG\cD A$.\end{prop}

\subsection{Koszul-Tate resolution implemented by a $\cD$-ideal}\label{KTRHT2}

A partial differential equation (see Appendix A, Section \ref{Vino}) of order $k$ acting on the sections $\zf$ of a smooth vector bundle $\zp:E\to X$ is a smooth fiber subbundle $\zS^0\subset J^kE$, and (at least if $\zS^0$ is formally integrable) its infinite prolongation $\zS\subset J^\infty E$ is a smooth manifold. If $\zS^0$ is implemented by a differential operator $D$ with representative morphism $\psi$, we have $\zS^0=\ker\psi$ and $\zS=\ker\psi^\infty$, where $\psi^\infty$ is the representative morphism of the infinite prolongation $j^\infty\circ D$ of $D$. In coordinates: the equation of $\zS^0$ is $\psi(x^i,u^a_\za)=0$ and the equation of $\zS$ is $(D^\zb_x\psi)(x^i,u^a_\za)=0,\forall \zb$. These equations are the algebraizations of the {\small PDE}-s $$\psi(x^i,\p_x^\za\zf^a)=0\quad\text{and}\quad d_x^\zb(\psi(x^i,\p_x^\za\zf^a))=0,\forall\zb\;.$$ Since the latter differential equations have the same solutions, we can focus on $\zS$ instead of $\zS^0$. Hence, a {\small PDE} $\zS^0$ can be thought of as a manifold $\zS$, or, in view of the space-algebra duality, as the function algebra $\Ci(\zS)$, which is (see above) the quotient of the algebra $\cO(J^\infty E)$ by the ideal $\cI$ of all functions of $\cO(J^\infty E)$ that vanish on $\zS$. A {\small PDE} acting on the sections of $E$ can thus finally be interpreted as an ideal $\cI\subset \cO(J^\infty E)$. It follows that, in our present $\cD$-geometric context, where we considered an algebraic vector bundle $\zp:E\to X$ over a smooth affine variety $X$, we think about a {\small PDE} acting on the sections of $E$, as a $\cD$-ideal (i.e., an $\cO$-ideal and a $\cD$-submodule) $\cI\subset \cJ$, where $$\cJ:={\cal J}^\infty(\cO^E_X)(X)=\zG(X,{\cal J}^\infty(\cO^E_X))\in\cD A\;$$ (see Equation (\ref{EquivCat})), and we think about $\cQ:=\cQ(\zp,\cI):=\cJ/\cI\in\tt\cD A$ as the $\cD$-algebra of the corresponding shell functions. Our goal is to resolve this $\cD$-algebra.\medskip

The fundamental concepts of the jet bundle formalism are the Cartan distribution and the Cartan connection, or, still, horizontal linear differential operators $\cC\!\op{Diff}(\zp_\infty^*(\zh_1),\zp_\infty^*(\zh_2))$ between pullback bundles $\zp_\infty^*(\zh_i):\zp_{\infty}^*F_i\to J^\infty E$ of smooth vector bundles $\zh_i:F_i\to X$. Hence, jets lead to a systematic base change $X\rightsquigarrow J^\infty E$. The remark is essential, in the sense that both, the classical Koszul-Tate resolution of Mathematical Physics (constructed above in the context of a regular first-order on-shell reducible gauge theory) and Verbovetsky's Koszul-Tate resolution (induced by the compatibility complex of the linearization of a differential equation), use the jet formalism to resolve shell functions, and thus enclose the base change $\bullet\to X$ $\;\rightsquigarrow\;$ $\bullet\to J^\infty E$. This means that, in the dual function algebra setting, or, in the present situation, in the dual $\cD$-algebra setting, we pass from $\tt DG\cD A$, i.e., from the coslice category $\cO(X)\downarrow \tt DG\cD A$ ($\cO(X):=\cO=\cO_X(X)$ is the base ring for the tensor product in $\tt DG\cD A$ and $(\cO, 0)$ is the initial object in $\tt DG\cD A$) to the coslice category $\cO(J^\infty E)\downarrow \tt DG\cD A$.\medskip

A first candidate for a resolution of ${\cal Q}=\cJ/\cI\in\tt \cD A$ is of course the {\bf cofibrant replacement of ${\cal Q}$} in $\tt DG\cD A$ given by the functorial `Cofibration -- Trivial Fibration' factorization of \cite[Theorem 28]{BPP4}, when applied to the unique $\tt DG\cD A$-morphism $\cO\to {\cal Q}$. Indeed, this decomposition implements a functorial cofibrant replacement functor $Q$ (\cite[Theorem 34]{BPP4}) with value $Q({\cal Q})=\cS V$ described in \cite[Theorem 28]{BPP4}: $$\cO\rightarrowtail \cS V\stackrel{\sim}{\twoheadrightarrow}{\cal Q}\;,$$ where $\rightarrowtail$ (resp., $\twoheadrightarrow$, $\stackrel{\sim}{\to}$) denotes a cofibration (resp., a fibration, a weak equivalence (here an isomorphism in homology)). Since ${\cal Q}$ is concentrated in degree 0 and has 0 differential, it is clear that $H_k(\cS V)$ vanishes, except in degree 0 where it coincides with ${\cal Q}$, so that $\cS V$ is indeed a resolution of $\cQ$.\medskip

In the next section, we suggest a general and precise definition of a Koszul-Tate resolution. Although such a definition does not seem to exist in the literature, it is commonly accepted that a Koszul-Tate resolution of the quotient $Q$ of a commutative ring $k$ by an ideal $I$ is a {\bf $k$-algebra that resolves $Q=k/I$}.\medskip

The natural idea -- to get a resolving ${\cal J}$-algebra for $\cQ$ -- is to replace $\cS V$ by ${\cal J}\0 \cS V$, and, more precisely, to consider the `Cofibration -- Trivial Fibration' decomposition \be\label{CTFJQ}{\cal J}\rightarrowtail {\cal J}\0 \cS V\stackrel{\sim}{\twoheadrightarrow}\cQ\;\ee of the canonical $\tt DG\cD A$-morphism $\cJ\to\cQ$ \cite[Theorem 28]{BPP4}. The differential graded $\cD$-algebra $\label{KTD} {\cal J}\0 \cS V$ {\it is} a {\bf $\cJ$-algebra that resolves} ${\cal Q}=\cJ/\cI$, but it is of course {\it not} a cofibrant replacement, since the left algebra in (\ref{CTFJQ}) is not the initial object $\cO$ in $\tt DG\cD A$ (further, the considered factorization does not canonically induce a cofibrant replacement in $\tt DG\cD A$, since it can be shown that the morphism $\cO\to {\cal J}$ is not a cofibration). However, as emphasized above, the Koszul-Tate problem requires a passage from the category $\tt DG\cD A$ to the category ${\cal J}\downarrow \tt DG\cD A$ (under the $\cD$-geometric counterpart $\cJ$ of $\cO(J^\infty E)$). It is easily checked that, in the latter undercategory, ${\cal J}\0 \cS V$ {\it is} a {\bf cofibrant replacement of $\cQ$}.

\begin{defi}[\cite{BPP4}]\label{KTRDG} Let $\cJ\in\tt \cD A$ be a $\cD$-algebra and let $\cI\subset\cJ$ be a $\cD$-ideal. The algebra $\cJ\0\,\cS V\in\tt DG\cD A$ given by the `Cofibration -- Trivial Fibration' factorization of $\cJ\to \cJ/\cI$ is a $\cJ$-algebra that resolves $\cJ/\cI$. Moreover, the algebra $\cJ\0\cS V$ $(\,$in fact $\cJ\rightarrowtail\cJ\0\cS V$$\,)$ is a cofibrant replacement of $\cJ/\cI$ $(\,$in fact of $\cJ\to\cJ/\cI$$\,)$ in the model category $\cJ\downarrow\tt DG\cD A$. We refer to $\cJ\0\cS V$ as the {\bf cofibrant replacement Koszul-Tate resolution} of $\cJ/\cI$. \end{defi}

\section{Koszul-Tate resolution in $\cD$-Geometry}\label{KTRDGDG}

In view of Subsection \ref{KTRHT2}, a Koszul-Tate resolution of a $\tt DG\cD A$-morphism $\cJ\to\cQ$, where $\cJ\in\tt \cD A$, should be an algebra $\cC\in\tt DG\cD A$, as well as a $\cal J$-algebra. This suggests to combine the $\cD$-action $\triangleright$ and the $\cJ$-action $\triangleleft$ in an action $\diamond$ of the ring $$\cJ[\cD]:=\cJ\0_\cO\cD$$ of linear differential operators with coefficients in $\cJ$, by setting, for any $j\in\cJ$, $D\in\cD$, and $c\in\cC$, $$(j\0 D)\diamond c=\left((j\0 1_\cO)\circ(1_\cJ\0 D)\right)\diamond c:=j\triangleleft(D\triangleright c)\;.$$

The introduction of the ring $\cJ[\cD]$ is the more natural as the algebra $\cJ=\cJ^\infty(\cO^E_X)(X)\in\tt\cD A$ is the $\cD$-geometric counterpart of the algebra $O(J^\infty E)=\cF=\cF(\zp_\infty)$ (that we denote in Appendix A, to simplify, also by $\cF(\zp)$), and as the $\cJ$-module $\cJ[\cD]=\cJ\0_\cO\cD\in\tt\cD M$ is the $\cD$-geometric analog of the $\cF$-module $\cF(\zp)\0_{\Ci(X)}\cD(X)\simeq \cC\cD(\cF,\cF)$ used in smooth geometry (see Appendix A). Indeed, as stressed in Subsection \ref{KTRHT2}, horizontal linear differential operators $\cC\cD(\cF,\cF)$ are the fundamental ingredient of the Koszul-Tate resolutions in Mathematical Physics and in Cohomological Analysis of {\small PDE}-s. Therefore, the passage from $\tt DG\cD A$ to $\tt DG\cJ[\cD]A$ corresponds to the necessary encryption of horizontal differential operators in the $\cD$-geometric approach to the Koszul-Tate resolution and to the Batalin-Vilkovisky formalism.\medskip

We will use the following notation. For any monoidal category $({\tt C}, \otimes, I)$ and any monoid $(\cA,\zm,\zh)$ in $\tt C$, we denote by $\tt Mod_{\tt C}(\cA)$ the category of (left) $\cA$-modules in $\tt C$, i.e., of $\tt C$-objects $M$ together with a ${\tt C}$-morphism $\zn:\cA\otimes M\rightarrow M$, such that the usual associativity and unitality diagrams commute. If $\tt C$ is symmetric monoidal, the category $\tt CMon(C)$ is the category of commutative monoids in $\tt C$. Finally, for any additive (or even Abelian) category $\tt E$, we denote by $\tt Ch_+(E)$ the category of non-negatively graded chain complexes in $\tt E$.

If $\cA\in\tt\cD A\subset DG\cD A$ is a differential graded $\cD$-algebra concentrated in degree 0 and with zero differential, we have \be\label{modules}\tt Mod_{DG\cD M}(\cA)= Ch_+(Mod_{\cD M}(\cA))= Ch_+(\cA[\cD]M)=DG\cA[\cD]M\;,\ee since, as well-known \cite{BD}, $$\tt Mod_{\cD M}(\cA)= Mod(\cA[\cD]))=:\cA[\cD]M\;.$$ It follows from (\ref{modules}) that \be\label{algebras}\tt DG\cA[\cD]A:=CMon(DG\cA[\cD]M)= CMon(Mod_{DG\cD M}(\cA))\simeq \cA\downarrow DG\cD A\;,\ee where the equivalence has been proven in detail in \cite{BPP4}. Equation (\ref{algebras}), together with Definition \ref{KTRDG}, provides additional evidence that a Koszul-Tate resolution of a $\tt DG\cD A$-map $\cA\to \cB$, with source $\cA\in\tt\cD A,$ should be an object $\cC$ in $$\cC\in \tt\cA\downarrow DG\cD A\simeq DG\cA[\cD]A=CMon(DG\cA[\cD]M)\;.$$

Hence, in the general situation, over a smooth -- not necessarily affine -- scheme $X$, we consider, in addition to the above mentioned category $${\tt DG_+qcCAlg}(\cD_X)={\tt CMon(DG_+qcMod}(\cD_X))\;,$$ also the category \be{\tt DG_+qcCAlg}(\cA_X[\cD_X])={\tt CMon(DG_+qcMod}(\cA_X[\cD_X]))\;\label{DGADAX}\ee of differential non-negatively graded $\cO_X$-quasi-coherent commutative unital $\cA_X[\cD_X]$-algebras, where $$\cA_X\in{\tt qcCAlg}(\cD_X)\quad\text{and}\quad \cA_X[\cD_X]=\cA_X\0_{\cO_X}\cD_X\;.$$ For simplicity, {\bf we refer to the objects of the category (\ref{DGADAX}) as differential graded $\cA[\cD_X]$-algebras} (thus writing $\cA$ instead of $\cA_X$). A few details on $\cA_X[\cD_X]$ and ${\tt DG_+qcCAlg}$ $(\cA_X[\cD_X])$ can be found in Appendix B, Section \ref{Jet} (we recommend to read Definition \ref{DGJDADefi} and Example \ref{DGJDAEx}).\medskip

Notice now that the cofibrant replacement Koszul-Tate resolution (see Definition \ref{KTRDG}) of a $\tt DG\cD A$-map $\cJ\to\cQ\,$, $\cJ\in\tt\cD A$, is the $\tt DG\cD A$-cofibration $\cJ\rightarrowtail\cJ\0\cS V$, whose target resolves $\cQ$ (see Equation (\ref{CTFJQ})) and which is, in view of Theorem \ref{FinGenModDGDA}, a retract of a relative Sullivan $\cD$-algebra, and, in view of \cite[Theorem 28]{BPP4}, even just a minimal (non-split) relative Sullivan $\cD$-algebra (see Definition \ref{DefRSDA}). This observation suggests the following two definitions, which generalize Definition \ref{DefRSDA} and the just recalled Definition \ref{KTRDG}, respectively, taking into account the above-motivated passage to the category (\ref{DGADAX}):

\begin{defi} Let $X$ be a smooth scheme and let $\cA$ be a $\cD_X$-algebra. A differential graded $\cA[\cD_X]$-algebra $\cC$ is said to be of {\bf Sullivan type}, if it admits an increasing filtration $\cC_0\subset\cC_1\subset\ldots$ by differential graded $\cD_X$-subalgebras, such that there is a differential graded $\cD_X$-algebra morphism $\cA\to\cC_0$ $(\,$we set $\cC_{-1}:=\cA\,)$ and that $\cC_k$ $(\,k\ge 0\,)$ is isomorphic as differential graded $\cD_X$-algebra to $\cC_k\simeq \cC_{k-1}\0\cS V_k,$ where $V_k$ is a locally projective graded $\cD_X$-submodule of $\cC_k$ such that $d_{\cC_k}V_k\subset\cC_{k-1}\,.$\end{defi}

\begin{defi}\label{KTRDef} Let $X$ be a smooth scheme, let $\cA$ be a $\cD_X$-algebra, and let $\zf:\cA\to\cB$ be a differential graded $\cD_X$-algebra morphism. A {\bf $\cD$-geometric Koszul-Tate resolution} of $\zf$ is a differential graded $\cA[\cD_X]$-algebra morphism $\psi:\cC\to \cB\,,$ which is a quasi-isomorphism in the category of differential graded $\cA[\cD_X]$-modules, and whose source $\cC$ is of Sullivan type. \end{defi}

\begin{rem}{\em\label{ExpSulTyp}Observe first that a quasi-isomorphism in the category of differential graded $\cA[\cD_X]$-modules is a morphism that induces a bijection in homology, i.e., is an $\cA$-linear quasi-isomorphism in the category of differential graded $\cD_X$-modules. Further, the differential on $\cC_{k-1}\0\cS V_k$ is $d_{\cC_k}$ and, since $d_{\cC_k}$ is a degree $-1$ graded derivation, it is completely defined by the differential of the differential graded $\cD_X$-subalgebra $\cC_{k-1}$ and the restriction $d_{\cC_k}|_{V_k}$ (note that, for $c\in\cC_{k-1}$ and $v,w\in V_k$, for instance, we have $c\0(v\odot w)=(c\0 1_{\cO_X})\star(1\0 v)\star(1\0 w)$, where $1$ is the unit in $\cC_{k-1}$ and $\star$ the multiplication in $\cC_k$).}\end{rem}

These definitions show that the confinement to the smooth affine case in Section \ref{KTRHT} does not only allow to use the artefacts of the model categorical environment, i.e., to compute the cofibrant replacement Koszul-Tate resolution, but allows also to discover the fundamental structure of this Koszul-Tate resolution, and to extend this structure to the general case of an arbitrary smooth scheme $X$.\medskip

The requirement that $\cC$ be equipped with an increasing filtration by differential graded $\cD_X$-subalgebras $\cC_k$ ($k\ge 0$) and that there exists a differential graded $\cD_X$-algebra morphism $j_0:\cA\to\cC_0$, is equivalent to the condition that $\cC$ be filtered by a sequence $\cC_0\subset\cC_1\subset\ldots$ of differential graded $\cA[\cD_X]$-subalgebras. Indeed, since $j_0:\cA\to \cC_0$, as well as the canonical inclusions $i_k:\cC_{k-1}\to \cC_k$ ($k\ge 1$), are differential graded $\cD_X$-algebra morphisms, we have differential graded $\cD_X$-algebra morphisms $j_k=i_k\circ\ldots\circ i_1\circ j_0:\cA\to\cC_k$ that provide a filtering sequence $\cC_0\subset\cC_1\subset\ldots$ of differential graded $\cA[\cD_X]$-subalgebras. Conversely, such a sequence gives a differential graded $\cD_X$-algebra morphism $\cA\ni a\mapsto a\triangleleft 1_{\cC_0}\in\cC_0\,.$ Hence, a $\cD$-geometric Koszul-Tate resolution of a differential graded $\cD_X$-algebra morphism $\zf:\cA\to\cB$ is the same as an $\cA$-semi-free resolution of $\zf$ in the sense of \cite{BD}. It follows \cite{BD} that the next proposition holds.

\begin{prop} Let $X$ be a smooth scheme and $\cA$ a $\cD_X$-algebra. Any differential graded $\cD_X$-algebra morphism $\cA\to\cB$ admits a $\cD$-geometric Koszul-Tate resolution. This holds in particular if $\cA=\cJ^\infty(\cO^E_X)\in{\tt qcCAlg}(\cD_X)$ is the $\cD_X$-algebra `of functions of the infinite jet space' of an algebraic vector bundle $\zp:E\to X$ of finite rank over a smooth scheme $X$.\end{prop}

\section{Comparison theorems}

In the following, we use the acronym {\small KTR} for `Koszul-Tate resolution'. Our goal is to show that all the {\small KTR}-s that we considered so far are $\cD$-geometric {\small KTR}-s, as well as to compare several {\small KTR}-s.

\subsection{Algebraic KTR and $\cD$-geometric KTR}

Tate's {\small KTR} \cite[Theorem 1]{Tate}, which we described briefly in the proof of Theorem \ref{TateTheo1}, is purely algebraic, there is no underlying space $X$, and there are no differential operators $\cD=\cD(X)$. Of course, one could consider the special situation where the Noetherian commutative unital ring $R$ is an algebra over a commutative unital algebra $\cO$ over some field, define linear differential operators $\cD$ on $\cO$ algebraically (the algebraic approach to differential operators is well-known, see, e.g., \cite{GKoP}), and compare Tate's resolution -- in this case -- with the $\cD$-geometric {\small KTR}. We see however no advantage in running through the technicalities of the geometrization of Tate's setting, and prefer to just compare the structures of the two resolutions.\medskip
 
A moment of reflection allows to see that the structure of the $\cD$-geometric {\small KTR} is exactly the same as that of Tate's resolution (ignore $\cD$ and take $\cA=R$).

\begin{rem}{\em $\cD$-geometric {\small KTR} can be traced back to minimal models in Homotopy Theory \cite{Halperin}. Let us start with a short historical note. Since the categories of topological spaces and simplicial sets have equivalent homotopy categories, simplicial sets are purely combinatorial models for classical Homotopy Theory. Kan constructed in 1958 algebro-combinatorial models: simplicial groups. In 1969, Quillen proved that the homotopy categories of simply-connected rational topological spaces and of connected differential graded Lie $\Q$-algebras are equivalent. Similarly, in 1977, Sullivan showed that there exists a categorical equivalence between the homotopy categories of simply-connected rational topological spaces with finite Betti numbers and of differential graded commutative $\Q$-algebras (category $\tt DG\Q A$) $(A^\bullet,d)$, whose cohomology spaces satisfy $H^0(A^\bullet,d)=\Q$, $H^1(A^\bullet,d)=0$, and $H^n(A^\bullet,d)$ is finite-dimensional for any $n$. This correspondence became really efficient due to the introduction of relative Sullivan minimal models of $\tt DG\Q A$-morphisms -- which are specific relative Sullivan $\Q$-algebras -- . Such models are (nowadays) obtained from the application of the small object argument to a most natural cofibrantly generated model structure on $\tt DG\Q A$. Hence, the cofibrant replacement {\small KTR}, which is a relative Sullivan minimal model, and its generalization, the $\cD$-geometric {\small KTR}, have no apparent link with Tate's {\small KTR} and with the {\small KTR}-s in Mathematical Physics and Cohomological Analysis, which are based on \cite{Tate}. Indeed, Tate's paper is a work in Homological Algebra and it originates from the attempt to replace the Koszul resolution of a regular sequence by a resolution that is valid even when the sequence is not regular. The analogy between these two types of {\small KTR}-s, the Tate type and the Sullivan type, might thus seem astonishing. However, both, Tate and Sullivan (and his successors), just looked for a good `resolution' of a commutative ring, and they used (in our opinion independently) the same `naive' technique -- the addition of generators to kill cycles or obstructions to isomorphisms in homology -- . This justifies our decision to refer to relative Sullivan minimal models -- minimal Koszul-Sullivan extensions in \cite{Halperin} -- as Koszul-Tate resolutions.}\end{rem}

It is now clear that the {\small KTR}-s in Algebra, Mathematical Physics, Cohomological Analysis, Homotopy Theory, and $\cD$-Geometry, have all roughly the same structure. In some areas specific assumptions reduce more or less strongly the size of the corresponding {\small KTR}. The difficulty is to switch between the different fields and respective languages (to establish a kind of dictionary) and to prove {\it precise} comparison results, such as, for instance, the result that, except for Tate's {\small KTR}, all the others are {\it rigorously} $\cD$-geometric ones.

\subsection{Cofibrant replacement KTR seen as $\cD$-geometric KTR}

\begin{prop} The cofibrant replacement Koszul-Tate resolution of a $\tt DG\cD A$-map $\zf:\cJ\to\cQ$, $\cJ\in\tt\cD A$, is a $\cD$-geometric Koszul-Tate resolution of $\zf\,.$\end{prop}

Since the notion of $\cD$-geometric resolution is a generalization of the notion of cofibrant replacement resolution to the case of an arbitrary smooth scheme, this proposition is rather obvious. Here is its precise proof.

\begin{proof} Let $\cJ\0\cS V$ be the cofibrant replacement resolution of a $\tt DG\cD A$-map $\zf:\cJ\to\cQ$, $\cJ\in\tt \cD A$. Since the underlying $X$ is a smooth affine variety, we replace the sheaves in Section \ref{KTRDGDG} by their total sections. The construction in Section 9 of \cite{BPP4} -- which leads to Theorem 28 of \cite{BPP4} -- directly implies that the minimal relative Sullivan $\cD$-algebra $\cJ\to \cJ\0\cS V$ is of Sullivan type. Indeed, $R:=\cJ\0\cS V$ is obtained as the union of a sequence $R_0\subset R_1\subset\ldots$ of differential graded $\cD$-algebras, where $R_k$ ($k\ge 0$) is defined by $R_k=R_{k-1}\0\cS G_k$ ($R_{-1}=\cJ$) and where $G_k$ is a free non-negatively graded $\cD$-module. Since the differential graded $\cD$-algebra structure on $R_{k-1}\0\cS G_k$ is obtained by means of Lemma \ref{LemRSA} in Subsection \ref{RSDA}, it is clear that the differential $\zd_k$ of $R_k$ satisfies $\zd_kG_k\subset R_{k-1}$. It now suffices to check that the $\tt DG\cD A$-trivial-fibration $q:\cJ\0\cS V\stackrel{\sim}{\twoheadrightarrow}\cQ$, which is also obtained by an iterated application of Lemma \ref{LemRSA}, is a $\tt DG\cJ[\cD]A$-map, i.e., that its source and target are objects in the latter category and that $q$ is $\cJ$-linear. In view of Example \ref{DGJDAEx}, the $\tt DG\cD A$-morphisms $j:\cJ\ni \iota \mapsto \iota\0 1_\cO\in \cJ\0\cS V$ and $\zf:\cJ\ni \iota\mapsto [\iota]\in \cQ$ endow the two target algebras $\cJ\0\cS V$ (with multiplication $\diamond$) and $\cQ$ (with multiplication $\ast$) with natural ${\tt DG} \cJ[\cD] \tt A$-structures $$\iota\triangleleft T=(\iota\0 1_\cO)\diamond T\quad \text{and}\quad \iota\triangleleft Q=[\iota]\ast Q\;.$$ As for the $\cJ$-linearity of $q$, we have $$q(\iota\triangleleft T)=q((\iota\0 1_\cO)\diamond T)=q(\iota\0 1_\cO)\ast q(T)=\zf(\iota)\ast q(T)=[\iota]\ast q(T)=\iota\triangleleft q(T)\;,$$ as, by construction, $q(\iota\0 1_\cO)=\zf(\iota)$. \end{proof}

\subsection{Change of perspective}\label{CoP}

Depending on the author(s), the concept of $\cD_X$-module is considered over a base space $X$ that is a {\it finite-dimensional smooth manifold} \cite{Cos} or a finite-dimensional complex manifold \cite{KS}, a smooth algebraic variety \cite{HTT} or a smooth scheme \cite{BD} over a fixed base field of characteristic zero.\medskip

In \cite{BPP4}, our base space is a smooth affine algebraic variety $X$. This enables us to replace sheaves by their {\it total sections} (which are much easier to handle) -- e.g., we substitute $\tt DG\cD A$, with $\cD\!=\!\cD_X(X)$, to ${\tt DG_+qcCAlg}(\cD_X)$. However, all the results that we obtain in \cite{BPP4} after the passage to total sections, are also valid for other underlying spaces $X$. Indeed, the only instance (after the passage), where we still use the nature of $X$, is the result that the $\cO$-module, $\cO\!=\!\cO_X(X)$, of linear differential operators $\cD\!=\!\cD_X(X)$ over a smooth affine algebraic variety $X$ is {\it flat} (and even projective \cite{BPP3}).\medskip

For the {\small KTR}-s in Mathematical Physics and in Cohomological Analysis of {\small PDE}-s, the space $X$ is an {\it $n$-dimensional smooth manifold}, and even an open subset $X\subset\R^n$, so that $\cD=\cD(X)$ is a free module over $\cO=\cO(X)$, hence a projective and a {\it flat} one. Moreover, the context for these {\small KTR}-s -- smooth geometry -- is usually presented in terms of {\it global sections} and morphisms between them \cite[Subsection 11.3]{BPP4}. It follows that:

\begin{rem}\label{RemTotSectSetting}{\em In the contexts of the {\small KTR}-s from Mathematical Physics and from Cohomological Analysis, total sections replace sheaves, $\cD$-modules can be used, and the results of \cite{BPP4} are valid. For instance, Lemma \ref{LemRSA} holds, the cofibrant replacement {\small KTR} makes sense, and so does the total-sections-version of the $\cD$-geometric {\small KTR}.}\end{rem}

We can thus try to show that the {\small KTR} of a regular first-order on-shell reducible gauge theory is a $\cD$-geometric {\small KTR}. The Koszul-Tate complex of such a theory, see Subsection \ref{KTRASF}, can be rewritten as $\op{KT}=\cF\otimes \cS {\frak V}$, where $\cF=\cF(\zp_\infty)$ and \be\label{DefKTCMathPhys}{\frak V}=\bigoplus_{\za,a}\R\cdot\zf^{\za*}_a\oplus\bigoplus_{\zb,\zd} \R\cdot C^{\zb*}_\zd\;,\ee and where the tensor products are over $\R$. The complex $(\op{KT},\zd_{\op{KT}})$ is thus a chain complex in the category of $\cF$-modules. \medskip

The algebra $\cF$ can be endowed with a $\cD$-module structure. Since we work in fixed coordinates, any $D\in\cD$ uniquely reads $D=\sum_{|\za|\le k}D_\za(x)\p_x^\za,$ for some integer $k\in\N$ and functions $D_\za\in\cO$. As observed in Equation (\ref{ActBaseNatActHor}) (and, maybe, partially in Equation (\ref{LDOCoeffFun})), the action of $D$ on $F\in\cF$ should be defined by $$D\cdot F=\cC(D)F=\sum_{|\za|\le k}D_\za(x)D_x^\za F\;,$$ where $\cC$ denotes the horizontal lift. It is easily seen that this definition actually provides a $\cD$-module structure, since, for any composable linear differential operators $\zD_1\in\op{Diff}(\zh_1,\zh_2)$ and $\zD_2\in\op{Diff}(\zh_2,\zh_3)$ between vector bundles $\zh_i$ over $X$, the horizontal lifts $$\cC(\zD_1)\in\cC\op{Diff}(\zp^*_\infty(\zh_1),\zp^*_\infty(\zh_2))\quad\text{and}\quad \cC(\zD_2)\in\cC\op{Diff}(\zp^*_\infty(\zh_2),\zp^*_\infty(\zh_3))$$ satisfy $$\cC(\zD_2\circ\zD_1)=\cC(\zD_2)\circ\cC(\zD_1)\;.$$ This result holds \cite{KV} for any vector bundles $\zp:E\to X$ and $\zh_i:F_i\to X$. For the trivial bundle $\zp:\R^n\times\R^r\to\R^n$ that we fixed at the beginning of Subsection \ref{RIGT} and for the trivial line bundle $\zh_i:\R^n\times\R\to \R^n$, we get $\op{Diff}(\zh_i,\zh_j)=\cD$ and $\cC\op{Diff}(\zp^*_\infty(\zh_i),\zp^*_\infty(\zh_j))=\cC\op{Diff}(\cF,\cF)$, i.e., we get the situation that we considered above.\medskip

It is clear that this $\cD$-module structure of $\cF$ and the $\cO$-algebra structure of $\cF$ are compatible in the sense that vector fields act as derivations. Hence, $\cF$ is a $\cD$-algebra. Moreover, the ideal $I(\zS)$ of those functions of $\cF$ that vanish on $\zS:D_x^\za\zd_{u^a}{\cal L}=0$, is an $\cO$-ideal and a $\cD$-submodule, hence a $\cD$-ideal. As for the submodule structure, note that if $F\in I(\zS)$ and $D\in\cD$, one has $$(D\cdot F)|_\zS=(\cC(D)F)|_\zS=\cC(D)_\zS\, F|_\zS=0\;,$$ see Corollary \ref{RestHor}, Subsection \ref{Vino}. Finally, the quotient $\Ci(\zS)=\cF/I(\zS)$ is a $\cD$-algebra for the action $D\cdot[F]=[D\cdot F]$ and the multiplication $[F][G]=[FG]$. It follows that the passage \be\label{Ingredients}\zf:\cF\ni F\mapsto [F]\in \Ci(\zS)\ee to the quotient is a $\cD$-algebra map. Example \ref{DGJDAEx} shows that the action $F\triangleleft [G]:= [F][G]=[FG]$ endows $\Ci(\zS)$ is an $\cF[\cD]$-algebra structure.

\begin{rem}{\em In view of Equation (\ref{Ingredients}) the algebra $\Ci(\zS)$ fits into the framework of Definition \ref{KTRDef} of a $\cD$-geometric {\small KTR}, as well as into the framework of Definition \ref{KTRDG} of a cofibrant replacement {\small KTR}.}\end{rem}

In Subsection \ref{ConsJetFunc}, we observed that the $\cD$-action on the fiber coordinates $x^{(k)}$ of an infinite jet space with base coordinate $t$ satisfies the equations $$\p_t\cdot x^{(k)}=D_t\,x^{(k)}=x^{(k+1)}\;.$$ In Subsection \ref{KTRASF}, we viewed the degree 1 generators $\zf^{\za*}_a$ (resp., the degree 2 generators $C^{\zb*}_\zd$) as fiber coordinates of an infinite horizontal jet space with base coordinates $(x^i,u_\za^a)$ and we noticed that this interpretation comes along with the replacement of the total derivatives $D_{x^i}$ by the extended total derivatives $\bar D_{x^i}\,$. It is therefore natural to define the $\cD$-action on the fiber coordinates $\zf^{\za*}_a$ (resp., $C^{\zb*}_\zd$) by $$\p_{x^i}\cdot \zf^{\za*}_a:=\bar D_{x^i}\zf^{\za*}_a=\zf^{i\za*}_a\;$$ (resp., by $$\p_{x^i}\cdot C^{\zb*}_\zd:=\bar D_{x^i}C^{\zb*}_\zd=C^{i\zb*}_\zd)\;.$$ In particular, we obtain \be\label{Fundamental} \p_x^\za\cdot\zf_a^*=\bar D_x^\za\zf^*_a=\zf^{\za*}_a\quad(\text{resp.,}\;\p_x^\zb\cdot C_\zd^*=\bar D_x^\zb C^*_\zd=C^{\zb*}_\zd)\;.\ee

Eventually, it is natural to replace the underlying module $\frak V$ of Equation (\ref{DefKTCMathPhys}) by the free non-negatively graded $\cD$-module \be\label{CKTChains0}V=\bigoplus_{a}\cD\cdot\zf^{*}_a\oplus\bigoplus_{\zd} \cD\cdot C^{*}_\zd\;\ee over the components of the antifields $\zf^*$ and $C^*$. The $\cF$-module of Koszul-Tate chains then reads \be\label{CKTChains}\op{KT}=\cF\0_\R\cS_\R{\frak V}=\cF\0_\cO\cS_\cO V\;,\ee where the {\small RHS} is also a graded $\cD$-algebra.\medskip

Any element $c$ of this graded $\cD$-algebra reads non-uniquely as a finite sum $$c=\sum F\, (D^a\cdot\zf^*_a) \ldots (\zD^\zd\cdot C^*_\zd)\;,$$ where $F\in\cF$ and $D^a,\zD^\zd\in\cD$, and where we omitted the tensor products. The Koszul-Tate differential $\zd_{\op{KT}}$, which is well-defined on $\op{KT}$, acts as a graded derivation and {\it is thus completely known, if it is known on the $D^a\cdot\zf^*_a$ and the $\zD^\zd\cdot C^*_\zd\,$}. For any $D=D_\za\p_x^\za$, we have, in view of the definitions given above, \be\label{KTDLin1}\zd_{\op{KT}}(D\cdot\zf^*_a)=D_\za\,\zd_{\op{KT}}(\p_x^\za\cdot \zf^*_a)=D_\za\,\zd_{\op{KT}}(\zf^{\za*}_a)=D_\za D_x^\za\zd_{u^a}{\cal L}=D\cdot(\zd_{u^a}{\cal L})=D\cdot \zd_{\op{KT}}(\zf^*_a)\;.\ee Similarly, we get $$\zd_{\op{KT}}(D\cdot C^*_\zd)=D_\za\,\zd_{\op{KT}}(\p_x^\za\cdot C^*_\zd)=D_\za\,\zd_{\op{KT}}(C^{\za*}_\zd)=D_\za \bar D_x^\za(R^a_{\zd\zb}\, \bar D_x^\zb\zf^*_a)=D_\za \bar D_x^\za(R^a_{\zd\zb}\, \zf^{\zb*}_a)\;.$$ The extended total derivative $\bar D_x^\za$ of $R^a_{\zd\zb}\, \zf^{\zb*}_a$ is a sum of terms of the type $$D_x^{\za_1}R^a_{\zd\zb}\; \bar D_x^{\za_2}\zf^{\zb*}_a=(\p_x^{\za_1}\cdot R^a_{\zd\zb})\; (\p_x^{\za_2}\cdot\zf^{\zb*}_a)\;,$$ so that, in view of the definition of the $\cD$-action on the tensor product of $\cF$ and $\cS_\cO V$, we find $$\bar D_x^\za(R^a_{\zd\zb}\, \zf^{\zb*}_a)=\p_x^\za\cdot(R^a_{\zd\zb}\, \zf^{\zb*}_a)\;.$$ Eventually, \be\label{KTDLin2}\zd_{\op{KT}}(D\cdot C^*_\zd)=D\cdot\zd_{\op{KT}}(C^*_\zd)\;.\ee

\begin{rem}{\em The equations (\ref{CKTChains}), (\ref{KTDLin1}), and (\ref{KTDLin2}) show that $\op{KT}$ is a graded $\cD$-algebra and that $(\op{KT},\zd_{\op{KT}})$ is a chain complex in the category of $\cD$-modules.}\end{rem}

\subsection{KTR of a reducible theory seen as $\cD$-geometric KTR}\label{CKTDKT}

In the following, we apply Lemma \ref{LemRSA} from Subsection \ref{RSDA}, which allows to construct non-split relative Sullivan $\cD$-algebras ({\small RS$\cD$A}-s), as well as ${\tt DG\cD A}$-morphisms from such a Sullivan algebra to another differential graded $\cD$-algebra.\medskip

Let $V_1:=\bigoplus_a \cD\cdot\zf^*_a\,$. To endow the graded $\cD$-algebra \be\cC_1:=\cF\0_\cO\cS_\cO V_1\label{C0}\ee with a differential graded $\cD$-algebra structure $d$, we set, \be\label{dLow1}d\zf^*_a:=\zd_{u^a}{\cal L}\in\cF\;,\ee extend $d$ to $V_1$ by $\cD$-linearity, and equip $\cC_1$ with the differential $d$ given by $$d(F (D\cdot \zf^*_a)\, (\zD\cdot \zf^*_b)):=(F\,d(D\cdot\zf^*_a))(\zD\cdot\zf^*_b)-(F\,d(\zD\cdot\zf^*_b))(D\cdot\zf^*_a)\;,$$ where we omitted the tensor products and considered, to increase clarity, an element of degree 2. Then the natural ${\tt DG\cD A}$-morphism $\imath:(\cF,0)\ni F\mapsto F\0 1_\cO\in(\cC_1,d)$ is a {\small RS$\cD$A}. Since $\zd_{\op{KT}}$ is also a graded derivation that is $\cD$-linear (Equation (\ref{KTDLin1})) and coincides with $d$ on the generators $\zf^*_a$, the {\small RS$\cD$A} is actually the ${\tt DG\cD A}$-morphism \be\label{imath}\imath:(\cF,0)\ni F\mapsto F\0 1_\cO\in(\cC_1,\zd_{\op{KT}})\;.\ee

Consider now the $\cD$-algebra $\Ci(\zS)=\cF/I(\zS)$ and the ${\tt \cD A}$-morphism $\zf:\cF\to\Ci(\zS)$ (Equation (\ref{Ingredients})). To define a ${\tt DG\cD A}$-morphism \be\label{q0}q_1:\cC_1\to \Ci(\zS)\;,\ee it suffices to set \be\label{Cond} q_1(\zf^*_a)=0\in (\Ci(\zS))_1\cap 0^{-1}(\zf(d\zf^*_a))\;,\ee to extend $q_1$ by $\cD$-linearity to $V_1$, and to define $q_1$ in degree 0 by $q_1(F)=\zf(F)=[F]$ and in degree $\ge 1$ by $q_1=0$. As for Condition (\ref{Cond}), note that $\zf(d\zf^*_a)=[\zd_{u^a}{\cal L}]=0$, in view of the definition of $\zS$.\medskip

An anew application of Lemma \ref{LemRSA}, where the role that was played above by $(\cF,0)$ (resp., $V_1$) is now assumed by $(\cC_1,\zd_{\op{KT}})$ (resp., $V_2:=\bigoplus_\zd \cD\cdot C^*_\zd$), endows the graded $\cD$-algebra \be\label{C1}\cC_2:=\cC_1\0_\cO\cS_\cO V_2\ee with a differential graded $\cD$-algebra structure $\dd$ that, similar to $d$ above, is fully defined by \be\label{dLow2}\dd C^*_\zd= R^a_{\zd\za} (\p_x^\za\cdot\zf^*_a)\in(\cC_1)_1\cap \zd^{-1}_{\op{KT}}\{0\}\;.\ee Indeed, in view of Equation (\ref{NI}), we have $$\zd_{\op{KT}}(R^a_{\zd\za} (\p_x^\za\cdot\zf^*_a))=R^a_{\zd\za}\,D_x^\za\zd_{u^a}{\cal L}\equiv 0\;.$$ To compare the differential $\dd$ with the differential $\zd_{\op{KT}}$, note that $\dd$ is extended to $V_2$ by $\cD$-linearity and that its value on $c=F (D\cdot \zf^*_a)\,(\zD\cdot C^*_\zd)\,(\nabla\cdot C^*_\ze)$, for instance, is \begin{eqnarray*}\dd c=&\zd_{\op{KT}}(F (D\cdot\zf^*_a))\;(\zD\cdot C^*_\zd)\,(\nabla\cdot C^*_\ze)\\&-(F (D\cdot \zf^*_a)\,\dd(\zD\cdot C^*_\zd))\;(\nabla\cdot C^*_\ze)\\&-\left(F (D\cdot \zf^*_a)\,\dd(\nabla\cdot C^*_\ze)\right)\;(\zD\cdot C^*_\zd)\;.\end{eqnarray*} As $\zd_{\op{KT}}$ is a graded derivation that is $\cD$-linear (Equation (\ref{KTDLin2})) and coincides with $\dd$ on the generators $C^*_\zd$, we get $\dd=\zd_{\op{KT}}$ on $\cC_2$. Hence, the ${\tt DG\cD A}$-morphism \be\label{jmath}\jmath:(\cC_1,\zd_{\op{KT}})\ni c\mapsto c\0 1_\cO\in(\cC_2,\zd_{\op{KT}})\ee is a relative Sullivan $\cD$-algebra.\medskip

Start now from the $\tt DG\cD A$-morphism $q_1$, and define a $\tt DG\cD A$-morphism \be\label{q1}q_2:\cC_2\to\Ci(\zS)\ee by setting $$q_2(C^*_\zd)=0\in (\Ci(\zS))_2\cap 0^{-1}(q_1(\zd_{\op{KT}}\,C^*_\zd))\;,$$ extending $q_2$ by $\cD$-linearity to $V_2$, and by defining $q_2$ in degree 0 by $q_2(F)=[F]$ and in degree $\ge 1$ by $q_2=0$. \medskip

Since $V=V_1\oplus V_2$ as graded $\cD$-module, the graded $\cD$-algebras $\cS_\cO V=\cS_\cO(V_1\oplus V_2)$ and $\cS_\cO V_1\0_\cO\cS_\cO V_2$ are isomorphic. Hence, the same holds for the graded $\cD$-algebras $$\op{KT}=\cF\0_\cO\cS_\cO V\quad\text{and}\quad \cC_2=\cF\0_\cO\cS_\cO V_1\0_\cO \cS_\cO V_2\;.$$ It follows that $\jmath\circ\imath:(\cF,0)\to (\op{KT},\zd_{\op{KT}})$ is a $\tt DG\cD A$-morphism and thus allows to endow $(\op{KT},\zd_{\op{KT}})$ with a {\small DG$\cF[\cD]$A}-structure -- see Example \ref{DGJDAEx}.

\begin{theo}\label{ClassDGeom} The Koszul-Tate resolution of the function algebra $\Ci(\zS)$ of the infinite prolongation manifold $\zS$ of the Euler-Lagrange equations of a regular first-order on-shell reducible gauge theory is a $\cD$-geometric Koszul-Tate resolution $(\,$in the smooth setting -- see beginning of Subsection \ref{CoP}$\,)$ of the canonical $\cD$-algebra map $\cF\to\Ci(\zS)$, where $\cF$ is the function algebra of the infinite jet space in which $\zS$ is located and where $\Ci(\zS)$ is the quotient of $\cF$ by the ideal of those functions of $\cF$ that vanish on $\zS$.\end{theo}

\begin{proof} Most of the proof is given in the preparation that precedes the theorem. For instance, it is clear from what has been said that $\op{KT}\simeq\cC_2$ admits an increasing filtration $\cC_1\subset\cC_2\subset \cC_2\subset\ldots$ by {\small DG} $\cD$-subalgebras, such that there is a {\small DG} $\cD$-algebra morphism $\cF\to\cC_1$ $(\,$we set $\cC_{0}:=\cF\,)$ and that $\cC_k$ $(\,k\ge 1\,)$ is isomorphic as {\small DG} $\cD$-algebra to $\cC_k\simeq \cC_{k-1}\0_\cO\cS_\cO V_k,$ where $V_k$ is a free graded $\cD$-submodule of $\cC_k$ such that $\zd_{\op{KT}}V_k\subset\cC_{k-1}:$ $\op{KT}$ is of Sullivan type. We already mentioned that $\op{KT}\simeq\cC_2$ and $\Ci(\zS)$ are {\small DG}$\cF[\cD]$-algebras. It now suffices to show that the ${\tt DG\cD A}$-morphism $q:=q_2:\op{KT}\to \Ci(\zS)$ is $\cF$-linear and induces an $\cF$- and $\cD$-linear bijection $q_\sharp$ of degree 0 between the graded module $H_\bullet(\op{KT})$ and the module $\Ci(\zS)$ concentrated in degree 0. First, $q$ is $\cF$-linear, as, if $F,G\in\cF$, we obtain $$F\triangleleft q(G)=F\triangleleft [G]=[FG]=q(FG)\;.$$ Hence, the induced map $q_\sharp$ has the required properties, except, maybe, bijectivity. In degree $\ge 1$, the homology $H_\bullet(\op{KT})$ vanishes, just as $\Ci(\zS)$. In degree $0$, the homology is given by $\Ci(\zS)=\cF/I(\zS)$, where $\cF$ (resp., $I(\zS)$) are the 0-cycles (resp., $0$-boundaries), and $q_\sharp[F]=q(F)=[F]$ is the identity.\end{proof}

\subsection{KTR of a reducible theory versus cofibrant replacement KTR}

Recall first that, in the setting of a {\small KTR} from Mathematical Physics, the concept of cofibrant replacement Koszul-Tate resolution makes sense. Secondly, it is clear a priori that the general functorial cofibrant replacement {\small KT} resolution $(\op{\cal KT},\zd_{\op{\cal KT}})$ is much larger than the {\small KT} resolution $(\op{KT},\zd_{\op{KT}})$, which is subject to size-reducing irreducibility (i.e., first-order reducibility) conditions and is far from being functorial.\medskip

More precisely, the {\small KT} resolution $(\op{KT},\zd_{\op{KT}})$ is the {\small DG$\cF[\cD]$A} $$\op{KT}=\cF\0_\cO\cS_\cO V\;,$$ where $V$ is the free graded $\cD$-module with homogeneous basis $$\bigcup\,\{\zf^*_a,C^*_\zd\}\;$$ (the degrees of the generators are $1,\,2$), endowed with the degree $-1$, $\cF$- and $\cD$-linear graded derivation defined by $$\zd_{\op{KT}}(\zf^*_a)=\zd_{u^a}{\cal L}\quad\text{and}\quad\zd_{\op{KT}}(C^*_\zd)=R^a_{\zd\za}\,(\p_x^\za\cdot \zf^*_a)\;.$$ The results of \cite{BPP4}, applied to the ${\tt DG\cD A}$-map $\zf:(\cF,0)\to(\Ci(\zS),0)$, show that the cofibrant replacement {\small KT} resolution $(\op{{\cal KT}},\zd_{\op{\cal KT}})$ is the {\small DG$\cF[\cD]$A} $$\op{\cal KT}=\cF\0_\cO\cS_\cO {\cal V}\;,$$ where ${\cal V}$ is the free graded $\cD$-module with homogeneous basis $$\bigcup\,\{\mbi_{f},\mbi^1_{\zs_n,0},\mbi^2_{\zs_n,0},\ldots,\mbi^k_{\zs_n,0},\ldots\}\;,$$ for all $f\in\Ci(\zS)$ and `numerous' $\zs_n$ (of degree $n\ge 0$), which are described in \cite[Theorem 28]{BPP4} and in the proof that precedes this result (the degrees of the generators are $0,\,n+1,\,n+1,\ldots, n+1,\ldots\,$). Here $\zd_{\op{\cal KT}}$ is the degree $-1$, $\cF$- and $\cD$-linear graded derivation defined by $$\zd_{\op{\cal KT}}(\mbi_{f})=0\quad\text{and}\quad\zd_{\op{\cal KT}}(\mbi^k_{\zs_n,0})=\zs_n\;.$$ When using the just mentioned description in \cite[Theorem 28]{BPP4}, one sees quite easily that the injective map $i$, defined by $$i(\zf^*_a)=\mbi^1_{(\zd_{u^a}{\cal L},\,0)}\in{\cal V}_1\quad\text{and}\quad i(C^*_\zd)=\mbi^2_{\left(R^a_{\zd\za}\left(\p_x^\za\cdot\, \mbi^1_{(\zd_{u^a}{\cal L},\,0)}\right),\;0\right)}\in{\cal V}_2\;,$$ is a $\tt DG\cF[\cD]A$-morphism $$i:(\op{KT},\zd_{\op{KT}})\to (\op{\cal KT},\zd_{\op{\cal KT}})\;.$$

\begin{prop} The Koszul-Tate resolution of the function algebra $\Ci(\zS)$ of the infinite prolongation manifold $\zS$ of the Euler-Lagrange equations of a regular first-order reducible gauge theory is a differential graded $\cF[\cD]$-subalgebra of the cofibrant replacement Koszul-Tate resolution $(\,$in the smooth setting -- see beginning of Subsection \ref{CoP}$\,)$ of the quotient $\cD$-algebra $\Ci(\zS)\,$.\end{prop}

\subsection{KTR of a reducible theory versus KTR in Cohomological Analysis}

We compare the Koszul-Tate complex $(\op{KT},\zd_{\op{KT}})$ of a regular first-order on-shell reducible field theory, which is defined in coordinates, with the Koszul-Tate complex $(\op{\frak{KT}},\zd_{\op{\frak{KT}}})$ of Subsection \ref{KTRCC}, which is subject to regularity and higher-order off-shell reducibility conditions, and is -- although fixed coordinates are considered -- mostly defined in the coordinate-free language of Cohomological Analysis of {\small PDE}-s. The difficulty is to pass from one setting to the other. Let us stress that in the following $\op{KT}$ and $\frak{KT}$ refer to these two different complexes, and let us mention that this section might be easier to read after a revision of Section \ref{CCKTR} and of parts of Appendix A, Section \ref{Vino}. \medskip

In the contexts of $\op{KT}$ and $\frak{KT}$ the underlying space is an open subset $X\subset\R^n$. We thus have $\cO=\Ci(X)$ and $\cD=\cD(X)$.\medskip

The chain complex $(\frak{KT},\zd_{\frak{KT}})$ is defined from a compatibility complex $$ 0\longrightarrow\cR_1\stackrel{\zD_1}{\longrightarrow}\cR_2\stackrel{\zD_2}{\longrightarrow}\ldots\stackrel{\zD_{k-2}}{\longrightarrow}\cR_{k-1}\longrightarrow 0\;$$ made of $\cF$-modules $\cR_j:=\zG(R_j):=\zG(\zp_\infty^*(F_j))$ -- here $\zp_\infty:J^\infty E\to X$ is the infinite jet space of $\zp:E\to X$, a rank $r$ smooth vector bundle over $X$, $\cF$ is the function algebra of $J^\infty E,$ and $\zr_j:F_j\to X$ is a rank $r_j$ smooth vector bundle over $X$ -- and of horizontal differential operators $\zD_j:\cR_j\to\cR_{j+1}$ between them. The Koszul-Tate chains $\frak{KT}$ are the elements of the algebra $\,\cS_\cF\,\cC\!\op{Diff}(\cR_\bullet,\cF)$, where $$\cR_{\bullet}:=\zG(R_\bullet):=\zG(\zp^*_\infty(F_\bullet))$$ and where $\cR_\bullet$ (resp., $R_\bullet$, $F_\bullet$) is the direct sum of the $\cR_j$ (resp., $R_j$, $F_j$). Since $$\cC:\cF\0_\cO\Diff(\zG(F_\bullet),\cO)\to \cC\!\Diff(\cR_\bullet,\cF)$$ is an $\cF$-module isomorphism (Equation (\ref{LDOCoeffGen})), we get $$\op{\frak{KT}}\simeq\cS_\cF\left(\cF\0_\cO\Diff(\zG(F_\bullet),\cO)\right)\simeq \cF\0_\cO\cS_\cO\Diff(\zG(F_\bullet),\cO)\;.$$

As already mentioned, we work in fixed coordinates. The coordinates of $E$ are denoted by $(x^i,u^a)$ and those of $J^\infty E$ by $(x^i,u^a_\za)$. Similarly, we symbolize the coordinates of $F_\bullet$ by $(x^i,v^\zl(j))$ -- where $j\in\{1,\ldots,k-1\}$ and $\zl\in\{1,\ldots,r_j\}$ -- , those of $R_\bullet$ by $(x^i,u^a_\za,v^\zl(j))$, and those of $\bar J^\infty(R_\bullet)$ by $(x^i,u^a_\za,v_\zb^\zl(j))$. Hence, a linear differential operator ${\op{\sf D}}\in\Diff(\zG(F_\bullet),\cO)$, when applied to a section $v\in\zG(F_\bullet)$, reads $$\op{\sf D}v=\sum_\za (D^1_\za(x)\ldots D^{\sum_j r_j}_\za(x))\p_x^\za\;\left(\begin{array}{c}
           \vdots \\
           v^\zl(j)(x^i) \\
           \vdots
         \end{array}\right)\;,$$
so that it is natural to view it as an element of the free non-negatively graded $\cD$-module \be\label{CCKTChains0}V:=\bigoplus_{j=1}^{k-1}\bigoplus_{\zl=1}^{r_j} \cD\cdot v^\zl(j)\ee over formal generators of degree $j$, which we also denote by $v^\zl(j)$. Hence, we get the $\cF$-module isomorphism \be\label{CCKTChains}\op{\frak{KT}}\simeq \cF\0_\cO\cS_\cO V\;,\ee where the {\small RHS} is also a graded $\cD$-algebra.\medskip

The comparison of Equations (\ref{CCKTChains0}) and (\ref{CCKTChains}) with Equations (\ref{CKTChains0}) and (\ref{CKTChains}) shows that the algebras $\op{KT}$ and $\frak{KT}$ are defined similarly. More precisely:

\begin{rem}\label{KTR1stHigher}{\em Whereas the complex $\op{KT}$ contains the antifields $\zf^*$ and $C^*$ -- with components $\zf^*_a$ and $C^*_\zd$ that correspond to the considered equations $\zd_{u^a}{\cal L}(x^i,u^a_\za)$ and the irreducible relations $$R^a_{\zd\za}D^\za_x\zd_{u^a}{\cal L}(x^i,u^a_\za)\equiv 0$$ between them -- , the complex $\op{\frak{KT}}$ contains antifields $v(1),v(2),v(3)$, ... -- whose components $v^\zl(1),v^\zl(2),v^\zl(3)$, ... correspond to the equations $\psi_D\in\cR_1$, i.e., the equations $\psi_D^\zl(1)(x^i,u^a_\za)$, the reducible relations $\zD_1(\psi_D)=0$ between them, i.e., the relations $$(\zD_1(\psi_D))^\zl(2)(x^i,u^a_\za)\equiv 0\;,$$ the relations $\zD_2\circ\zD_1=0$ between these relations, ... -- .}\end{rem}

To further compare $\op{KT}$ and $\frak{KT}$, we must of course use here the same basic definitions as in Subsection \ref{CoP}. Hence, in analogy with (\ref{Fundamental}), we set \be\label{Fundamental2} \p_x^\zb\cdot v^\zl(j):=\bar D_x^\zb v^\zl(j)=v^\zl_\zb(j)\;,\ee where \be\bar D_{x^i}=\p_{x^i}+u^a_{i\za}\p_{u^a_\za}+v^\zl_{i\zb}(j)\p_{v^\zl_\zb(j)}\:.\ee

We are now prepared to compare the Koszul-Tate differentials $\zd_{\op{KT}}$ and $\zd_{\op{\frak{KT}}}$. As mentioned in Subsection \ref{KTRCC}, the differential $\zd_{\op{\frak{KT}}}$ is completely defined by its values on $$\cC\!\op{Diff}(\cR_\bullet,\cF)\simeq \h_\cF(\bar{\cal J}^\infty(\cR_\bullet),\cF)\simeq {\cal P}\!\op{ol}^1(\bar J^\infty(R_\bullet))$$ and its values on $\cF$. Here superscript 1 refers to functions that are linear in the fiber coordinates $v^\zl_\zb(j)$. To simplify the notation and to nevertheless distinguish the sections $v^\zl(j)(x^i,u^a_\za)$ of $R_\bullet$ (resp., the sections $v^\zl_\zb(j)(x^i,u^a_\za)$ of $\bar J^\infty(R_\bullet)$) from the fiber coordinates $v^\zl(j)$ of $R_\bullet$ (resp., the fiber coordinates $v^\zl_\zb(j)$ of $\bar J^\infty(R_\bullet)$), we write $\tilde v^\zl(j)$ (resp., $\tilde v^\zl_\zb(j)$) for sections. In the considered fixed coordinates, the preceding identifications read, i.e., such a differential operator $\nabla$ and the corresponding linear jet space function ${\frak F}_\nabla$ read (with obvious notation) $$\nabla v=\sum_\zb (\ldots\nabla^\zl_\zb(j)(x^i,u^a_\za)\ldots) D_x^\zb\,\left(\begin{array}{c}
           \vdots \\
           \tilde v^\zl(j) \\
           \vdots
         \end{array}\right)\simeq $$ 
         \be\label{Identifications} {\frak F}_\nabla(x^i,u^a_\za,v^\zl_\zb(j)) =\sum_\zb (\ldots\nabla^\zl_\zb(j)(x^i,u^a_\za)\ldots)\,\left(\begin{array}{c}
           \vdots \\
           v_\zb^\zl(j) \\
           \vdots
         \end{array}\right)
\;.\ee Since $\zd_{\op{\frak{KT}}}$ vanishes on $\cF$, it is completely defined by its values on the $v^\zl_\zb(j)$, exactly as $\zd_{\op{KT}}$ is fully defined by its values on the $\zf^{\za*}_a$ and the $C^{\zb*}_\zd$. Note still, before proceeding, that, for horizontal linear differential operators $\cC\!\op{Diff}(\cR_j,\cR_{j+1})$ valued in a not necessarily rank 1 bundle, the identifications (\ref{Identifications}) are exactly the same, except that the row of coefficients $\nabla^\zl_\zb(j)$ is replaced by a matrix of coefficients $\nabla^{\zm\zl}_\zb(j+1,j)$.\medskip

Recall now from Subsection \ref{KTRCC} that, if $F\in\cF$ and $\nabla_j\in\cC\!\op{Diff}(\cR_j,\cF)$, we have \be\label{DefKTRDiffCA}\zd_{\frak{KT}}(F)=0\,,\quad\zd_{\frak{KT}}(\nabla_1)=\nabla_1(\psi_D)\,,\quad\text{and}\quad \zd_{\frak{KT}}(\nabla_j)=\nabla_j\circ\zD_{j-1}\,,\quad \forall j\ge 2\;.\ee The equations (\ref{Identifications}) and (\ref{DefKTRDiffCA}) lead to the equation \be\label{CCKTDiff1}\zd_{\op{\frak{KT}}}(v^\zl_\zb(1))=\zd_{\op{\frak{KT}}}(D_x^\zb \tilde v^\zl(1))=D_x^\zb (\psi_D^\zl(1))\;\ee -- which is entirely similar to the definition \be\label{1stKTRDiff1}\zd_{\op{KT}}(\zf^{\za*}_a)=D_x^\za (\zd_{u^a}{\cal L})\;.\ee For $j\in\{2,\ldots,k-1\}$, we find analogously $$\zd_{\op{\frak{KT}}}(v^\zl_\zb(j))=\zd_{\op{\frak{KT}}}(\,D_x^\zb \tilde v^\zl(j))=D_x^\zb\left((\zD_{j-1}\,\tilde v(j-1))^\zl(j)\right)=$$ $$D_x^\zb\left(\,(\zD^{\zl\zm}_{\zg}(j,j-1))(x^i,u^a_\za)\,\;D_x^\zg\, \tilde v^\zm(j-1)\,\right)\;,$$ in view of the above remark on matrix coefficients. 
When using again the identification (\ref{Identifications}), we finally get
$$\label{CCKTDiff2}\zd_{\op{\frak{KT}}}(v^\zl_\zb(j))= \bar D_x^\zb\left(\,(\zD^{\zl\zm}_{\zg}(j,j-1))(x^i,u^a_\za)\,\;v_\zg^\zm(j-1)\,\right)=\bar D_x^\zb \left({\frak F}^\zl_{\zD_{j-1}}\right)\;.$$ For $j=2$, we thus find the equation \be\label{CCKTDiffP2}\zd_{\op{\frak{KT}}}(v^\zl_\zb(2))=\bar D_x^\zb\left(\,\zD^{\zl\zm}_{\zg}(2,1)\,\;\bar D_x^\zg\, v^\zm(1)\,\right)\;,\ee where we omitted the variables $(x^i,u^a_\za)$ -- which is fully analogous to the definition \be\label{1stKTRDiff2}\zd_{\op{KT}}(C^{\zb*}_\zd)=\bar D_x^\zb\left(R^\zm_{\zd\zg}\bar D_x^\zg \zf^*_\zm\right)\;.\ee

We conclude with the observation that the Koszul-Tate differential $$\zd_{\op{\frak{KT}}}=\sum_{\zb \zl}\,\bar D_x^\zb \left(\psi_D^\zl\right)\;\p_{v^\zl_\zb(1)}+\sum_{j=2}^{k-1}\sum_{\zb \zl}\,\bar D_x^\zb \left({\frak F}^\zl_{\zD_{j-1}}\right)\,\p_{v^\zl_\zb(j)}$$ is the evolutionary vector field, or symmetry of the Cartan distribution, that is obtained as the prolongation $\zd_{\cal X}$ to the horizontal jet space $\bar J^\infty(R_\bullet)\to J^\infty E$ of the vertical vector field $${\cal X}=\sum_\zl\,\psi_D^\zl\;\p_{v^\zl(1)}+\sum_{j=2}^{k-1}\sum_\zl\,{\frak F}^\zl_{\zD_{j-1}}\;\p_{v^\zl(j)}$$ of the bundle $R_\bullet\to J^\infty E$ with coefficients in $\cF(\bar J^\infty(R_\bullet)),$ see Equation (\ref{STh}).\medskip

Remark \ref{KTR1stHigher}, Equations (\ref{CCKTChains0}), (\ref{CCKTChains}), (\ref{CKTChains0}), and (\ref{CKTChains}), as well as Equations (\ref{CCKTDiff1}), (\ref{CCKTDiffP2}), (\ref{1stKTRDiff1}), and (\ref{1stKTRDiff2}), show that:

\begin{rem}\label{KTRCA1stRed} The {\small KTR} in Cohomological Analysis \cite{Ver} is the natural extension of the {\small KTR} of a first-order reducible field theory and it thus corresponds exactly to the {\small KTR} of a higher-order reducible field theory \cite{HT}.\end{rem}

\subsection{KTR in Cohomological Analysis seen as $\cD$-geometric KTR}

It is clear that, since the {\small KTR} in a first-order reducible theory is a $\cD$-geometric {\small KTR} (Theorem \ref{ClassDGeom}), the natural extension of this {\small KTR} is $\cD$-geometric as well. In view of Remark \ref{KTRCA1stRed}, we thus have the

\begin{theo} The Koszul-Tate resolution of $\Ci(\zS)$ from Cohomological Analysis of {\small PDE}-s is a $\cD$-geometric Koszul-Tate resolution of the $\cD$-algebra map $\cF\to\Ci(\zS)$ $(\,$in the smooth setting -- see Remark \ref{RemTotSectSetting}$\,)$, where $\cF$ is the function algebra of the infinite jet space in which $\zS$ is located.\end{theo}

\begin{proof} The proof is similar to the proof of Theorem \ref{ClassDGeom}.\end{proof}

\section{Appendix A: Partial differential equations in the jet bundle formalism}\label{Vino}

The goal of the present section is to explain a number of concepts that are of importance in the Geometry of {\small PDE}-s. Additional details can be found, for instance, in \cite{KV}.

\subsection{Jets and differential operators}\label{JetDiff}

Consider a differential equation ({\small DE}) \be\label{PDE}\psi(t,\zf(t),d_t\zf,\ldots, d_t^k\zf)\equiv 0\;,\ee with evident notation. When defining the $k$-jet of $\zf(t)$ by $$j^k_t\zf=(t,\zf(t),d_t\zf,\ldots, d_t^k\zf)\;,$$ {\bf we may rewrite this {\small DE} as} \be\label{Basic1}\psi(t,u,u_1,\ldots,u_k)|_{j^k_t\zf}\equiv 0\;.\ee Here $(t,u,u_1,\ldots,u_k)$ are independent variables of what is called the $k$-jet space. Roughly speaking, the (purely) algebraic equation \be\label{AlgEq}\psi(t,u,u_1,\ldots,u_k)=0\ee defines a hypersurface $\zS^0$ in the $k$-jet space (or, better, since $t$ plays a distinguished role, a subbundle $\zS^0$ of the $k$-jet bundle), and a {\it solution} of the considered {\small DE} is nothing but a function $\zf(t)$ such that the graph\footnote{Usually the $k$-jet is defined by $j^k_t\zf=(\zf(t),d_t\zf,\ldots, d_t^k\zf)$, so that `graph' is actually the proper denomination. In view of our modified definition, `graph' means in this text `image'.} of its $k$-jet is located on $\zS^0$. This is one of the key-aspects of the jet bundle approach to partial differential equations ({\small PDE}-s) -- which will be formalized in the following. \medskip

Let $\zp:E\to X$ be a smooth vector bundle of rank $\op{rk}(\zp)=r$ over a smooth $n$-dimensional manifold. For $k\in\N\,,$ the {\bf $k$-jet} $j^k_m\zf$ at $m\in X$ of a local smooth section $\zf\in\zG(\zp)$ of $\zp$ that is defined around $m$, is the equivalence class of all local sections of $\zp$, such that in any trivializing chart $(x,u)=(x^i,u^a)$ of $\zp$ around $m$, the local coordinates of these sections coincide at $x(m)$, together with their partial derivatives at $x(m)$ up to order $k$ (it actually suffices that they coincide in one trivializing chart). We define the $k$-jet set $J^k(\zp)$ of $\pi$ by $$J^k(\zp)=\{j^k_m\zf: m\in X, \zf\in\zG(\zp)\}\;.$$ The $k$-jet set is a smooth finite rank vector bundle $\zp_k:J^k(\zp)\to X$ -- the {\bf $k$-jet bundle}. Indeed, any trivializing chart $(x^i,u^a)$ of $\pi$ induces a trivializing chart $(x^i,u_\za^a)$ of $\zp_k$, defined by $$x^i(j^k_m\zf)=x^i(m)\quad\text{and}\quad u_\za^a(j^k_m\zf)=\partial_x^\za \zf^a|_{x(m)}\;,$$ where $\za\in\N^n$ and $|\za|\le k$. For $k\le\ell$, there is a `truncation' vector bundle (epi)morphism $\pi_{k\ell}:J^\ell(\pi)\to J^k(\pi)$, so that $(J^k(\zp),\zp_{k\ell})$ is an inverse system. The limit of this diagram is the {\bf $\infty$-jet space} $\zp_\infty:J^\infty(\pi)\to X$ together with the natural projections $\zp_{k\,\infty}:J^\infty(\zp)\to J^k(\zp)$. Coordinates $(x^i,u^a_\za)$ of $J^\infty(\zp)$ can be obtained from coordinates $(x^i,u^a)$ of $\zp$, as above, by defining an infinite number of coordinates $u^a_\za$ that correspond to the partial derivatives $\p_x^\za$ of the components $\zf^a=u^a(\zf(x))$ of the sections $\zf$ of $\zp\,.$ We denote the {\bf algebra of smooth functions} of $J^k(\pi)$ by $\cF_k=\cF_k(\pi)$. The canonical epimorphisms $\zp_{k\ell}$ induce inclusions $\cF_k\subset \cF_\ell\,$. The colimit of this direct system is the algebra $\cF=\bigcup_k\cF_k$ (we will also write $\cF(\zp),$ $\cF_\infty,$ or $\cF_\infty(\zp)$) of smooth functions of $J^\infty(\zp)$. It follows that any smooth function of $J^\infty(\pi)$ is a smooth function of some $J^k(\zp)$. Note eventually that $j^k:\zG(\zp)\to \zG(\zp_k)$ and that $j^\infty:\zG(\zp)\to \zG(\zp_\infty)$ (in fact, we should, as above, consider the case $k=\infty$ separately, as a limit case; however, here and in the following, we refrain from presenting these details).\medskip

We will use jet bundles to define differential operators between sections of vector bundles. Let $\pi':E'\to X$ be a second vector bundle and take the pullback bundle $\pi^*_k(\pi')$, $k\in\N$, see Figure 1.
\begin{figure}[h]
\begin{center}
\begin{tikzpicture}
  \matrix (m) [matrix of math nodes, row sep=3em, column sep=3em]
    {  \zp_k^*E' & E'  \\
       J^k(\zp) & X  \\ };
 \path[->]
 (m-1-2) edge  node[right] {$\scriptstyle{\zp'}$} (m-2-2);
 \path[->]
 (m-1-1) edge  node[above] {$\scriptstyle{p}$} (m-1-2);
  \path[->]
 (m-1-1) edge  node[left] {$\scriptstyle{\zp_k^*(\zp')}$} (m-2-1);
  \path[->]
 (m-2-1) edge  node[above] {$\scriptstyle{\zp_k}$} (m-2-2);
\end{tikzpicture}
\caption{Pullback bundle}
\end{center}
\end{figure}
Consider now the $\cF_k(\zp)$-module of sections $\zG(\pi^*_k(\pi'))$. If $\zp':X\times \R\to X$, the latter can be naturally identified with $\cF_k(\zp)$. This justifies the notation $\mathcal{F}_k(\pi,\pi'):=\zG(\zp_k^*(\zp'))$. We denote the composite of $$\psi\in\cF_k(\zp,\zp')\subset\Ci(J^k(\zp),\zp_k^*E')$$ and $p\in\Ci(\zp_k^*E',E')$ also by $\psi$. Hence, $\psi\in\Ci(J^k(\zp),E'),$ and, for any point $j^k_m\zf\in J^k(\zp)$, we have $\psi(j^k_m\zf)\in E'_m$, i.e., $\psi$ is a fiber bundle morphism $\psi\in{\tt FB}(J^k(\zp),E')$. We thus get an isomorphism of $\Ci(X)$-modules: \be\label{SecFBiso}\zG(\zp^*_k(\zp'))=\cF_k(\zp,\zp')\simeq{\tt FB}(J^k(\zp),E')\;.\ee Since, for every section $\zf\in\zG(\zp)$, the composite of $$j^k\zf\in\zG(\zp_k)\subset\Ci(X,J^k(\zp))$$ and $\psi$ is a section $\psi\circ(j^k\zf)\in\zG(\zp')$, we see that $\psi\in\cF_k(\zp,\zp')$ implements a map $$D:\zG(\zp)\ni \zf\mapsto D(\zf)=\psi\circ(j^k\zf)\in\zG(\zp')\;,$$ such that the value $D(\zf)|_m$ only depends on $j^k_m\zf$. We therefore say that $D$ is a not necessarily linear differential operator of order $k$ between $\zp$ and $\zp'\,.$

\begin{defi} A (not necessarily linear) {\bf differential operator} $D\in\op{DO}_k(\zp,\zp')$ of order $k$ from $\zp$ to $\zp'$ is a map $D:\zG(\zp)\to\zG(\zp')$ that factors through the $k$-jet bundle, i.e., that reads \be\label{NLDOk2}D=\psi_D\circ(j^k-)\;,\ee for some section or fiber bundle morphism $\psi_D\in\cF_k(\zp,\zp')\simeq{\tt FB}(J^k(\zp),E')$. This morphism, which is visibly unique, is the {\bf representative morphism} of $D\,.$\end{defi}

In trivializations of $\zp$ and $\zp'$ over the same chart $(U,x)$ of $X$, such a $k$-th order differential operator reads \be\label{NLDOkLoc}\psi^b_D(x,\p_x^\za\zf^a)=\psi^b_D(x,u^a_\za)|_{j^k_x\zf},\quad (a\in\{1,\ldots,\op{rk}(\zp)\}, b\in\{1,\ldots,\op{rk}(\zp')\}, |\za|\le k)\;.\ee If both ranks are 1 and we write $\psi$ (resp., $t$) instead of $\psi_D$ (resp., $x=(x^1,\ldots,x^n)$), we recover \be\label{NLDOk11loc} \psi(t,\zf(t), d_t\zf,\ldots, d_t^k\zf)=\psi(t,u,u_1,\ldots,u_k)|_{j^k_t\zf}\;\ee (see beginning of \ref{JetDiff}).

The composite of a differential operator $D\in\op{DO}_k(\zp,\zp')$ and a differential operator $D'\in\op{DO}_\ell(\zp',\zp'')$ is a differential operator $D'\circ D\in\op{DO}_{k+\ell}(\zp,\zp'')\,.$

The set $\op{DO}_k(\zp,\zp')$ is a $\Ci(X)$-module. There is a canonical {\bf $\Ci(X)$-module isomorphism} \be\label{SecDOiso}\op{DO}_k(\zp,\zp')\simeq \cF_k(\zp,\zp')\simeq{\tt FB}(J^k(\zp),E')\;.\ee The natural surjective morphisms $\pi_{k\ell}$, $k\le\ell$, give rise to inclusions $\op{DO}_k(\zp,\zp')\subset\op{DO}_\ell(\zp,\zp')$, thus leading to an increasing sequence of $\Ci(X)$-modules. The colimit is the filtered $\Ci(X)$-module \be\label{FiltDOs}\op{DO}(\zp,\zp')=\bigcup_i\op{DO}_i(\zp,\zp')\;\ee of all differential operators from $\zp$ to $\zp'\,.$\medskip

If, for $r,r'\in\R$ and $\zf,\zf'\in\zG(\zp)$, we have $$D(r\zf+r'\zf')=r\,D(\phi)+r'\,D(\zf')\;,$$ the differential operator $D$ is said to be linear. We denote the $\Ci(X)$-submodule made of the linear differential operators of order $k$ (resp., of all linear differential operators) from $\zp$ to $\zp'$ by $$\op{Diff}_k(\zp,\zp')\subset\op{DO}_k(\zp,\zp')\quad (\text{resp.,}\quad \op{Diff}(\zp,\zp')\subset\op{DO}(\zp,\zp'))\;.$$

In trivializations of $\zp$ and $\zp'$ over the same chart $(U,x)$ of $X$, a linear differential operator $D$ of order $k$ reads \be\label{NLDOkLoc}\psi^b_D(x,\p_x^\za\zf^a)=\psi^b_D(x,u^a_\za)|_{j^k_x\zf},\quad (a\in\{1,\ldots,\op{rk}(\zp)\}, b\in\{1,\ldots,\op{rk}(\zp')\}, |\za|\le k)\;,\ee where the $\psi^b_D$ are $\Ci(x(U))$-linear in the derivatives, i.e., $$\psi^b_D(x,\p_x^\za\zf^a)=\sum_{\za,a} M^b_{\za a}(x)\p_x^\za\zf^a\;.$$

In fact, a differential operator is a {\it linear} operator $D\in\Diff_k(\zp,\zp')$ if and only if its representative morphism is a {\it vector} bundle morphism $\psi_D\in{\tt VB}(J^k(\zp),E')$ (not only a fiber bundle morphism), i.e., a $\Ci(X)$-linear map $\psi_D\in\h_{\Ci(X)}(\zG(\zp_k),\zG(\zp'))$ (denoted by the same symbol). This passage from the vector bundle map to the linear map between sections allows to replace $D(-)=\psi_D\circ (j^k-)$, see (\ref{NLDOk2}), by $D(-)=(\psi_D\circ j^k)(-)\,.$ Therefore,

\begin{prop} A {\bf linear differential operator} $D\in\op{Diff}_k(\zp,\zp')$ is an $\R$-linear map $D:\zG(\zp)\to\zG(\zp')$ that factors through the $k$-jet bundle, i.e., that reads \be\label{LDOk3}D=\psi_D\circ j^k\;,\ee for some (and thus unique) vector bundle or $\Ci(X)$-module morphism $\psi_D\in{\tt VB}(J^k(\zp),E')\simeq\h_{\Ci(X)}(\zG(\zp_k),\zG(\zp'))$. Hence the isomorphisms of $\Ci(X)$-modules
\be\label{LDOk2}\Diff_k(\pi,\pi')\simeq {\tt VB}(J^k(\zp),E')\simeq\Hom_{C^\infty(X)}(\zG(\pi_k),\zG(\pi'))\;,\ee
and
\be\label{LDO2}\Diff(\pi,\pi')\simeq {\tt VB}(J^\infty(\zp),E')\simeq\Hom_{C^\infty(X)}(\zG(\pi_\infty),\zG(\pi'))\;.\ee
\end{prop}

We close the present section with the remark that, in the case $\zp=\zp'=\op{pr}_1:X\times \R\to X$, the differential operators $\Diff(\zp,\zp')$ act on functions $\Ci(X)$, and that we then write $\cD(X)$ instead of $\Diff(\op{pr}_1,\op{pr}_1)$; in other words:

\begin{rem}{\em\label{LDOFunct} We denote by $\cD(X)$ the associative unital $\R$-algebra of linear differential operators acting on functions $\Ci(X)$ of a smooth manifold $X$.}\end{rem}

\subsection{Partial differential equations and their prolongations}\label{PDEP}

A second fundamental feature is that one prefers replacing the original system of {\small PDE}-s by an enlarged system, its prolongation, which also takes into account the differential consequences of the original one. More precisely, if $\zf(t)$ satisfies the original {\small DE} (\ref{PDE}), {\bf we have}, for any $\ell\in\N\,,$ $$d^r_t(\psi(t,\zf(t),d_t\zf,\ldots,d_t^k\zf))=(\p_t+u_1\p_u+u_2\p_{u_1}+\ldots)^r \psi(t,u,u_1,\ldots,u_k)|_{j^{k+\ell}_t\zf}=:$$ \be\label{ProlPDE}D_t^r \left(\psi(t,u,u_1,\ldots,u_k)\right)|_{j^{k+\ell}_t\zf}\equiv 0,\;\forall r\le \ell\;.\ee
Let us stress that the `total derivative' $D_t$ or `horizontal lift' $D_t$ of $d_t$ is actually an infinite sum. The {\small DE} (\ref{PDE}) and the system of {\small DE}-s (\ref{ProlPDE}), have clearly the same solutions, so we may focus just as well on (\ref{ProlPDE}). The corresponding system of algebraic equations \be\label{ProlAlgE}(D_t^r \psi)(t,u,u_1,\ldots,u_k,u_{k+1},\ldots,u_{k+r})=0,\;\forall r\le \ell\;\ee defines a `surface' $\zS^\ell$ in the $(k+\ell)$-jet space. {\it A solution of the original {\small DE} (\ref{PDE}) is now a function $\zf$ such that the graph $\op{gr}(j^{k+\ell}\zf)$ is a subset of ${\zS}^\ell$}. The `surface' $\zS^\ell$ is referred to as the $\ell$-th prolongation of the considered {\small DE} or differential operator.\medskip

To grasp the interest in differential consequences, consider for instance the integration problem $\p_{x^i}F=f_i$ ($i\in\{1,\ldots,n\}$) in $\R^n$ -- where notation is obvious -- . The differential consequences of this (system of) {\small PDE}(-s) include the equations $\p_{x^j}\p_{x^i}F=\p_{x^j}f_i$ ($i,j\in\{1,\ldots,n\}$), hence, they include the compatibility conditions $\p_{x^j}f_i=\p_{x^i}f_j$.\medskip

In the case $k=\ell=1$, the equation of $\zS^0\subset J^1$ (resp., of $\zS^1\subset J^2$) is $$\psi(t,u,u_1)=0\quad \text{(resp.,}\quad\psi(t,u,u_1)=0\;\,\text{and}\;\, (D_t\psi)(t,u,u_1,u_2)=0)\;,$$ (see (\ref{ProlAlgE})). Hence, $\zS^1$ is the set of points $j^2_{t_0}\zf\in J^2$ such that $j^1_{t_0}\zf\in\zS^0$ and $$(\p_t\psi+u_1\p_u\psi+u_2\p_{u_1}\psi)|_{j^2_{t_0}\zf}=\p_t\psi|_{j^1_{t_0}\zf}+d_t\zf|_{t_0}\p_u\psi|_{j^1_{t_0}\zf}+d^2_t\zf|_{t_0}\p_{u_1}\psi|_{j^1_{t_0}\zf}=0\;.$$ The last requirement means that the tangent vector $(1,d_t\zf|_{t_0},d^2_t\zf|_{t_0})$ at $t_0$ of the curve $(t,\zf(t),d_t\zf)\in J^1$ is an element of the vector space $$T_{j^1_{t_0}\zf}\zS^0: \p_t\psi|_{j^1_{t_0}\zf}\;t+\p_u\psi|_{j^1_{t_0}\zf}\;u+\p_{u_1}\psi|_{j^1_{t_0}\zf}\;u_1=0\;$$ that is tangent to $\zS^0$ at $j^1_{t_0}\zf\,.$ Thus, \be\label{ProlCoord}\zS^1=\{j^2_{t_0}\zf\in J^2:\op{gr}(j^1\zf)\; \text{is tangent to}\; \zS^0\; \text{at}\; j^1_{t_0}\zf\}\;.\ee

Observe that the equations of $\zS^0$ and $\zS^1$ show that $\zS^\ell$ is not necessarily a smooth manifold and that $\zp_{12}:\zS^1\to \zS^0$ is not necessarily a smooth fiber bundle.\medskip

We now define partial differential equations and their prolongations in a coordinate-free manner.

\begin{defi}\label{PDESol} A {\bf partial differential equation} $(\,$resp., a {\bf linear} partial differential equation$\,)$ of order $k$ $(\,$$k\ge 0$$\,)$ acting on sections $\zf\in\zG(\zp)$ of a vector bundle $\zp$, is a smooth fiber $(\,$resp., vector$\,)$ subbundle $\zp_k:\zS^0\to X$ of $J^k(\zp)$. The {\bf $\ell$-th prolongation} of $\zS^0$ $(\,$$0\le\ell\le\infty$$\,)$ is the subset \be\label{ProlGeo}\zS^\ell=\{j^{k+\ell}_{m}\zf\in J^{k+\ell}(\zp): \op{gr}(j^k\zf)\;\,\text{is tangent up to order $\ell$ to}\;\, \zS^0\;\,\text{at}\;\, j_m^k\zf\}\;\ee of $J^{k+\ell}(\zp)\,.$ A $(\,$local$\,)$ {\bf solution} of $\zS^0$ is a $(\,$local$\,)$ section $\zf$ of $\zp$ such that $\op{gr}(j^k\zf)\subset\zS^0\,.$ \end{defi}

Note that the definition of the prolongation means that the points $j^{k+\ell}_m\zf$ of $\zS^\ell$ provide $\ell$-th order approximations $\op{gr}(j^k\zf)$ of possible solutions of $\zS^0\,.$

\begin{rem}{\em\label{FromIntSol}\begin{enumerate}\item In the following we always assume that the considered equation $\zS^0\subset J^k(\zp)$ is {\bf formally integrable} (see also Subsection \ref{CC}), i.e., that \begin{itemize}\item the prolongations $\zS^\ell$ are smooth manifolds $(0\le \ell\le \infty)$, and
\item the maps $\zp_{k+\ell,k+\ell+1}:\zS^{\ell+1}\to\zS^{\ell}$ $(0\le \ell<\infty)$ are smooth fiber bundles.\end{itemize} \item Let us stress as well that it follows from Definition \ref{PDESol} (see also introduction to the present subsection \ref{PDEP}) that $\zf$ is a solution of $\zS^0$ if \be\label{SolCond}\op{gr}(j^{k+\ell}\zf)\subset \zS^\ell\;,\ee for some $0\le \ell\le\infty\,,$ and that, conversely, we have (\ref{SolCond}) for every $\ell\,,$ if $\zf$ is a solution. \end{enumerate}
}\end{rem}

A {\small PDE} (resp., a linear {\small PDE}) $\zS^0$ of order $k$ in $\zp$ is {\bf implemented by a differential operator} $D\in\op{DO}_k(\zp,\zp')$ (resp., $D\in\op{Diff}_k(\zp,\zp')$), if $\zS^0=\ker\psi_D$, where $\zp':E'\to X$ is a vector bundle and where $\psi_D\in{\tt FB}(J^k(\zp),E')$ (resp., $\psi_D\in{\tt VB}(J^k(\zp),E')$) is the representative morphism of $D\,.$ In this case, the differential operator $j^\ell\circ D$ is of order $k+\ell$ and acts from $\zp$ to $\zp'_\ell$. Its decomposition \be\label{ProlPDECoorFree}j^\ell\circ D=\psi_{j^\ell\circ D}\circ j^{k+\ell}\ee corresponds to Equation (\ref{ProlPDE}). In the sequel we write \be\label{ProlDO}\psi^\ell_D:J^{k+\ell}(\zp)\to J^\ell(\zp')\ee for the representative morphism $\psi_{j^\ell\circ D}$ of the $\ell$-th prolongation $j^\ell\circ D$ of $D$. It is now clear that \be\label{ProlEq}\zS^\ell=\ker\psi^\ell_D\;,\ee i.e., that the $\ell$-th prolongation of the equation is given by the $\ell$-th prolongation of the corresponding differential operator (see Equation (\ref{ProlAlgE})).

\subsection{Cartan distribution}

Jet spaces $\zp_k:J^k(\zp)\to X$, $0\le k\le \infty$, come equipped with a natural geometric structure, their {\bf Cartan distribution} $\cC^k=\cC^k(\zp)$, i.e., with an assignment \be\label{CartDistJ}\cC^k:J^k(\zp)\ni\zk_k\mapsto \cC^k_{\zk_k}\subset T_{\zk_k}(J^k(\zp))\;\ee of a vector subspace $\cC^k_{\zk_k}$ of the corresponding tangent space to any point of the jet space. This subspace can be defined in a coordinate-free manner, which will however not be detailed here. The next proposition gives the coordinate description of $\cC^k_{\zk_k}$.

\begin{prop} Let $\zp: E\to X$ be a vector bundle of rank $r$ over a manifold of dimension $n\,.$ For any $k\ge 0$ and any $\zk_k\in J^k(\zp)\,,$ the Cartan space $\cC^k_{\zk_k}=\cC^k_{\zk_k}(\zp)$ is generated by the vectors $$D_{x^i}^{\le k-1}|_{\zk_k}=\p_{x^i}+\sum_{a=1}^r\sum_{|\za|\le k-1} u^a_{i\za}\p_{u^a_\za}\;|_{\zk_k}\quad\text{and}\quad \p_{u^a_\za}\,|_{\zk_k}\;,$$ \be{\label{CartDistFundaGen}}i\in\{1,\ldots,n\}, a\in\{1,\ldots,r\},|\za|=k\;,\ee where $(x^i,u^a_\za)$ is a trivializing chart of $J^k(\zp)$ around $\zp_k(\zk_k)\,.$ In the limit case $k=\infty\,,$ the Cartan space $\cC^\infty_{\zk_\infty}$ is generated by the total derivatives \be\label{CartDistFundaGenInfty}D_{x^i}|_{\zk_\infty}\,,\quad i\in\{1,\ldots,n\}\;.\ee\end{prop}

The existence of the extra generators $\p_{u^a_\za}$ ($a\in\{1,\ldots,r\}$, $|\za|=k$) makes {\it the Cartan distribution $\cC^k=\cC^k(\zp)$ non-integrable}. Indeed, take, to simplify, the case $k=n=r=1$ and note that the bracket $[D_t^{\le 0},\p_{u_1}]=[\p_t+u_1\p_u,\p_{u_1}]=-\p_u$ of local vector fields in $\cC^1$ is not located in $\cC^1\,.$ This problem disappears at the limit $k=\infty$: {\it the Cartan distribution $\cC^\infty=\cC^\infty(\zp)$ is $n$-dimensional and integrable} (indeed $[D_{x^i},D_{x^j}]=0\,$).

\begin{rem}{\em In the sequel, we deal with limits, e.g., infinite prolongations $\zS^\infty$. Whenever no confusion arises, we omit the sub- and superscripts $\infty$, thus writing $\zS$ (resp., $\zk, \cC,\ldots$) instead of $\zS^\infty$ (resp., $\zk_\infty,\cC^\infty,\ldots\,$).}\end{rem}

Consider now a {\small PDE} $\zS^0\subset J^k(\zp)$ of order $k$ on $\zp\,$ (as mentioned before, we systematically assume that the considered {\small PDE}-s are formally integrable). We define the Cartan distribution $\cC^k(\zS^0)$ of $\zS^0$ by \be\label{CartDistE}\cC^k(\zS^0):\zS^0\ni\zk_k\mapsto \cC^k_{\zk_k}\cap T_{\zk_k}\zS^0\subset T_{\zk_k}\zS^0\;,\ee and the Cartan distribution $\cC(\zS)$ of $\zS\subset J^\infty(\zp)$ by \be\label{CartDistELim}\cC(\zS):\zS\ni\zk\mapsto\cC_\zk\cap T_\zk\zS\subset T_\zk\zS\;.\ee It can be shown that \be\label{FormIntCons}\cC_{\zk}=\cC_\zk(\zp)\subset T_{\zk}\zS\;,\ee so that \be\label{FormIntCons2} \cC(\zS)=\cC(\zp)|_\zS\;.\ee Moreover, just as $\cC(\zp)\,,$ {\it the Cartan distribution $\cC(\zS)=\cC(\zp)|_\zS$ is $n$-dimensional and integrable}.\medskip

Eventually: \begin{prop} The maximal dimensional $(\,$$n$-dimensional$\,)$ integral manifolds of the Cartan distribution $\cC(\zp)$ $(\,$resp., $\cC(\zS)$$\,)$ are the graphs $\op{gr}(j^\infty\zf)$ of the infinite jets of the local sections $\zf\in\zG_{\op{loc}}(\zp)$ $(\,$resp., the local solutions $\zf\in\zG_{\op{loc}}(\zp)$ of $\zS^0$$\,)$.\end{prop}\medskip

Hence, the set of maximal dimensional integral manifolds in $(\zS,\cC(\zS))$ can be identified with the set of solutions of $\zS^0$. Since all relevant information about the original {\small PDE} $\zS^0$ is thus encrypted in the pair $(\zS,\cC(\zS))$, the partial differential equation $\zS^0$ is frequently identified with the `diffiety' $(\zS,\cC(\zS))$. {\bf Diffieties}, or, explicitly, differential varieties, are for partial differential equations what algebraic varieties are for algebraic equations. Diffieties are (often infinite-dimensional) manifolds equipped with a Cartan distribution; they are locally equivalent to infinite prolongations of differential equations. The Cartan distribution allows developing on a diffiety a specific differential calculus, called Secondary Calculus, whose objects are cohomology classes of differential complexes. Many characteristics of a diffiety, i.e., of the corresponding systems of partial differential equations, can be expressed in terms of Secondary Calculus and vice versa.

\subsection{Cartan connection}

\subsubsection{Horizontal vector fields}

Since $$\cC(\zp):J^\infty(\zp)\ni\zk\mapsto\cC_\zk(\zp)\subset T_\zk J^\infty(\zp)\;,$$ where $\cC_\zk(\zp)$ is the tangent space at $\zk$ to the graphs $\op{gr}(j^\infty\zf)$ of the {\it sections} $j^\infty\zf$ that pass through $\zk$ at $m=\zp_\infty(\zk)\,,$ the following statements are rather obvious: \begin{itemize}\item $T_\zk\zp_\infty: \cC_\zk(\zp)\to T_mX$ is a vector space isomorphism (it is easily seen that this derivative sends $D_{x^i}|_\zk$ to $\p_{x^i}|_{\zp_\infty(\zk)}$). \item The $\cF(\zp)$-module $\cC\Theta(\zp):=\zG(\cC(\zp))$ (resp., $\Theta^v(\zp)$) of sections of the subbundle $\cC(\zp)\subset T\,J^\infty(\zp)$ (resp., of $\zp_\infty$-vertical vector fields of $J^\infty(\zp)$) is a submodule of the $\cF(\zp)$-module $\Theta(\zp)$ of vector fields of $J^\infty(\zp)\,.$ More precisely, we have \be\label{Dist1}\Theta(\zp)=\cC\Theta(\zp)\oplus\Theta^v(\zp)\;.\ee \end{itemize} This suggests the idea of connection, i.e., of a $\Ci(X)$-linear lift (map with the obvious projection property) \be\label{CartConn}\cC:\Theta(X)\ni\zy\mapsto\cC\zy\in\cC\Theta(\zp)\;.\ee Indeed, its suffices to set, for any $\zk\in J^\infty(\zp)$ with projection $\zp_\infty(\zk)=m\,,$ \be\label{CartConnDef}(\cC\zy)_\zk:=(T_\zk\zp_\infty)^{-1}\zy_m\in\cC_\zk(\zp)\subset T_\zk J^\infty(\zp)\;.\ee This connection $\cC$ on $J^\infty(\zp)$ is the {\bf Cartan connection} induced by the Cartan distribution $\cC(\zp)$ on $J^\infty(\zp)\,.$\medskip

As, in trivializing coordinates $(x^i,u^a_\za)$ of $J^\infty(\zp)$ over $U$ around $m=\zp_\infty(\zk)$, the Cartan space $\cC_\zk(\zp)$ is generated by the $D_{x^i}|_\zk$, the {\it horizontal vector fields $H\in \cC\Theta(\zp)$ are locally generated over functions of $J^\infty(\zp)$ by the total derivatives $D_{x^i}$}: \be\label{HorVecField} H|_{\zp_\infty^{-1}(U)}=\sum_j H^j(x^i,u^a_\za)D_{x^j}\;.\ee Since $T_\zk\zp_{\infty}(D_{x^j}|_\zk)=\p_{x^j}|_m,$ {\it a vector field $\zy|_U=\sum_j\zy^j(x^i)\p_{x^j}$ is lifted to} \be\label{HorLiftVecField}(\cC\zy)|_{\zp_\infty^{-1}(U)}=\sum_j \zy^j(x^i)D_{x^j}\;.\ee Let us also mention, for the sake of completeness, that {\it a vector field $T\in\Theta(\zp)$ $(\,$resp., a vertical vector field $V\in\Theta^v(\zp)$$\,)$ locally reads} \be\label{JetBd(V)Vf}T|_{\zp_\infty^{-1}(U)}=\sum_jT^j(x^i,u^a_\za)\p_{x^j}+\sum_{b\zb}T^b_\zb(x^i,u^a_\za)\p_{u^b_\zb}\quad (\text{resp.,}\quad V|_{\zp_\infty^{-1}(U)}=\sum_{b\zb}V^b_\zb(x^i,u^a_\za)\p_{u^b_\zb})\;.\ee

We are now able to rewrite the definition of a horizontal lift $\cC\zy$ in a useful way. If $\zy\in\Theta(X)$ and $F\in\cF(\zp)$, and if $\zf$ is a local section in $\zG(\zp)$ that is defined around $m\in X$, we get $$(j^\infty\zf)^*((\cC\zy)F)|_m=((\cC\zy)F)|_{j^\infty_m\zf}=\left((\cC\zy)_{j^\infty_m\zf} F\right)|_{j^\infty_m\zf}=\left(((T\zp_\infty)^{-1}\zy_m) F\right)|_{j^\infty_m\zf}=$$ $$\zy_m(F\circ j^\infty\zf)|_m=\zy((j^\infty\zf)^*F)|_m\;.$$ Indeed, the isomorphism $(T\zp_\infty)^{-1}$ sends a partial derivative to the corresponding total derivative. Observe also that, although the function $F\circ j^\infty\zf$ depends on $\zf$, its derivative $\zy_m(F\circ j^\infty\zf)|_m$ depends only on $j^\infty_m\zf\,.$ Hence, the \begin{prop} For any $\zy\in\Theta(X)$, $F\in\cF(\zp)$, and $\zf\in\zG_{\op{loc}}(\zp)$, we have \be\label{LiftVf}(j^\infty\zf)^*((\cC\zy)F)=\zy((j^\infty\zf)^*F)\;.\ee\end{prop} It is clear that we could define the Cartan connection (\ref{CartConnDef}) by means of (\ref{LiftVf}), and that Equation (\ref{LiftVf}) is the generalization of Equation (\ref{ProlPDE}).\medskip

We already explained that $[\cC\Theta(\zp),\cC\Theta(\zp)]\subset \cC\Theta(\zp)$. Moreover, it immediately follows from (\ref{LiftVf}) that $\cC[\zy,\zy']=[\cC\zy,\cC\zy']$. In other words, {\it the integrable Cartan distribution of $J^\infty(\zp)$ induces a flat Cartan connection on $J^\infty(\zp)\to X$}. Further, the increasing sequence $\cC(\Theta(X))\subset\cC\Theta(\zp)\subset\Theta(\zp)$ is a {\it sequence of Lie subalgebras}. Eventually, if $\zS$ is the infinite prolongation of a {\small PDE} on $\zp$, we set $\cC\Theta(\zS):=\zG(\cC(\zS))$, where $\cC(\zS)$ is the Cartan distribution of $\zS$. This $\cF(\zS)$-module is locally generated by the $D_{x^i}|_\zS$. When restricting the lifts $\cC\zy$ to $\zS$, we get a connection $\cC:\Theta(X)\to \cC\Theta(\zS)$, the Cartan connection on $\zS$, which is flat as well. Hence, {\it the integrable Cartan distribution of $\zS$ induces a flat Cartan connection on $\zS\to X,$ which is the restriction of the connection on the infinite jet space.}

\subsubsection{Horizontal differential operators}

Total differential operators ({\small TDO}-s) \be\label{HDiff}\Psi=\sum_\zb \Psi^\zb(x^i,u^a_\za)D_x^\zb\ee are of primary importance in Field Theory. The fundamental property is that {\it {\small TDO}-s act not only on $\cF(\zp)$, but also on $\cF(\zS)$}. This is of course due to the fact that total derivatives restrict to (horizontal) vector fields of $\zS$ (see Equation (\ref{FormIntCons})), and is not true for ordinary differential operators \be\label{Diff}\mathbb{T}=\sum_\zg \mathbb{T}^\zg(x^i,u^a_\za)\;\ldots\circ\p^{\zg_j}_{x^j}\circ\ldots\circ\p^{\zg_{b\zb}}_{u^b_\zb}\circ\ldots\ee of $J^\infty(\zp)$. An interesting subclass of {\small TDO}-s are the lifts \be\label{LiftsDiff}\cC\zD=\sum_\zb \zD^\zb(x^i) D_x^\zb\ee of linear differential operators $\zD=\sum_\zb \zD^\zb(x^i) \p_x^\zb$ acting on $\Ci(X)$. These lifts can be defined exactly as the lifts of base vector fields in (\ref{LiftVf}).\medskip

Note first that differential operators act usually not only on functions $\Ci(X)$ (resp., on $\cF(\zp)$ (functions of $J^\infty(\zp)$)), but act between sections $\zG(\zh_k)$ (locally: $\R^{r_k}$-valued functions on `$X$') of rank $r_k$ vector bundles $\zh_k:E_k\to X$ (resp., between sections $\cF(\zp,\zh_k)=\zG(\zp_\infty^*(\zh_k))$ (locally: $\R^{r_k}$-valued functions on `$J^\infty(\zp)$') of the bullbacks $\zp_\infty^*(\zh_k): \zp_\infty^*(E_k)\to J^\infty(\zp)$ of these bundles). Hence, the

\begin{defi} Let $\zp:E\to X$ and $\zh_k:E_k\to X$ $(k\in\{1,2\})$ be vector bundles. The {\bf lift of a linear differential operator} $\zD:\zG(\zh_1)\to\zG(\zh_2)$ is the linear differential operator $\cC\zD:\cF(\zp,\zh_1)\to \cF(\zp,\zh_2)$ $(\,$of same order$\,)$ defined by \be\label{LiftsDiffGen}(j^\infty\zf)^*((\cC\zD)S)=\zD((j^\infty\zf)^*S)\;,\ee where $S\in\cF(\zp,\zh_1)$ and $\zf\in\zG_{\op{loc}}(\zp)$.\end{defi}

The difference with lifts $$\cC\zy=\sum_j \zy^j(x^i)D_{x^j}\in\cC\Theta(\zp)$$ of vector fields is that the horizontal or $\cC$-vector fields $\cC\Theta(\zp)$ had been defined before the lifts $\cC\zy$. Here, i.e., for lifts $\cC\zD$ of differential operators, we still need to find the proper definition of $\cC$-differential operators $\cC\op{Diff}(\zp_\infty^*(\zh_1),\zp_\infty^*(\zh_2))$. In view of (\ref{HorVecField}), these $\cC$-differential operators should locally be the {\small TDO}-s $$\Psi=\sum_\zb \Psi^\zb(x^i,u^a_\za)D_x^\zb\;,$$ see (\ref{HDiff}). Since, for any $F\in\cF(\zp)$ and any $\zf\in\zG(\zp)$, this model $\cC$-differential operator $\Psi$ satisfies $$(\Psi F)\circ j^\infty\zf=\sum_\zb (\Psi^\zb\circ j^\infty\zf)\; ((D_x^\zb F)\circ j^\infty\zf)=\sum_\zb (\Psi^\zb\circ j^\infty\zf)\; \p_x^\zb (F\circ j^\infty\zf)=:\Psi_\zf(F\circ j^\infty\zf)\;,$$ we have $$(j^\infty\zf)^*(\Psi F)=\Psi_\zf((j^\infty\zf)^*F)\;,$$ where the {\small RHS} $\Psi_\bullet$ (see its definition) is a not necessarily linear differential operator in $\zf\in\zG(\zp)$ with values $\Psi_\zf$ in linear differential operators on $\Ci(X)$. This motivates the

\begin{defi} A linear differential operator $\Psi:\cF(\zp,\zh_1)\to\cF(\zp,\zh_2)$ is a {\bf $\cC$-differential operator} $\Psi\in\cC\op{Diff}(\zp_\infty^*(\zh_1),\zp_\infty^*(\zh_2))$, if, for any $\zf\in\zG(\zp)$, there exists a linear differential operator $\Psi_\zf:\zG(\zh_1)\to \zG(\zh_2)$, such that, for any $S\in\cF(\zp,\zh_1)$, the equality \be\label{CDiff}(j^\infty\zf)^*(\Psi S)=\Psi_\zf((j^\infty\zf)^*S)\ee holds. \end{defi}

This definition captures correctly our intuition of $\cC$-differential operators. Since it is clear from its definition that the lift $\cC$ of differential operators respects composition, we have, locally, $$\sum_\zb \Psi^\zb(x^i,u^a_\za)D_x^\zb=\sum_\zb \Psi^\zb(x^i,u^a_\za)\cC(\p_x^\zb)\;.$$ It can be shown \cite{KV} that this result is global:

\begin{prop} Any $\Psi\in\cC\Diff(\zp_\infty^*(\zh_1),\zp_\infty^*(\zh_2))$ reads \be\label{CDiffStruct} \Psi=\sum_\zb \Psi^\zb\cC\zD_\zb\;,\ee where the sum is finite, where $\Psi^\zb\in\cF(\zp)$, and where $\zD_\zb\in\Diff(\zh_1,\zh_2)$. In other words, $\cC$-differential operators are generated over $\cF(\zp)$ by lifts.\end{prop}

Moreover, just as {\small TDO}-s, $\cC$-differential operators can be restricted to the infinite prolongation $\zS$ of a {\small PDE}. More precisely \cite{KV},

\begin{cor}\label{RestHor} For any $\cC$-differential operator $\Psi:\cF(\zp, \zh_1)\to \cF(\zp, \zh_2)$ and any infinite prolongation $\zS\subset J^\infty(\zp)$, there is a linear differential operator $\Psi_\zS:\cF(\zS, \zh_1)\to \cF(\zS, \zh_2)$ such that, for every $s\in\cF(\zp,\zh_1)$, we have $\Psi_\zS (s|_\zS)=(\Psi s)|_\zS\,$.\end{cor}

Finally, we have the important

\begin{cor} There is a canonical $\cF(\zp)$-module isomorphism \be\label{LDOCoeffGen} \cC:\cF(\zp)\0_{\Ci(X)}\Diff(\zh_1,\zh_2)\to \cC\Diff(\zp_\infty^*(\zh_1),\zp_\infty^*(\zh_2))\ee between the linear differential operators with coefficients in the jet space functions and the corresponding $\cC$-differential operators. In particular, in the case of the trivial line bundle $\zh_1=\zh_2$, we get the isomorphism \be\label{LDOCoeffFun} \cC:\cF(\zp)\0_{\Ci(X)}\cD(X)\to \cC\cD(J^\infty(\zp))\;.\ee \end{cor}

\begin{proof} Observe first that the action of a differential operator $F\0\zD$, with $F\in\cF(\zp)$ and $\zD\in\cD(X)$, on a function $f\in\Ci(X)$ is naturally defined by $$(F\0\zD)(f)=F\;((\zD f)\circ\zp_\infty)\;.$$ The action $(F\0\zD)(s)$, $\zD\in\Diff(\zh_1,\zh_2)$ and $s\in\zG(\zh_1)$, is defined similarly: \be\label{LDOCoeffGenAct} (F\0\zD)(s)=F\;((\zD s)\circ\zp_\infty)\;.\ee The map \be\label{IsoDiffCoeffCDiff} \cC:\cF(\zp)\0_{\Ci(X)}\Diff(\zh_1,\zh_2)\ni F\0\zD\mapsto F\;\cC\zD\in\cC\Diff(\zp_\infty^*(\zh_1),\zp_\infty^*(\zh_2))\;,\ee is obviously well-defined and $\cF(\zp)$-linear. To prove injectivity, assume that $F\;(\cC\zD)(S)$ $=0$, for all $S\in\zG(\zp_\infty^*(\zh_1))$, in particular, for all $S=s\circ\zp_\infty$, $s\in\zG(\zh_1)$. It follows from (\ref{LiftsDiffGen}) that $$(F\circ j^\infty\zf)\; \zD s=(F\;((\zD s)\circ\zp_\infty))\circ j^\infty\zf=0\;,$$ for all $s,\zf$. Eventually, (\ref{LDOCoeffGenAct}) allows to conclude that $F\0\zD=0\,.$ As for surjectivity, recall that any $\cC$-differential operator $\Psi$ reads $\sum_\zb\Psi^\zb\cC\zD_\zb$, and note that $\sum_\zb\Psi^\zb\0\zD_\zb$ is a preimage of $\Psi$.
\end{proof}

Let us summarize in coordinate language what we achieved so far. Consider a {\small PDE} $$\psi^b(x^i,\p_x^\za\zf^a)\equiv 0\;,\forall b\;,$$ whose {\small LHS} sends sections $\zf=(\zf^a(x))_a\in\zG(\zp)$ to sections $\psi=(\psi^b(x))_b:=(\psi^b(x^i,\p_x^\za\zf^a))_b\in\zG(\zh_1)$.  We take into account the linear differential consequences $$\zD\;\psi^b(x^i,\p_x^\za\zf^a):=\sum_\zb M^{c}_{\zb b}(x)\p_x^\zb\;\,\psi^b(x^i,\p_x^\za\zf^a)\equiv 0\;,\forall c\;$$ of this equation, where $\zD\in\Diff(\zh_1,\zh_2)$. The latter condition can be rewritten in the form $$(\cC\zD)\;\psi^b(x^i,u^a_\za)\;|_{j_x^\infty\zf}=\sum_\zb M^c_{\zb b}(x)D_x^\zb\;\,\psi^b(x^i,u^a_\za)\;|_{j_x^\infty\zf}\equiv 0\;,\forall c\;,$$ thus leading to a $\cC$-differential operator $\cC\zD\in\cC\Diff(\zp_\infty^*(\zh_1),\zp_\infty^*(\zh_2))$. Just as the value $$\psi^b(x^i,\p_x^\za\zf^a)\;|_m$$ at $m\in X$ (in fact we mean here the coordinates of $m$; the same notational abuse will be tolerated in the sequel) of the image of $\zf=(\zf^a(x))_a\in\zG(\zp)$ by a differential operator in $\op{DO}_k(\zp,\zh_1)$ only depends on the values $\p^\za_x\zf^a|_m$ of the coefficients of the `Taylor expansion' of $\zf$ at $m$ up to order $k$, the value $$\sum_\zb N^c_{\zb b}(x^i,u^a_\za) D_x^\zb\;\,\psi^b(x^i,u^a_\za)\;|_\zk$$ at $\zk\in J^\infty(\zp)$ of the image of $\psi=(\psi^b(x^i,u^a_\za))_b\in\zG(\zp_\infty^*(\zh_1))$ by a $\cC$-differential operator in $\cC\Diff_k(\zp_\infty^*(\zh_1),\zp_\infty^*(\zh_2))$ only depends on the values $D_x^\zb\,\psi^b(x^i,u^a_\za)|_\zk$ of the total or horizontal derivatives of $\psi$ at $\zk$ up to order $k$. In fact, the $\cC$-differential calculus is similar to the ordinary differential calculus. For $k\in\N\cup\{\infty\}$, the {\bf horizontal $k$-jet} $\bar \jmath^k_\zk S$ at $\zk\in J^\infty(\zp)$ of a local section $S\in\zG(\zp_\infty^*(\zh_1))$ that is defined around $\zk$ is the equivalence class of all such local sections, whose coordinate forms in a trivializing chart $(x^i,u^a_\za,v^b)$ around $\zk$ coincide at $\zk$, together with their total derivatives at $\zk$ up to order $k$.

\begin{rem}{\em In the following, if $\zp:E\to X$ and $\zr:F\to X$ are two vector bundles, we set $R:=\zp_\infty^*(\zr)$ and ${\cal R}:=\zG(R)=\zG(\zp_\infty^*(\zr))$.}\end{rem}

The set $$\bar J^k({H}_1)=\{\bar \jmath^k_\zk S:\zk\in J^\infty(\zp), S\in{\cal H}_1\}$$ is a vector bundle $H_{1,k}:\bar J^k({H}_1)\to J^\infty(\zp)$, called the {\bf horizontal $k$-jet bundle}. A trivializing chart $(x^i,u^a_\za,v^b)$ of $H_1$ induces a trivializing chart $(x^i,u^a_\za, v^b_\zb)$ of $H_{1,k}$ given by \be\label{HorJetBdCoord} x^{i}(\bar\jmath^k_\zk S)=x^i(\zk), u^a_\za(\bar\jmath^k_\zk S)=u^a_\za(\zk), v^b_\zb(\bar\jmath^k_\zk S)=D_x^\zb S^b|_\zk\;.\ee As already suggested above here, the $\cC$-differential or horizontal differential operators $$\Psi\in\cC\Diff_k(H_1,H_2)$$ are those $$\Psi\in\h_\R({\cal H}_1,{\cal H}_2)$$ that factor through the horizontal $k$-jet bundle $\bar J^k({H}_1)$, i.e., that read $\Psi=\psi\circ \bar\jmath^k$, for some (and thus unique) vector bundle map $$\psi\in{\tt VB}(H_{1,k}\,,H_2)\simeq \h_{\cF(\zp)}(\zG(\bar J^k({H}_1)),{\cal H}_2)\;.$$ Actually, the whole theory of jet bundles can be transferred to horizontal jet bundles \cite{Ver}. Indeed, it follows from what has been said that, in the coordinate setting, horizontal jet bundles are just jet bundles with extra coordinates $u^a_\za$ in the base.

\subsection{Classical and higher symmetries I and II}

\subsubsection{Classical symmetries I}

The concept of symmetry is of fundamental importance in many fields of Science and deserves special attention. The notion is quite straightforward -- at least in elementary situations -- . For instance, when thinking about an axial symmetry of a plane domain $S$, we get a bijection $p$ from the plane to itself, such that $p(S)=S$. Similarly, a {\bf symmetry of an equation} $\zS^0\subset J^k(\zp)$ should be a fiber bundle automorphism (or, just a diffeomorphism) $\psi$ of $J^k(\zp)$ such that \be\label{Sym1}\psi(\zS^0)=\zS^0\;.\ee However, since the essential structure of $J^k(\zp)$ is the Cartan distribution $\cC^k$ (i.e., the infinitesimal object that encodes jet prolongations of sections), it seems natural to ask that a symmetry respect the Cartan distribution (or, better, that its tangent map does).\medskip

In the following, we focus on automorphisms of $J^k(\zp)$ that respect $\cC^k$, thus omitting first Condition (\ref{Sym1}). We refer to such automorphisms as {\bf Lie automorphisms} of $\zp_k$. In particular, we may ask whether it is possible to build a Lie automorphism of $\zp_k$ as a prolongation of an automorphism of $\zp$.

\subsubsection{Prolongations of diffeomorphisms and vector fields}

It is easily seen that, if $\Psi=(\psi_0,\psi)$ is a {\it fiber bundle automorphism} of $\zp:E\to X$, we can prolong it to a fiber bundle automorphism $j^\ell\Psi:=(\psi_0,j^\ell\psi)$ of $\zp_\ell:J^\ell(\zp)\to X$. It actually suffices to recall that $\psi\zf\psi_0^{-1}\in\zG(\zp)$, for any $\zf\in\zG(\zp)$ (as elsewhere in this text, we do not insist here on the possibility that $\zf$ might be defined only locally), and to consider the well-defined fiber bundle automorphism $$j^\ell\psi:J^\ell(\zp)\ni j^\ell_m\zf\mapsto j^\ell_{\psi_0(m)}(\psi\zf\psi_0^{-1})\in J^\ell(\zp)\;.$$ It can easily be checked that the lift $j^\ell\Psi$ is a Lie automorphism, i.e., that, for any $\zk_\ell\in J^\ell(\zp)$, the inclusion \be\label{RespCart}(T_{\zk_\ell}j^\ell\psi)(\cC^\ell_{\zk_\ell})\subset \cC^\ell_{j^\ell_{\zk_\ell}\psi}\;\ee holds. 
Let us still mention that the prolongation $j^\ell\psi: J^\ell(\zp)\to J^\ell(\zp)$ of $\psi:J^0(\zp)\to J^0(\zp)$ is really a lifting, in the sense that $\zp_{0\,\ell}\circ j^\ell\psi=\psi\circ\zp_{0\,\ell}\,.$\medskip

Instead of considering finite automorphisms or diffeomorphisms, we can take an interest in infinitesimal ones, i.e, in vector fields. Note that a vector field $\Xi\in\Theta(\zp_0)$, i.e., a field of $\zp:E\to X$ (we avoid writing $\Theta(\zp)$, since this notation is used instead of the more precise $\Theta(\zp_\infty)$), is a $\zp$-{\it projectable vector field} if and only if $T\zp\,\Xi_{e}=\xi_{\zp(e)},$ for all $e\in E$, i.e., if and only if there is a vector field $\xi\in\Theta(X)$ that is $\zp$-related to $\Xi$. It is well-known that this means that $\zp$ intertwines the flows $\psi_t^\Xi$ and $\psi_t^\xi$, i.e., that $\zp\circ \psi^\Xi_t=\psi_t^\xi\circ \zp$ (assume for simplicity that the flows are globally defined). In other words, $\Psi^\Xi_t=(\psi^\xi_t,\psi^\Xi_t)$ is a 1-parameter group of fiber bundle isomorphisms of $\zp:E\to X$, and it can thus be prolonged to a 1-parameter group of Lie isomorphisms $j^\ell\Psi^\Xi_t=(\psi_t^\xi,j^\ell\psi_t^\Xi)$ of $\zp_\ell:J^\ell(\zp)\to X$. The latter implements a vector field $j^\ell\Xi\in\Theta(\zp_\ell)$ -- the $\ell$-jet prolongation of the projectable vector field $\Xi\in\Theta(\zp_0)$ -- . In other words, the lift $j^\ell\Xi$ is given by $$(j^\ell\Xi)_{j^\ell_m\zf}= d_t|_{t=0}j^\ell_{\psi^\xi_t(m)}(\psi_t^\Xi\zf\psi_{-t}^\xi)\;.$$ The flow of the prolongation $j^\ell\Xi$ of $\Xi$ is thus the prolongation $j^\ell\psi_t^\Xi$ of the flow of $\Xi$, which is made of Lie isomorphisms. The explicit coordinate computation of the lift of the projectable field \be\label{ProjVf}\Xi=\sum_j A^j(x^i)\p_{x^j}+\sum_b B^b(x^i,u^a)\p_{u^b}=\sum_j A^j(\p_{x^j}+u^b_j\p_{u^b})+\sum_b (B^b-A^ju^b_j)\p_{u^b}\ee leads to \be\label{ProlVf}j^\ell\Xi=\sum_j A^jD_{x^j}^{\le \ell-1}+\sum_b\sum_{|\zb|\le \ell-1} D_x^\zb(B^b-A^ju^b_j)\p_{u^b_\zb}\;\ee \cite{Kru}. Note that the first term (resp., second term) of the lift is obtained by extending the total derivatives $D_{x^j}^{\le 0}$ in (\ref{ProjVf}) to $D_{x^j}^{\le \ell-1}$ (resp., by adding new terms whose coefficients are the corresponding total derivatives of the coefficients in (\ref{ProjVf})).\medskip

Hence, any {\it fiber bundle automorphism} of $\zp$ (resp., any {\it projectable vector field} of $\zp$) can be prolonged to a fiber bundle automorphism of $\zp_\ell$ (resp., a vector field of $\zp_\ell$) that respects (whose flow respects) the Cartan distribution $\cC^\ell$. The result can be generalized to arbitrary diffeomorphisms $\psi:J^0(\zp)\to J^0(\zp)$ (resp., vector fields $\Xi\in\Theta(\zp_0)$). More precisely, any diffeomorphism (resp., vector field) of $\zp$ can be lifted to a diffeomorphism (resp., vector field) of $\zp_\ell$ that (whose flow) respects the Cartan distribution. We refer to such distribution respecting diffeomorphisms and vector fields as {\bf Lie transformations} and {\bf Lie fields}, respectively (in the case $\ell=0$, any vector in $T_e E$ is tangent to a section, so $\cC^0_e=T_eE$, and Lie transformations (resp., Lie fields) are just diffeomorphisms (resp., vector fields)). The lift to $\zp_\ell$ of an arbitrary vector field of $\zp_0$, i.e., of \be\label{Vf}\Xi=\sum_j A^j(x^i,u^a)\p_{x^j}+\sum_b B^b(x^i,u^a)\p_{u^b}=\sum_j A^j(\p_{x^j}+u^b_j\p_{u^b})+\sum_b (B^b-A^ju^b_j)\p_{u^b}\;,\ee is locally given by the same formula (\ref{ProlVf}) as before \cite{Vit} (any Lie transformation (resp., Lie field) of $\zp_k$ can be lifted to a Lie transformation (resp., Lie field) of any $\zp_{k+\ell}$).
Conversely, any Lie transformation (resp., any Lie field) of $\zp_\ell$ is the lift of a diffeomorphism (resp., a vector field) of $\zp$, at least if $\op{rk}(\zp)>1$, \cite{KV}, \cite{Vit}.

\subsubsection{Classical symmetries II}

In view of what has been said above, a {\bf symmetry of an equation} $\zS^0\subset J^k(\zp)$ is a Lie transformation $\psi$ of $J^k(\zp)$ such that $\psi(\zS^0)=\zS^0$. As also mentioned before, we do in this text usually not insist on possible local aspects. For instance, we could consider here local symmetries of $\zS^0\subset J^k(\zp)$, i.e., Lie transformations $\psi$ of an open subset ${\cal U}\subset J^k(\zp)$ such that $\psi({\cal U}\cap \zS^0)={\cal U}\cap \zS^0$. The notion of {\bf infinitesimal symmetry of an equation} $\zS^0\subset J^k(\zp)$ is now clear as well. It is a Lie field $\zt$ of $J^k(\zp)$ that is tangent to $\zS^0$, i.e., such that $\zt_{\zk}\in T_\zk\zS^0,$ for all $\zk\in{\zS^0}$.

\subsubsection{Higher symmetries I}

Let us recall that we systematically assume that the considered equations are formally integrable (see Remark \ref{FromIntSol} and Subsection \ref{CC}). Just as a Lie transformation (resp., a Lie field) of $J^k(\zp)$ lifts to a Lie transformation (resp., a Lie field) of any $J^{k+\ell}(\zp)$, a symmetry (resp., an infinitesimal symmetry) of $\zS^0\subset J^k(\zp)$ lifts to a symmetry (resp., an infinitesimal symmetry) of any $\zS^\ell\subset J^{k+\ell}(\zp)$ (the converse is true as well) \cite[Prop. 3.23]{KV}. Hence, a symmetry (resp., an infinitesimal symmetry) of $\zS^0$ induces a symmetry (resp., an infinitesimal symmetry) of $\zS:=\zS^\infty$. To avoid diffeomorphisms of infinite dimensional spaces, {\it we consider in the following only infinitesimal symmetries and call them just symmetries}. Further, we will study not only the symmetries of $\zS$ that are implemented by symmetries of $\zS^0$ (such induced symmetries are Lie fields, i.e., the derivatives of the diffeomorphisms obtained from their flows respect the Cartan distribution), but `all symmetries' of $\zS$ (such `higher symmetries' will respect the Cartan distribution in a generalized sense).\medskip

Recall that a symmetry of $\zS=\zS^\infty$ is a vector field $T\in\Theta(\zp)$ of $J^\infty(\zp)$ that is tangent to $\zS$ and that is Lie. A {\bf higher symmetry of $\zS$} (or simply a symmetry of $\zS$ whenever no confusion is possible) is a vector field $T\in\Theta(\zp)$ that is tangent to $\zS$ and respects the Cartan distribution $\cC=\cC(\zp)$ of $J^\infty(\zp)$, not in the preceding sense that the derivatives of its flow respect $\cC$, but in the sense that \be\label{GenSym}[T,\cC\Theta(\zp)]\subset\cC\Theta(\zp)\;,\ee where $\cC\Theta(\zp)=\zG(\cC(\zp))$ is the space of Cartan fields.

\subsubsection{Symmetries of the Cartan distribution}

Just as above, where we omitted Condition (\ref{Sym1}), we will forget now temporarily the tangency condition, and study infinite jet space vector fields that satisfy the Cartan condition (\ref{GenSym}). These fields will be called in the following {\bf symmetries of $\cC$}. In view of the Jacobi identity, the space $\Theta_\cC(\zp)$ of symmetries of $\cC$ is a Lie $\R$-subalgebra of $\Theta(\zp)$. Since $\cC$ is integrable, Cartan fields $\cC\Theta(\zp)$ are {\bf trivial symmetries of $\cC$}, and, by definition, they thus form a Lie ideal of $\Theta_\cC(\zp)$. The quotient $$\op{sym}(\zp):=\Theta_\cC(\zp)/\cC\Theta(\zp)$$ is the Lie algebra of {\bf proper symmetries of $\cC$}. In view of the Cartan connection (\ref{Dist1}), we have the direct sum decomposition \be\Theta_\cC(\zp)=\cC\Theta(\zp)\oplus \op{E\Theta}(\zp)\;,\ee where \be\op{E\Theta}(\zp)=\{T\in\Theta^{v}(\zp):[T,\cC\Theta(\zp)]\subset\cC\Theta(\zp)\}\;.\ee It follows that \be \op{sym}(\zp)\simeq\op{E\Theta}(\zp)\;,\ee i.e., that any proper symmetry of $\cC$ is naturally represented by a vertical symmetry, or, still, by an {\bf evolutionary vector field}.\medskip

Although it is not difficult, we will not explain here that for any $V\in\Theta^v(\zp)$ the symmetry or the evolutionary condition is equivalent to $$[V,\cC(\Theta(X))]=0\;.$$ Since the lifts $\cC(\Theta(X))$ are locally generated over $\Ci(X)$ by the total derivatives and since the local form of a vertical vector field is completely defined by its values on the coordinate functions $u^a_\za$, this condition reads locally $[V,D_{x^i}]=0$, or, still, $[V,D_{x^i}](u^a_\za)=0$. Noticing that $D_{x^i}u^a_{\za}=u^a_{i\za}$, we finally obtain $$V^a_{i\za}=V(u^a_{i\za})=V(D_{x^i}u^a_\za)=D_{x^i}(V(u^a_\za))=D_{x^i}V^a_\za\;.$$ In other words, {\it $V\in\Theta^v(\zp)$ is a local symmetry or evolutionary field if and only if its coefficients satisfy} \be V^a_{i\za}=D_{x^i}V^a_\za\;.\label{EVf}\ee This shows that evolutionary vector fields $V\in\op{E\Theta}(\zp)$ are completely determined (locally, by their coefficients $V^a$, i.e., globally,) by their restriction $V|_{\cF_0}\in\op{Der}^v(\cF_0,\cF)$.\medskip

More precisely, there is a 1:1 correspondence between $\op{E\Theta}(\zp)$ and $\op{Der}^v(\cF_0,\cF)$. It is worth to further elaborate on this idea. Let ${\frak X}\in\op{Der}(\cF_0,\cF)$. Locally, this is a vector field $\frak X$ of $J^0(\zp)$ with coefficients in functions of $J^\infty(\zp)$: \be\label{GenVf}{\frak X}=\sum_j A^j(x^i,u^a_\za)\p_{x^j}+\sum_b B^b(x^i,u^a_\za)\p_{u^b}=\sum_j A^j(\p_{x^j}+u^b_j\p_{u^b})+\sum_b (B^b-A^ju^b_j)\p_{u^b}\;.\ee Such a field can be prolonged to a field of $J^\infty(\zp)$ in the way specified by formula (\ref{ProlVf}), exactly as in the particular cases (\ref{ProjVf}) and (\ref{Vf}) -- except that $\ell=\infty$ here. The prolonged vector field (\ref{ProlVf}) is the sum of a term in $\cC\Theta(\zp)$ (horizontal fields are locally generated over $\cF$ by the total derivatives) and a term in $\op{E\Theta}(\zp)$ (see Equation (\ref{EVf})). In particular, if we start from ${\cal X}\in\op{Der}^v(\cF_0,\cF)$, i.e., locally, from \be\label{VVf}{\cal X}=\sum_b B^b(x^i,u^a_\za)\,\p_{u^b}\;,\ee we obtain the {\it evolutionary vector field} \be\label{STh}\zd_{\cal X}=\sum_{b,\zb} D_x^\zb B^b\;\p_{u_\zb^b}\in\op{E\Theta}(\zp)\;.\ee Note that a local vertical derivation (\ref{VVf}) is the same as a local section $B=(B^b(x^i,u^a_\za))_b$ of $\zp^*_\infty(\zp)$. The point is that this {\it isomorphism} \be\label{VDerDO}\op{Der}^v(\cF_0,\cF)\simeq\zG(\zp^*_\infty(\zp))=\cF(\zp,\zp)=:\varkappa(\zp)\ee holds globally and that the local evolutionary fields (\ref{STh}), computed from the global ${\cal X}\in\op{Der}^v(\cF_0,\cF)$, can be glued to provide a global evolutionary field $\zd_{\cal X}\in\op{E\Theta}(\zp)$.\medskip

It is noteworthy that the {\it 1:1 correspondence} \be\label{Funda1:1}  \zd:\varkappa(\zp)\ni {\cal X}\mapsto \zd_{\cal X}\in\op{E\Theta}(\zp)\;\ee allows to push the $\cF(\zp)$-module structure of $\varkappa(\zp)$ forward to $\op{E\Theta}(\zp)$ (this multiplication is different (!) from that of vector fields of $\zp_\infty$ by functions of $\zp_\infty$) and to pull the Lie algebra structure of $\op{E\Theta}(\zp)$ back to $\varkappa(\zp)$.\medskip

Eventually, the 1:1 correspondence $\zd$ allows introducing a linearization of a not necessarily linear differential operator $D\in\op{DO}(\zp,\zp')\simeq \psi_D\in\cF(\zp,\zp')$ between two vector bundles $\zp$ and $\zp'\,$. For any ${\cal X}\in\varkappa(\zp)$, one can extend the action on $\cF(\zp)$ of $\zd_{\cal X}\in\op{E\Theta}(\zp)$ to an action on $\cF(\zp,\zp')$. Locally, this claim is obvious -- the point is that the extended action is actually a global one -- . The operator \be\label{Lin}\ell_D:\varkappa(\zp)\ni {\cal X}\mapsto \ell_D{\cal X}:=\zd_{\cal X}\psi_D\in\cF(\zp,\zp')\;\ee is the {\bf universal linearization operator} of $D$. In view of (\ref{STh}), we have \be\ell_D{\cal X}=\zd_{\cal X}\psi_D=\sum_{b,\zb} \p_{u^b_\zb}\psi_D \,D_x^\zb{\cal X}^b\;.\label{LinArg}\ee In fact, the partial derivatives $\p_{u^b_\zb}$ ($b\in\{1,\ldots,\op{rk}(\zp)\}$) act on the components $\psi_D^a$ ($a\in\{1,\ldots,\op{rk}(\zp')\}$) of $\psi_D$. In other words, the {\it coordinate expression of the linearization operator} is \be\label{LinLoc}\ell_D= \sum_\zb\lp \p_{u^b_\zb}\psi_D^a\rp_{{a,b}}D_x^\zb\;,\ee where $a$ (resp., $b$) refers to the row (resp., column). The linearization of any (not necessarily linear) differential operator $$D\in\op{DO}(\zp,\zp')$$ is a {\it $(\,$linear$\,)$ horizontal differential operator} \be\ell_D\in\cC\op{Diff}(\zp_\infty^*(\zp),\zp_\infty^*(\zp'))\;.\ee Observe also that the coefficients $\p_{u^b_\zb}\psi_D$ of the linearization of $D\simeq\psi_D$ or of $\ker\psi_D=\zS^0$ are coefficients of the equation of the tangent space of $\zS^0$.

\subsubsection{Higher symmetries II}

To upgrade an evolutionary vector field $V\in\op{E\Theta}(\zp)$ of $J^\infty(\zp)$ to a {\bf symmetry of $\zS^0$} (a proper generalized symmetry of the equation $\zS^0$), we must still add the requirement that $V_\zk\in T_\zk J^\infty(\zp)$ be tangent to the prolongation $\zS\subset J^\infty(\zp)$ when $\zk\in\zS$: $V_\zk\in T_\zk\zS$, for all $\zk\in\zS$. In other words, the considered evolutionary field is a symmetry of the equation $\zS^0$ if and only if it acts on functions $\cF(\zS)$ of the infinite prolongation $\zS$ of $\zS^0$. The space of all symmetries of $\zS^0$ is a Lie $\R$-algebra that we denote by $\op{E\Theta}(\zS)$.\medskip

To finish this review of symmetries, we ask what classical and higher symmetries mean locally, in coordinates, in the case the considered formally integrable equation $\zS^0$ is implemented by a differential operator $D\simeq\psi_D$, i.e., $\zS^0=\ker\psi_D\,.$\medskip

Let first $\zt\in\Theta(\zp_k)$ be a Lie field that is tangent to $\zS^0$. This Lie field is (if $\op{rk}(\zp)>1$) the lift $\zt=j^k\Xi$ of a vector field $\Xi\in\Theta(\zp_0)$. Further, the tangency property means locally that, for any $\zk_k\in\zS^0$, we have \be\label{ClasSymPhys}L_{j^k\Xi}\psi_D|_{\zk_k}\simeq\frac{1}{h}\left(\psi_D(\zk_k+h\zt_{\zk_k})-\psi_D(\zk_k)\right)=0\;.\ee This is exactly the concept of infinitesimal symmetry used in Physics (it means that the infinitesimal transformation induced by $\Xi$ transforms a solution into a solution up to terms of order $\ge 2$ in the infinitesimal parameter).\medskip

Consider now ${\cal X}\in\varkappa(\zp)$, as well as the corresponding proper symmetry $\zd_{\cal X}\in\op{E\Theta}(\zp)$ of $\cC$. When remembering that this field is a symmetry $\zd_{\cal X}\in\op{E\Theta}(\zS)$ of $\zS^0$ if and only if it acts on $\cF(\zS)$, we conclude rather easily that the $\zS^0$-symmetry condition for $\zd_{\cal X}$ is \be\label{HighSymPhys}(\zd_{\cal X}\psi_D)|_\zS=0\;,\ee or, still, \be\label{zSSymCondEVf}(\ell_D{\cal X})|_\zS=\ell_D|_\zS{\cal X}|_\zS=0\;,\ee since $\ell_D$ is a horizontal differential operator and can thus be restricted. In other words, if we denote the restriction of the linearization $\ell_D$ (resp., of the section $\cal X$) by $\ell_\zS$ (resp., ${\cal X}_\zS$), we get the

\begin{prop} Let $\zS^0$ be a formally integrable {\small PDE} in $\zp$, implemented by a differential operator and with infinite prolongation $\zS$. An evolutionary vector field $\zd_{\cal X}$ generated by ${\cal X}\in\varkappa(\zp)$ is a symmetry $\zd_{\cal X}\in\op{E\Theta}(\zS)$ of $\zS^0$ under the necessary and sufficient condition that \be\label{zSSymCondEVf2}{\cal X}_\zS\in\ker\,\ell_\zS\;.\ee\end{prop}

\subsection{Higher symmetries in Gauge Theory}

\begin{rem}{\em We suggest to read this Subsection after Subsection \ref{RIGTDef}.}\end{rem}

We finally explain the gauge theoretical concepts of symmetry of the Euler-Lagrange equations, symmetry of the action, and gauge symmetry. As usual, we denote the coordinates of the considered trivial bundle $\zp:E=\R^n\times\R^r\to X=\R^n$ by $(x^i,u^a)$ and the Lagrangian of the theory by ${\cal L}(x^i,u^a_\za)\,$.\medskip

As mentioned above, a vector field ${\frak X}$ of $J^0(\zp)$ with coefficients in functions of $J^\infty(\zp)$ (see Equation (\ref{GenVf})) can be prolonged to a field of $J^\infty(\zp)$ in the way described by Equation (\ref{ProlVf}) (with $\ell=\infty$). This prolongation $j^\infty{\frak X}\in\Theta(\zp)$ is the sum of a horizontal vector field $A^jD_{x^j}\in\cC\Theta(\zp)$ and an evolutionary vector field $\zd_{\frak X}\in\op{E}\!\Theta(\zp)$.\medskip

In conformity with the symmetry conditions (\ref{ClasSymPhys}) and (\ref{HighSymPhys}), we say that the {\bf generalized vector field} ${\frak X}\in\op{Der}(\cF_0,\cF)$ is a {\bf symmetry of the Euler-Lagrange equations} $\zd_{u^a}{\cal L}|_{j^k\zf}=0\,,\forall a\,$, if \be\label{HighSymPhysEL}\zd_{\frak X}(\zd_{u^a}{\cal L})\approx 0\;,\forall a\;.\ee As said before, the requirement means that the infinitesimal transformation induced by ${\frak X}$ transforms a solution into a solution up to terms of order $\ge 2$ in the infinitesimal parameter.\medskip

As for the concept of symmetry of the action, remember first a well-known fact of Lagrangian Mechanics. In Electromagnetism, the gauge the transformation $$F'=F-\p_t\zy,\quad \vec A'=\vec A+\vec\nabla\zy\;$$ ($F$ and $\vec A$ are the scalar and vector potentials, $\zy$ is a function of time and positions, and $\vec\nabla$ is the gradient) modifies the generalized electromagnetic potential $U=e(F -\vec v\cdot\vec A)$ ($e$ is the charge and $\vec v$ the velocity of the considered particle) and thus leads to different Lagrangians ${\cal L}$ and ${\cal L}'$. However, it is easily seen that the latter differ by the total derivative ${\cal L}'-{\cal L}=d_t\,\jmath$ of a function $\jmath$ of time and positions, and that the Euler-Lagrange equations associated to ${\cal L}$ and ${\cal L}'$, hence, the dynamics, are therefore the same. This observation can be extended to the present field theoretic context. Two Lagrangians ${\cal L},{\cal L}'\in\widetilde\cF$ implement the same Euler-Lagrange equations if and only if they differ by a total divergence: $$\zd_{u^a}{\cal L}=\zd_{u^a}{\cal L}',\quad \forall a\quad \Leftrightarrow\quad
{\cal L}'-{\cal L}=D_{x^i}\jmath^i,\quad \jmath^i\in\widetilde\cF\;.$$ This indicates that two action functionals $S_{\cal L}$ and $S_{{\cal L}'}$, which are defined by Lagrangians ${\cal L}$ and ${\cal L}',$ coincide (on all compactly supported sections) if and only if the underlying Lagrangians ${\cal L},{\cal L}'$ differ by a total divergence. It is thus natural to identify the space of action functionals $S_{\cal L}$ with the space of classes $[\cal L]$ of functions ${\cal L}\in\widetilde\cF$ considered up to total divergence. Alternatively, an action can be viewed as a class $[{\cal L}\dd\!x]$, where $\dd\!x=\dd x^1\ldots \dd x^n$ and where $${\cal L}\dd\!x\simeq {\cal L}\dd\!x+D_{x^i}\jmath^i\dd\!x\;.$$

A {\bf symmetry of the action} is now a generalized vector field $\frak X$, such that $$\zd_{\frak X}[{\cal L}\dd\!x]=[0]\;.$$ This definition only makes sense, if we define how the prolongation $\zd_{\frak X}$ acts on the differential form $\dd\!x$ and show that its action on $[{\cal L}\dd\!x]$ is well-defined. We confine ourselves here to mentioning that the symmetry condition finally reads $$\zd_{\frak X}{\cal L}=D_{x^i}\jmath^i\;,$$ where $\jmath^i\in\cF$, i.e., just requires that $\zd_{\frak X}{\cal L}$ be a total divergence. Moreover, any symmetry of the action is a symmetry of the Euler-Lagrange equations (but the converse is not true).\medskip

Eventually, a {\bf gauge symmetry} is a symmetry \be\label{GenVfSec}{\frak X}(f)=A^j(x^i,u^a_\za)\p_{x^j}+B^b(x^i,u^a_\za)\p_{u^b}=A^j(\p_{x^j}+u^b_j\p_{u^b})+(B^b-A^ju^b_j)\p_{u^b}\;\ee of the action, whose coefficients $$A^j=A^j(f)=A^{j}_\za D_x^\za f\quad\text{and}\quad B^b=B^b(f)=B^{b}_\zb D_x^\zb f$$ are the values of some total differential operators on an arbitrary / a varying function $f\in\cF$.\medskip

Symmetries of the action (resp., symmetries of the action obtained as value of a gauge symmetry on a specific / a fixed function $f\in\cF$) are often termed as {\bf global symmetries} (resp., {\bf local symmetries}). Further, we call {\bf symmetry in characteristic form} a symmetry given by a vertical generalized vector field $${\cal X}=C^b(x^i,u^a_\za)\p_{u^b}\in\op{Der}^{{v}}(\cF_0,\cF)\;.$$ For all types of symmetry (symmetry of the Euler-Lagrange equations, symmetry of the action, or gauge symmetry), any symmetry $\frak X$ (see Equation (\ref{GenVfSec})) provides a symmetry $${\cal X}=(B^b-A^ju^b_j)\p_{u^b}$$ in characteristic form (note that $\cal X$ is a symmetry, since $\zd_{\frak X}=\zd_{\cal X}$).

\subsection{Noether's theorems}

Einstein qualified Noether's results as a monument of mathematical thinking. The tight relationship between symmetries and conserved quantities is part of each course in Classical Mechanics. More precisely, Noether's theorems claim that there exists a 1:1 correspondence between (equivalence classes of) symmetries of the action in characteristic form and (equivalence classes of) `conserved currents', and that there exists a 1:1 correspondence between gauge symmetries in characteristic form and Noether identities \cite{Noether}, \cite{Yvette}.\medskip

The latter correspondence is via formal adjoint operators. More precisely, if $N^a_\za D_x^\za\zd_{u^a}{\cal L}\equiv 0$ is a Noether identity, we consider the total differential operator $N$ with components $N^a=N^a_\za D_x^\za$, and define the corresponding gauge symmetry in characteristic form ${\cal X}(f)=C^a(f)\p_{u^a}$ as the adjoint $N^{+}$ of $N$, i.e., by $C^a(f)=N^{a+}(f)=(-D_x)^\za\lp N^a_\za f\rp$. The converse association is similar. It follows that non-trivial Noether identities correspond to non-trivial gauge symmetries in characteristic form.

\subsection{Compatibility complex, formal exactness, formal integrability}\label{CC}

\subsubsection{Compatibility complex and formal exactness}

An overdetermined system is a system of linear equations that are not independent, so that the existence of a solution is subject to {\it compatibility conditions}.\medskip

The simplest example of an {\bf overdetermined system} is a system of linear equations $LX=C$, where $L\in\op{gl}(p\times n,\R)$, $X\in\R^n$, and $C\in\R^p$, whose rank $\zr(L)\neq p$. This means that, between the ({\small LHS}-s of the) equations, i.e., between the rows $L_{i\,\star}$ of $L$, there do exist non-trivial linear relations. In the following, we assume for simplicity that there is exactly one such relation, $L_{p\,\star}=\sum_{j=1}^{p-1}\zl_jL_{j\,\star}$, with $\zl_j\in\R$. This existence of {\bf non-trivial linear relations between the equations} is equivalent to the existence of a non-zero linear operator, in the considered case, the {\bf non-zero linear operator $\zL=(\zl_1,\ldots,\zl_{p-1},-1)\in\op{gl}(1\times p,\R)$, such that $\zL\circ L=0$}. Hence, the existence of a solution $X$ requires that $C$ satisfies the {\it compatibility condition} $C\in\ker\zL$, i.e., $C_p=\sum_{j=1}^{p-1}\zl_jC_j$. In this case, the {\bf original system reduces} to $L'X=C'$, with self-explaining notation, and, in view of our assumption, we have $\zr(L')=p-1$. Of course, a {\bf homogeneous system} always reduces. The most general {\sl solution} then depends on $n-(p-1)\ge 0$ parameters, so that $C\in\op{im}L$ and the complex $$\R^n\stackrel{L}{\longrightarrow}\R^p\stackrel{\zL}{\longrightarrow}\R$$ is {\sl exact}.\medskip

Another basic example is integration in $\R^n$, which corresponds to the system of linear {\small PDE}-s $\dd_0 f=\zw$, where $\dd_0:\Ci(\R^n)\to\zW^1(\R^n)$ is the de Rham differential. The {\bf non-trivial linear partial differential relations \be\label{NO1}\p_{x^j}\p_{x^i}f-\p_{x^i}\p_{x^j}f=0\ee between the {\small PDE}-s} can be equivalently written {\bf as $\dd_1\dd_0=0$, where the non-zero linear partial differential operator $\dd_1$} is the de Rham operator on 1-forms: $$\Ci(\R^n)\stackrel{\dd_0}{\longrightarrow} \zW^1(\R^n)\stackrel{\dd_1}{\longrightarrow} \zW^2(\R^n)\;.$$ The existence of a solution implies that the {\it compatibility condition} $\,\zw\in\ker\dd_1$ holds. Since the complex is {\sl exact}, we then have $\zw\in\op{im}\dd_0$, i.e., the considered {\small PDE} admits a {\sl solution}.\medskip

More generally, let $D\in\op{Diff}(\zp,\zp')$ be a linear differential operator between smooth sections of vector bundles $\zp:E\to X$ and $\zp':E'\to X$ over a manifold $X$. The linear (homogeneous) {\small PDE} implemented by $D\simeq\psi_D$ is called {\bf overdetermined}, if there exists a non-zero linear differential operator $\zD\in\op{Diff}(\zp',\zp''),$ such that $$\zG(\zp)\stackrel{D}{\longrightarrow}\zG(\zp')\stackrel{\zD}{\longrightarrow}\zG(\zp'')\;$$ is a complex (of $\Ci(X)$-modules). We then say that $\zD$ is a {\bf compatibility operator} for $D$, if the pair $(\zD,\zp'')$ is universal in the obvious sense.\medskip

Just as the original operator $D$ can be overdetermined (non-trivial linear differential relations between the corresponding equations -- compatibility operator), a compatibility operator $\zD$ can itself be overdetermined (relations between the relations -- new compatibility operator). This then leads to a {\bf compatibility complex} of the original operator $D\;$:  $$\zG(\zp)\stackrel{D}{\longrightarrow}\zG(\zp')\stackrel{\zD_1}{\longrightarrow}\zG(\zp'')\stackrel{\zD_2}{\longrightarrow}\zG(\zp^{'''})\stackrel{\zD_3}{\longrightarrow}\ldots\;$$

In fact, any $D\in\op{Diff}_k(\zp,\zp')$ admits a compatibility complex in the abelian category ${\tt Mod}(\cO)$ of modules over $\cO=\Ci(X)$, but not necessarily in the non-abelian category ${\tt r\Ci VB}(X)$ of finite rank smooth vector bundles over $X$. Indeed, for any $k_1\in\N$, the algebraicized $k_1$-prolongation $\psi^{k_1}_D\in\op{Hom}_\cO(\zG(\zp_{k+k_1}),\zG(\zp'_{k_1}))$ of $D$ admits a cokernel $\psi\in\h_\cO(\zG(\zp'_{k_1}),\cP_2)$ in ${\tt Mod}(\cO)$, which represents a differential operator $\zD_1\in\op{Diff}_{k_1}(\zp',\cP_2)$. Since $\psi$ is the cokernel of $\psi^{k_1}_D$, the operator $\zD_1$ satisfies \be\label{Diag}\zD_1\circ D=\psi\circ j^{k_1}\circ D=\psi\circ\psi^{k_1}_D\circ j^{k+k_1}=0\;.\ee In fact $\zD_1$ is universal and is thus a compatibility operator of $D$. When turning the crank again and again, we obtain a compatibility complex of $D$: \be\label{CompaCompl}\zG(\zp)\stackrel{D}{\longrightarrow}\zG(\zp')\stackrel{\zD_1}{\longrightarrow}\cP_2\stackrel{\zD_2}{\longrightarrow}\cP_3\stackrel{\zD_3}{\longrightarrow}\ldots\;\ee Here we actually use the algebraic approach -- in the frame of $\cO$-modules -- to differential operators, see for instance \cite{KV}, \cite{GKuP}, \cite{GKoP}. However, the $\cO$-modules $\cP_2,\cP_3,\ldots$ are not necessarily projective of finite rank, i.e., they are not necessarily modules $\zG(\zp''),\zG(\zp'''),\ldots$ of sections of vector bundles of finite rank.\medskip

In the following, we stay within the setting of algebraic differential operators and consider a diagram of the type we just used to construct a compatibility operator (see Equations (\ref{Diag}), (\ref{ProlPDECoorFree}), (\ref{ProlPDE})): \be\label{Basic}\begin{array}{ccccccccc}\cdots & \longrightarrow &  \cP_{i-1}&\stackrel{\zD_{i-1}}{\longrightarrow}&  \cP_i & \stackrel{\zD_i}{\longrightarrow} &  \cP_{i+1} &\longrightarrow &\cdots\\[0.2cm]  &&j^{\scriptscriptstyle{k_{i-1}+k_i+\ell}}\left\downarrow\rule{0cm}{.4cm}\right.\phantom && j^{\scriptscriptstyle{k_i+\ell}}\left\downarrow\rule{0cm}{.4cm}\right.\phantom\downarrow &&j^{\scriptscriptstyle\ell} \left\downarrow\rule{0cm}{.4cm}\right.\phantom &&\\[0.2cm] \cdots& \longrightarrow  & \cJ^{\scriptscriptstyle{k_{i-1}+k_i+\ell}}(\cP_{i-1})&\stackrel{\psi^{\scriptscriptstyle{k_i+\ell}}_{\zD_{i-1}}}{\longrightarrow} & \cJ^{\scriptscriptstyle{k_i+\ell}}(\cP_i) & \stackrel{\psi^{\ell}_{\zD_i}}{\longrightarrow}  & \cJ^{\ell}(\cP_{i+1}) &\longrightarrow &\cdots\\[0.2cm]\end{array}\ee
Here $\cP_{i-1},\cP_i,\cP_{i+1}$ are $\cO$-modules, $\zD_{i-1}\in\op{Diff}_{k_{i-1}}(\cP_{i-1},\cP_i)$, $\zD_{i}\in\op{Diff}_{k_{i}}(\cP_{i},\cP_{i+1})$, $\ell\in\N$, and $\cJ^k(\cP)$ is the algebraic counterpart of $\zG(J^k(P))$, where $P\to X$ is a vector bundle and $J^k(P)$ is the ordinary $k$-jet bundle (`algebraic counterpart' means that, in the geometric case $\cP=\zG(P)$, we have $\cJ^k(\cP)=\zG(J^k(P))$).\medskip

The bottom row of (\ref{Basic}) is made of prolonged algebraicized operators, or, still, prolonged formal operators (acting on formal derivatives, i.e., on jet space coordinates). The study of formal operators is referred to as the {\bf formal theory}. Note that the word `formal' appears naturally here and refers to the algebraicized or jet space setting.\medskip

It is clear (see above) that one of the main questions in the context of compatibility complexes is exactness (exactness of the top row in (\ref{Basic})), i.e., `the question whether the considered equation admits a solution whenever the compatibility condition is satisfied'. The question of exactness can of course also be considered in the (simpler) formal theory (exactness of the bottom row).\medskip

More precisely, a compatibility complex (top row) is called {\bf formally exact}, if the corresponding formal complex (bottom row) is exact, for any $\ell\in\N$. In this case, the main task is to look for criteria for exactness of the original (top row) complex.\medskip

We will not investigate the latter problem. On the other hand, it is important to know that \cite{KV}, for any sufficiently large $k_1\in\N$, the compatibility complex (\ref{CompaCompl}) is formally exact, for any operator $D$. We actually have the

\begin{prop} Any linear differential operator $D\in\op{Diff}(\zp,\zp')$ admits a formally exact compatibility complex. The same is true for any horizontal linear differential operator $D\in\cC\op{Diff}(\zp^*_\infty(\zh),\zp^*_\infty(\zh'))$.\end{prop}

\subsubsection{Formal integrability}

We now briefly comment on formal integrability of a linear partial differential equation $\zS^0$ or linear differential operator $D$.\medskip

The first observation is that the category ${\tt r\Ci VB}(X)$ is not Abelian. Indeed, kernels, like e.g., $\zS^\ell=\ker\psi^\ell_D$, are not necessarily vector bundles over $X$. The reason is that, if $\psi:E\to E'$ is a map of vector bundles over $X$, the rank $\zr(\psi_m)$ of the linear map $\psi_m:E_m\to E'_m$ may vary with $m\in X$. Then, the kernel $\ker\psi:=\coprod_{m\in X}\ker\psi_m$ is a bundle of vector spaces of varying dimension $\op{rk}(E)-\zr(\psi_m)$. However, if the rank $\zr(\psi)$ is constant, it is easily seen that the kernel $\ker\psi$ is a vector bundle over $X$. Therefore, it is natural to ask that $D\simeq\psi_D$ be {\bf regular}, i.e., that the rank $\zr(\psi^\ell_D)$ be constant, for any $\ell\in\N$, or, still, that $\zS^\ell=\ker\psi^\ell_D$ be a vector bundle over $X$, for any $\ell\in\N$.\medskip

The second remark is that, if $D$ is of order $k$, the prolongation $\zS^\ell$ is the kernel in $J^{k+\ell}(E)$ of the differential consequences $\psi^\ell_D$ up to order $\ell$ of the equation $\psi_D=0$. It follows that any solution in $J^{k+\ell+1}(E)$ of the system $\psi^{\ell+1}_D=0$ (differential consequences up to order $\ell+1$) projects by $\zp_{k+\ell,k+\ell+1}$ to a solution in $J^{k+\ell}(E)$ of the system $\psi^\ell_D=0$ (differential consequences up to order $\ell$): $$\zp_{k+\ell,k+\ell+1}\zS^{\ell+1}\subset \zS^\ell\;.$$ On the other hand, any family $j^{k+\ell}_m\zf$ ($m\in X$) of solutions of $\psi^\ell_D=0$ can be extended to a family $j^{k+\ell+1}_m\zf$ ($m\in X$) of solutions of $\psi^{\ell+1}_D=0$. Of course, the best situation is when any solution of $\psi^\ell_D=0$ can be extended to a solution of $\psi^{\ell+1}_D=0$, i.e., when $$\zp_{k+\ell,k+\ell+1}\zS^{\ell+1}=\zS^\ell\;.$$ This shows that the {\bf existence of extended formal solutions}, i.e., formal integrability, is a simplifying requirement.\medskip

Actually we say that a linear differential operator $D\simeq\psi_D$ is {\bf formally integrable}, if it is regular and if extended formal solutions do exist, i.e., more precisely, if $\zS^\ell$ is a vector bundle, for all $\ell\in\N$, and the vector bundle map $\zp_{k+\ell,k+\ell+1}:\zS^{\ell+1}\to\zS^\ell$ is surjective, for all $\ell\in\N$. In the present text, all partial differential equations $\zS^0$, even those that are not implemented by a differential operator, are assumed to be {\bf formally integrable in the sense of Remark \ref{FromIntSol}} \cite{KV}.

\section{Appendix B: Partial differential equations in algebraic $\cD$-geometry}\label{Jet}

\begin{rem}{\em This section should be read together with Section \ref{KTRHT}, where notation and motivation are explained.}\end{rem}

\subsection{A proof of Proposition \ref{VBFunAlg}}\label{ProofVBFunAlg}

Proposition \ref{VBFunAlg}, which states roughly speaking that the function algebra of the total space of a vector bundle can be viewed as an algebra over the function algebra of the base, is almost obvious. We nevertheless check the details carefully.\medskip

Let $\zp:E\to X$ be an affine morphism of schemes (i.e., a locally ringed space morphism $(\zp,\zp^\sharp): (E,\cO_E)\to (X,\cO_X)$ such that there is an affine cover of $X$ whose preimages by $\zp$ are affine), in particular a vector bundle. In the following, we consider the sheaf $\cO_E\in{\tt Sh}(E)$ as sheaf $\cO^E_X:=\zp_*\cO_E\in{\tt Sh}(X)$, where $\zp_*$ denotes the direct image of sheaves. It is known \cite{Har} that $\zp_*$ induces an equivalence of the categories ${\tt qcMod}(\cO_E)$ and ${\tt qcMod}(\cO_X)\cap\,{\tt Mod}(\cO^E_X)\,,$ with self-explaining notation. It follows that $\cO^E_X\in{\tt qcMod}(\cO_X)$. Moreover, $\cO^E_X$ is clearly a (sheaf of) commutative unital ring(s). To see that $\cO^E_X\in{\tt qcCAlg}(\cO_X)$, recall first that such an algebra is a commutative monoid in ${\tt qcMod}(\cO_X)$, i.e., that it is an object in ${\tt qcMod}(\cO_X)$ that carries an associative commutative unital multiplication, which is a morphism in ${\tt qcMod}(\cO_X)$ (and similarly for the unit). It suffices to examine the $\cO_X$-linearity of the multiplication (and of the unit -- what is also simple). Start with noticing that, for any open $V\subset X$, $f\in\cO_X(V)$ and $F\in\cO^E_X(V)=\cO_E(\zp^{-1}(V))$, the action of $f$ on $F$ is defined via the ring morphism $\zp^\sharp:\cO_X(V)\to \cO_E(\zp^{-1}(V))$ by $$f\cdot F:=\zp^\sharp(f)\star F\;,$$ where $\star$ is the ring multiplication. Hence, the multiplication $\star$ is $\cO_X(V)$-bilinear, i.e., $$\star:\cO^E_X(V)\0_{\cO_X(V)}\cO^E_X(V)\to\cO^E_X(V)$$ is $\cO_X(V)$-linear, and this presheaf morphism induces a sheaf morphism $\star:\cO^E_X\0_{\cO_X}\cO^E_X\to\cO^E_X\,$ in $\cO_X$-modules.

\subsection{Jet functor}\label{ConsJetFunc}

We give some information about the construction of the jet functor $${\cal J}^\infty:{\tt qcCAlg}({\cO}_X)\to {\tt qcCAlg}({\cD}_X):\op{For}\;$$ as left adjoint of the forgetful functor. We assume that the smooth scheme $X$ is a smooth affine algebraic variety, so that we can substitute global sections to sheaves and thus avoid sheaf-theoretic subtleties -- but the same proof goes through in the general case. We denote by $\cO$ (resp., $\cD$) the algebra $\cO_X(X)$ (resp., $\cD_X(X)$).\medskip

The functor ${\cal J}^\infty$ must be left adjoint to the forgetful functor $\op{For}$, i.e., for $B\in\cO{\tt A}:={\tt CAlg}(\cO)$ and $A\in{\cD}{\tt A}:={\tt CAlg}(\cD)$, we must have \be\label{JetAlg}\h_{ \cD{\tt A}}({\cal J}^{\infty}B,A)\simeq \h_{\cO{\tt A}}(B,\op{For}A)\;,\ee functorially in $A,B$. The construction of ${\cal J}^{\infty}B$ is quite natural. We start from the $\cD$-module $\cD\0_{\cO} B$ (in the tensor product we consider $\cD$ as endowed with its right $\cO$-module structure), and consider the $\cD$-algebra $\cS_{\cO}(\cD\0_{\cO}B)$ over $\cD\0_{\cO}B$ ($\cS$ is the symmetric tensor algebra functor). Since Equation (\ref{JetAlg}) suggests the existence of an $\cO$-algebra morphism $B\to {\cal J}^{\infty}B$, we define ${\cal J}^{\infty}B$ as the quotient of the $\cD$-algebra $\cS_{\cO}(\cD\0_{\cO}B)$ by a $\cD$-ideal such that the natural inclusion $$i:B\ni b\mapsto 1\0 b\in\cS_{\cO}(\cD\0_{\cO}B)$$ becomes an $\cO$-algebra morphism $\zP\circ i:B\to {\cal J}^{\infty}B$ when composed with the natural projection $\zP$. Since an $\cO$-algebra morphism is an $\cO$-linear map (a condition that is automatically satisfied here) that respects the multiplications and the units, we must ensure that $$\zP(1\0(bb'))=\zP(1\0 b)\odot\zP(1\0 b')=\zP((1\0 b)\odot (1\0 b'))\quad\text{and}\quad \zP(1\0 1_B)=\zP(1)\;,$$ where $1$ (resp., $1_B$) denotes the unit in $\cO$ (resp., $B$) and where $\odot$ is the symmetric tensor product (we denote the product of two residue classes by the same symbol). Hence, we consider the $\cD$-ideal $K$ generated by the elements $$D\cdot\left((1\0 b)\odot (1\0 b')-1\0(bb')\right)\in\cS_{\cO}(\cD\0_{\cO}B)\quad\text{and}\quad D\cdot\left(1\0 1_B-1\right)\in\cS_{\cO}(\cD\0_{\cO}B)\;,$$ where $D\,\cdot$ denotes the action of an arbitrary differential operator $D\in\cD$.

It now suffices to show that $${\cal J}^{\infty}: \cO{\tt A}\ni B\mapsto {\cal J}^{\infty}B:=\cS_{\cO}(\cD\0_{\cO}B)/K\in\cD{\tt A}$$ possesses the adjointness property (\ref{JetAlg}).

If $f:{\cal J}^{\infty}B\to A$ is a $\cD$-algebra morphism, the map $$\tilde{f}:B\ni b\mapsto f(\zP(1\0 b))\in \op{For}A$$ is obviously an $\cO$-algebra morphism.

Conversely, let $g:B\to \op{For}A$ be an $\cO$-algebra morphism. The map $$\bar g:\cD\0_{\cO}B\ni D\0 b\mapsto D\cdot(g(b))\in A\;$$ is a well-defined $\cD$-module morphism. Since $\cS_\cO(\cD\0_\cO B)$ is the free $\cD$-algebra over the $\cD$-module $\cD\0_\cO B$, the $\cD$-module morphism $\bar g$ can be uniquely extended to a $\cD$-algebra morphism $\bar g:\cS_{\cO}(\cD\0_{\cO}B)\to A$. As $\bar g$ vanishes on $K$ (note that $\bar g(1)=1_A$, where $1_A$ is the unit in $A$), it descends to the quotient ${\cal J}^{\infty}B$. Hence, the searched $\cD$-algebra morphism $\bar g:{\cal J}^{\infty}B\to A$.\medskip

Consider the example of a trivial line bundle $\zp:E=\R^2\ni(t,x)\mapsto t\in X=\R$ and set $\cO=\cO_X(X):=\R[t]$ and $B:=\cO_X^E(X)=\cO_E(E):=\R[t,x]\in\tt\cO A$. It is easily seen that the symmetric algebra $\cS_{\cO}(\cD\0_{\cO}B)$ coincides with the polynomial algebra $\R[t,\p_t^i\0 x^j]$, where $i,j\in\N$. When dividing the ideal $K$ out, we obtain $$\cJ^\infty(B)=\R[t,x,\p_t\0 x,\p_t^2\0 x,\ldots]\in\tt\cD A\;.$$ Indeed, the initial generator $\p_t\0 x^2$ (resp., $\p_t\0 1_B$), for instance, coincides in the quotient with $$\p_t\0 x^2=\p_t\cdot ((1\0 x)\odot (1\0 x))\quad (\text{resp.,}\;\,\p_t\0 1_B=\p_t\cdot 1)\;.$$ This generator is thus a polynomial in $\p_t\0 x$ and $1\0 x\simeq x$ (resp., is thus equal to 0, since $\p_t$ acts on the element $1$ of the $\cD$-module $\cO$) and can therefore be omitted in the quotient. Hence, the announced result. When setting $x^{(k)}:=\p_t^k\0 x$, we get $$\cJ^\infty(B)=\R[t,x,x^{(1)},x^{(2)},\ldots]\in\tt\cD A\;,$$ i.e., we obtain indeed the polynomial function algebra of the infinite jet space of $\zp$.\medskip

Observe also that by definition $\p_t\cdot x^{(k)}=x^{(k+1)}$, i.e., that \be\label{ActBaseNatActHor} \p_t\cdot x^{(k)}=(\p_t + x^{(1)}\p_x + x^{(2)}\p_{x^{(1)}}+\ldots)\,x^{(k)}=D_t\,x^{(k)}\;,\ee where $D_t$ is the total derivative. Since the vector field $\p_t\in\cD$ acts on a function in $\cJ^\infty(B)$ as derivation, the action of a differential operator of the base on a function in $\cJ^\infty(B)$ coincides with the natural action of the corresponding total differential operator.

\subsection{Differential graded algebras over differential operators with coefficients in a $\cD$-algebra}\label{CoefDO}

Let $X$ be a smooth scheme and let ${\cal A}\in{\tt qcCAlg}(\cD_X)$ with multiplication $\star$ (let us recall that $\cD_X$ is generated by the sheaf $\cO_X$ of functions and the sheaf $\Theta_X$ of vector fields). We denote the action on $a\in {\cal A}$ by $f\in\cO_X$ (resp., $\zy\in\Theta_X$) by $f\cdot a$ (resp., $\nabla_\zy\, a$).\medskip

The ${\tt qcCAlg}(\cD_X)$-morphism $\zvf:\cO_X\to \cA$, which is defined by $$\zvf(f) = \zvf(f\cdot 1_{\cO_X}) = f\cdot \zvf(1_{\cO_X}) = f\cdot 1_\cA\;,$$ is injective, since it is the composition of the injective ${\tt qcCAlg}(\cD_X)$-morphism $\cO_X\ni f\mapsto f\0 1_\cA\in\cO_X\0_{\cO_X}\cA$ and the bijective ${\tt qcCAlg}(\cD_X)$-morphism $\cO_X\0_{\cO_X}\cA\ni f\0 a\mapsto f\cdot a\in\cA\,$. Hence, an element $f\in\cO_X$ is viewed as an element in ${\cal A}$ via the identification $f\simeq f\cdot1_{\cal A}\,$, and \be\label{Uniq1}f\cdot a=f\cdot(1_{\cal A}\star a)=(f\cdot 1_{\cal A})\star a\simeq f\star a\;.\ee

The ring ${\cal A}[\cD_X]$ of {\bf differential operators on $X$ with coefficients in ${\cal A}$} is the $\cD_X$-module $${\cal A}[\cD_X]:={\cal A}\0_{\cO_X}\cD_X\;,$$ endowed with the associative unital $\R$-algebra structure $\circ$ defined, for $a,a'\in {\cal A}$, $\zy\in\Theta_X$, and $D\in\cD_X$, by \be\label{Mult1} (a\0 1_\cO)\circ (a'\0 D)=(a\star a')\0 D\ee and \be\label{Mult2} (1_{\cal A}\0 \zy)\circ(a'\0 D)=(\nabla_\zy\,a')\0 D+a'\0 (\zy\circ D)\;.\ee This multiplication is canonically extended to a first factor of the type $$a\0 (f\circ \zy\circ \zy')=((a\star f)\0 1_\cO)\circ (1_{\cal A}\0 \zy)\circ (1_{\cal A}\0 \zy')\;.$$ It is straightforwardly checked that the usual relations like, e.g., $\zy\circ\zy'=\zy'\circ\zy +[\zy,\zy']$, do not lead to any contradiction. Moreover, the embedding $${\cal A}\ni a\mapsto a\0 1_\cO\in {\cal A}[\cD_X]$$ is an associative unital algebra morphism (i.e., ${\cal A}$ is a subalgebra of ${\cal A}[\cD_X]$), whereas the embedding $$\Theta_X\ni\zy\mapsto 1_{\cal A}\0\zy\in {\cal A}[\cD_X]$$ is a Lie algebra morphism (i.e., $\Theta_X$ is a Lie subalgebra of ${\cal A}[\cD_X]$). These inclusions extend to an associative unital algebra morphism $$\cD_X\ni D\mapsto 1_{\cal A}\0 D\in {\cal A}[\cD_X]\;.$$

Let us now focus on the category ${\tt DG_+qcCAlg}(\cA[\cD_X])$ of differential non-negatively graded $\cO_X$-quasi-coherent commutative unital $\cA[\cD_X]$-algebras. As already mentioned in Equation (\ref{DGADAX}):

\begin{defi}\label{DGJDADefi} A {\bf differential non-negatively graded $\cA[\cD_X]$-algebra} is an object of the category ${\tt CMon(DG_+qcMod}(\cA[\cD_X]))$ of commutative monoids in the category of differential non-negatively graded $\cO_X$-quasi-coherent $\cA[\cD_X]$-modules, i.e., it is a differential graded commutative $\cA$-algebra, as well as a differential graded $\cD_X$-module ${\frak A}_\bullet\in {\tt DG_+qcMod}(\cD_X)$, such that vector fields act as derivations on the $\cA$-action on ${\frak A}_\bullet$ and on the multiplication of $\frak A_\bullet\,.$ A morphism of differential graded $\cA[\cD_X]$-algebras is a morphism of differential graded $\cD_X$-modules that is $\cA$-linear and respects the multiplications and the units. The category of differential graded $\cA[\cD_X]$-algebras and morphisms between them will be denoted by ${\tt DG_+qcCAlg}(\cA[\cD_X])\,.$ \end{defi}

In other words, a differential graded $\cA[\cD_X]$-algebra is a differential graded $\cA$-algebra, as well as a differential graded $\cD_X$-algebra, such that the $\cA$-action and the $\cD_X$-action are compatible in the sense that vector fields $\Theta_X\subset \cD_X$ act on the $\cA$-action $\triangleleft$ as derivations.

\begin{ex}\label{DGJDAEx}{\em Let $\cA$ be, as above, a $\cD_X$-algebra. Any differential graded $\cD_X$-algebra morphism $f:\cA\to {\cB}_\bullet$ allows to endow ${\cB}_\bullet$ with a differential graded $\cA[\cD_X]$-algebra structure, i.e., to view ${\cB}_\bullet$ as an object ${\cB}_\bullet\in {\tt DG_+qcCAlg}(\cA[\cD_X])$. Indeed, it suffices to set $$a\triangleleft b:=f(a)\star_{\cB} b\;,$$ with self-explaining notation. Verifications are straightforward. In particular, $\cA$ can be interpreted as differential graded $\cA[\cD_X]$-algebra with $\cA$-action $\triangleleft$ given by the $\cA$-multiplication $\star_\cA\,$. }\end{ex}

\subsection{Construction of non-split relative Sullivan $\cD$-algebras}\label{RSDA}

For convenience, we recall Lemma 22 of \cite{BPP4}, which is needed in the main part of this text.

\begin{lem}\label{LemRSA} Let $(T,d_T)\in\tt DG\cD A$, let $(g_j)_{j\in J}$ be a family of symbols of degree $n_j\in \N$, and let $V=\bigoplus_{j\in J}\cD\cdot g_j$ be the free non-negatively graded $\cD$-module with homogeneous basis $(g_j)_{j\in J}$.\smallskip

(i) To endow the graded $\cD$-algebra $T\0\cS V$ with a differential graded $\cD$-algebra structure $d$, it suffices to define \be\label{CondRSADiff}d g_j\in T_{n_j-1}\cap d_T^{-1}\{0\}\;,\ee to extend $d$ as $\cD$-linear map to $V$, and to equip $T\0\cS V$ with the differential $d$ given, for any $t\in T_p,\,v_1\in V_{n_1},\,\ldots,\,v_k\in V_{n_k}\,$, by \be\label{DefRSADiff}d({t}\0 v_1\odot\ldots\odot v_k)=$$ $$d_T({t})\0 v_1\odot\ldots\odot v_k+(-1)^p\sum_{\ell=1}^k(-1)^{n_\ell\sum_{j<\ell}n_j}({t}\ast d(v_\ell))\0v_1\odot\ldots\widehat{\ell}\ldots\odot v_k\;,\ee where $\ast$ is the multiplication in $T$. If $J$ is a well-ordered set, the natural map $$(T,d_T)\ni {t}\mapsto {t}\0 1_\cO\in (T\0\cS V,d)$$ is a relative Sullivan $\cD$-algebra.\smallskip

(ii) Moreover, if $(B,d_B)\in{\tt DG\cD A}$ and $p\in{\tt DG\cD A}(T,B)$, it suffices -- to define a morphism $q\in{\tt DG\cD A}(T\0\cS V,B)$ (where the differential graded $\cD$-algebra $(T\0\cS V,d)$ is constructed as described in (i)) -- to define \be\label{CondRSAMorph}q(g_j)\in B_{n_j}\cap d_B^{-1}\{p\,d(g_j)\}\;,\ee to extend $q$ as $\cD$-linear map to $V$, and to define $q$ on $T\0\cS V$ by \be\label{DefRSAMorph}q({t}\0 v_1\odot\ldots\odot v_k)=p({t})\star q(v_1)\star\ldots\star q(v_k)\;,\ee where $\star$ denotes the multiplication in $B$.\end{lem}

Lemma \ref{LemRSA} is natural. Indeed, the differential $d$ (Equation (\ref{DefRSADiff})) is the unique differential that restricts to $d_T$ on $T$, maps $V$ to $T,$ and provides a differential graded $\cD$-algebra structure on the graded $\cD$-algebra $T\0\cS V$. Similarly, the morphism $q$ (Equation (\ref{DefRSAMorph})) is the unique $\tt DG\cD A$-morphism $q:(T\0\cS V,d)\to (B,d_B)$ that restricts to $p:(T,d_T)\to (B,d_B)$ on $T$. Since Lemma \ref{LemRSA} allows to build relative Sullivan $\cD$-algebras, a similar construction might exist in Rational Homotopy Theory (in any case, we found this canonical construction independently).

\section{Acknowledgments}

The authors are grateful to Jim Stasheff for having read the first version of the present paper. His comments and suggestions allowed to significantly improve their text. They would also like to thank Gennaro di Brino for a discussion on an algebraic geometric issue. Further, the authors acknowledge the systematic use of online encyclopedias.

\vskip1cm
\noindent Damjan Pi\v{s}talo\\Email: damjan.pistalo@uni.lu\\ Norbert Poncin\\ Email: norbert.poncin@uni.lu \\\\University of Luxembourg\\Mathematics Research Unit\\Maison du nombre\\
6, Avenue de la Fonte\\
4364 Esch-sur-Alzette\\
Luxembourg


\begin{thebibliography}{Dillo 83}

\bibitem[BBH00]{BBH} G. Barnich, F. Brandt, M. Henneaux, {\em Local BRST cohomology in gauge theories}, Physics Reports, {\bf 338 (5)}, 2000, 439-569.

\bibitem[Bar10]{Bar} G. Barnich, {\em Global and gauge symmetries in classical field theories}, Lectures given at the Seminar `Algebraic Topology, Geometry, and Physics' of the University of Luxembourg, {\tt homepages.ulb.ac.be/$\sim$gbarnich}.

\bibitem[BD04]{BD} A. Beilinson and V. Drinfeld, {\em Chiral algebras}, American Mathematical Society Colloquium Publications, {\bf 51}, American Mathematical Society, Providence, RI, 2004.

\bibitem[BPP15a]{BPP1} G. di Brino, D. Pi\v{s}talo, and N. Poncin, {\em Model structure on differential graded commutative algebras over the ring of differential operators}, {\tt ArXiv:1505.07720}.

\bibitem[BPP15b]{BPP2} G. di Brino, D. Pi\v{s}talo, and N. Poncin, {\em Model categorical Koszul-Tate resolution for algebras over differential operators}, {\tt ArXiv:1505.07964}.

\bibitem[BPP17a]{BPP3} G. di Brino, D. Pi\v{s}talo, and N. Poncin, {\em Homotopical Algebraic Context over Differential Operators}, {\tt ArXiv:1706.05922}.

\bibitem[BPP17b]{BPP4} G. di Brino, D. Pi\v{s}talo, and N. Poncin, {\em Koszul-Tate resolutions as cofibrant replacements of algebras over differential operators}, to appear in OrbiLu.

\bibitem[CE48]{CE} C. Chevalley and S. Eilenberg, {\em Cohomology Theory of Lie Groups and Lie Algebras}, Trans. Amer. Math. Soc., {\bf 63 (1)}, 1948, 85-124.

\bibitem[Cos11]{Cos} K. Costello, {\em Renormalization and Effective Field Theory}, Mathematical Surveys and Monographs Volume, {\bf 170}, American Mathematical Society, 2011.




\bibitem[Gil06]{Gil06} J. Gillespie, {\em The flat model structure on complexes of sheaves}, Trans. Amer. Math. Soc., {\bf 358(7)}, 2006, 2855-2874.

\bibitem[GS06]{GS} P. G. Goerss, K. Schemmerhorn, \emph{Model Categories and Simplicial Methods}, {\tt ArXiv:math/0609537}.

\bibitem[GKP13b]{GKuP} J. Grabowski, D. Khudaverdyan, and N. Poncin, {\em The Supergeometry of Loday Algebroids}, J. Geo. Mech., {\bf 5(2)}, 2013, 185-213.

\bibitem[GKP13a]{GKoP} J. Grabowski, A. Kotov, and N. Poncin, {\em Lie superalgebras of differential operators}, J. Lie Theo., {\bf 23(1)}, 2013, 35-54.

\bibitem[Hal83]{Halperin} S. Halperin, {\em Lectures on minimal models}, M\'emoires de la SMF, 2e s\'erie {\bf 9-10}, 1983.

\bibitem[Har97]{Har} R. Hartshorne, {\em Algebraic Geometry}, Graduate Texts in Mathematics, {\bf 52}, Springer, 1997.


\bibitem[HT92]{HT} M. Henneaux and C. Teitelboim, {\em Quantization of Gauge Systems}, Princeton University Press, 1992.


\bibitem[Hir05]{Hir2} P. Hirschhorn, {\em Overcategories and undercategories of model categories}, {\tt http://www-math.mit.edu/~psh/undercat.pdf}, 2005.

\bibitem[HTT08]{HTT} R. Hotta, K. Takeuchi, and T. Tanisaki, {\em $\cD$-Modules, Perverse Sheaves, and Representation Theory}, Progress in Mathematics, {\bf 236}, Birkh\"auser, 2008.

\bibitem[Hov07]{Hov} M. Hovey, \emph{Model Categories}, American Mathematical Society, 2007.

\bibitem[KS90]{KS} M. Kashiwara, P. Schapira, {\em Sheaves on Manifolds}, Springer Science and Business Media, 1990.

\bibitem[Kos11]{Yvette} Y. Kosmann-Schwarzbach, {\em The Noether Theorems}, Invariance and Conservation Laws in the Twentieth Century, Sources and Studies in the History of Mathematics and Physical Sciences, Springer, ISBN 978-0-387-87868-3, 2011.

\bibitem[KV98]{KV} I.S. Krasil'shchik and A.M. Verbovetsky, {\em Homological methods in equations of mathematical physics}, Open Education, Opava, 1998, and Diffiety Inst. Preprint Series, {\bf 7}, 1998.

\bibitem[KV10]{KV2} I.S. Krasil'shchik and A.M. Verbovetsky, {\em Geometry of Jets and Integrable Systems}, {\tt ArXiv:1002.0077v6}.

\bibitem[Kru73]{Kru} D. Krupka, {\em Some geometric aspects of variational problems in fibred manifolds}, Folia Fac. Sci. Nat. Univ. Purk. Brunensis, Physica, {\bf 14}, Brno, 1973.



\bibitem[Noe18]{Noether} E. Noether, {\em Invariante Variationsprobleme}, G\"ottinger Nachrichten, 1918.

\bibitem[Qui69]{Qui} D. Quillen, {\em Rational homotopy theory}, Ann. Math., {\bf 90}, 1969, 205-295.

\bibitem[Sta97]{Jim Stasheff} J. Stasheff, {\em  The (secret?) homological algebra of the Batalin-Vilkovisky approach}, Proceedings of the Conference `Secondary Calculus and Cohomological Physics', 1997, and {\tt ArXiv:hep-th/9712157}.






\bibitem[Sul77]{Sul} D. Sullivan, {\em Infinitesimal computations in topology}, Publ. IHES, {\bf 47}, 1977, 269-331.


\bibitem[Tat57]{Tate} J. Tate, {\em Homology of Noetherian rings and local rings}, Illinois J. Math., {\bf 1(1)}, 1957, 14-27.

\bibitem[TV08]{TV08} B. To\"en, G. Vezzosi, {\em Homotopical Algebraic Geometry II: geometric stacks and applications}, Mem. Amer. Math. Soc., {\bf 193 (902)}, 2008.

\bibitem[Ver02]{Ver} A. Verbovetsky, {\em Remarks on two approaches to the horizontal cohomology: compatibility complex and the Koszul-Tate resolution}, Acta Appl. Math., {\bf 72 (1)}, 2002, 123-131.

\bibitem[Vin01]{Sascha} A. M. Vinogradov, {\em Cohomological analysis of partial differential equations and secondary calculus}, Translations of Mathematical Monographs, {\bf 204}, 2001.

\bibitem[Vit11]{Vit} L. Vitagliano, {\em Hamilton-Jacobi Diffieties}, {\tt ArXiv:1104.0162v2}.



\end{thebibliography}
\end{document}